\documentclass[english,11pt]{article}
\usepackage[T1]{fontenc}
\usepackage[latin1]{inputenc}
\usepackage{geometry}
\usepackage{float}
\usepackage{mathrsfs}
\usepackage{amsmath}
\usepackage{amssymb}
\usepackage{graphicx}
\usepackage[T1]{fontenc}
\usepackage{scalerel,stackengine}
\stackMath
\newcommand\reallywidehat[1]{%
\widehat{#1}
}

\usepackage{verbatim}
\usepackage{float}
\usepackage{amsmath}
\usepackage{amssymb}

\makeatletter

\floatstyle{ruled}

\usepackage{hyperref}
\hypersetup{
    colorlinks=true,
    linkcolor=blue,
    filecolor=magenta,      
    urlcolor=cyan,
    citecolor = red,
}

\makeatletter

\floatstyle{ruled}
\newfloat{algorithm}{tbp}{loa}
\providecommand{\algorithmname}{Algorithm}
\floatname{algorithm}{\protect\algorithmname}

\usepackage{amsthm}

\usepackage{mathrsfs}

\usepackage{amsfonts}

\usepackage{epsfig}

\usepackage{bm}

\usepackage{mathrsfs}

\usepackage{enumerate}

\@ifundefined{definecolor}{\@ifundefined{definecolor}
 {\@ifundefined{definecolor}
 {\usepackage{color}}{}
}{}
}{}

\usepackage{subfig}\usepackage[all]{xy}

\newcommand{\bbR}{\mathbb R}

\newtheorem{theorem}{Theorem}[section]
\newtheorem{lem}{Lemma}[section]
\newtheorem{rem}{Remark}[section]
\newtheorem{prop}{Proposition}[section]
\newtheorem{cor}{Corollary}[section]
\newcounter{hypA}
\newenvironment{hypA}{\refstepcounter{hypA}\begin{itemize}
  \item[({\bf A\arabic{hypA}})]}{\end{itemize}}
\newcounter{hypB}

\newcounter{hypD}

\newcounter{hypH}
\newenvironment{hypH}{\refstepcounter{hypH}\begin{itemize}
  \item[({\bf H\arabic{hypH}})]}{\end{itemize}}

\usepackage{babel}\date{}

\usepackage{babel}

\makeatother

\usepackage{babel}

\newcommand{\calN}{\mathcal{N}}

\newcommand{\cU}{\mathcal{U}}

\newcommand{\bbE}{\mathbb{E}}

\newcommand{\bbP}{\mathbb{P}}

\newcommand{\bbZ}{\mathbb{Z}}

\begin{document}

\begin{center}

{\Large \textbf{On Unbiased Estimation for Discretized Models}}

\vspace{0.5cm}

BY JEREMY HENG$^{1}$, AJAY JASRA$^{2}$, KODY J.~H. LAW$^{3}$ \& ALEXANDER TARAKANOV$^{3}$

{\footnotesize $^{1}$ESSEC Business School, Singapore, 139408, SG.}
{\footnotesize E-Mail:\,} \texttt{\emph{\footnotesize heng@essec.edu}}\\
{\footnotesize $^{2}$Computer, Electrical and Mathematical Sciences and Engineering Division, King Abdullah University of Science and Technology, Thuwal, 23955, KSA.}
{\footnotesize E-Mail:\,} \texttt{\emph{\footnotesize ajay.jasra@kaust.edu.sa}}\\
{\footnotesize $^{3}$School of Mathematics,
University of Manchester, Manchester, M13 9PL, UK.}
{\footnotesize E-Mail:\,} \texttt{\emph{\footnotesize kodylaw@gmail.com; tarakanov517@gmail.com}}

\begin{abstract}
In this article, we consider computing expectations w.r.t.~probability measures which are subject to discretization error.
Examples include partially observed diffusion processes or inverse problems, where one may have to
discretize time and/or space, in order to practically work with the probability of interest. Given access only to these
discretizations, we consider 
the construction of unbiased Monte Carlo estimators of expectations w.r.t.~such target probability distributions. 
It is shown how to obtain
such estimators using a novel adaptation of randomization schemes and Markov simulation methods.
Under appropriate assumptions, these estimators 
possess finite variance and finite expected cost. 
There are two important consequences of this approach: 
(i) unbiased inference is achieved at the canonical complexity rate, 
and (ii) the resulting estimators can be generated 
independently, thereby allowing strong scaling to arbitrarily many parallel processors.
Several algorithms are presented, and applied to some examples of
Bayesian inference problems, with both simulated and real observed data. \\

\end{abstract}

\end{center}

\section{Introduction}

Consider a probability measure $\pi$ on measurable space $(\mathsf{X},\mathcal{X})$ for which one wants to compute $\pi(\varphi):=\int_{\mathsf{X}}\varphi(x)\pi(dx)$
with $\varphi:\mathsf{X}\rightarrow\mathbb{R}$, $\pi-$integrable. Suppose one can only deal with a sequence of biased probability measures
$(\pi_l)_{l\in\mathbb{Z}^+}$ on $(\mathsf{X},\mathcal{X})$, with $\pi_l(|\varphi|)<+\infty~\forall l\in\mathbb{Z}^+$,  such that $\lim_{l\rightarrow\infty}\pi_l(\varphi)=\pi(\varphi)$ and 
$|\pi_{l+1}(\varphi)-\pi(\varphi)| \leq |\pi_{l}(\varphi)-\pi(\varphi)|$; 
examples include partially observed diffusion processes e.g.~\cite{mlpf} or inverse problems e.g.~\cite{beskos}. These latter models have a wide range of real applications such as engineering, finance and applied mathematics; see for instance \cite{stuart}.

In many applications of interest, one often resorts to constructing a $\pi_l-$invariant and ergodic Markov chain Monte Carlo (MCMC) kernel $K_l$ to estimate the expectation 
$\pi_l(\varphi)=\int_{\mathsf{X}}\varphi(x)\pi_l(dx)$. 
It is often the case that as $l$ grows, the cost of applying $K_l$ will also increase, often exponentially in $l$. Therefore one would often fix $l$ to achieve a given bias, and run the Markov chain for long enough to obtain a pre-specified variance which balances the bias.
In this article, we consider the task of producing unbiased estimators with finite variance. In particular,
using a stochastic simulation scheme based upon a family of Markov kernels $(K_l)_{l\in\mathbb{Z}^+}$, one can construct an estimator $\widehat{\pi(\varphi)}$ such that $\mathbb{E}[\widehat{\pi(\varphi)}]=\pi(\varphi)$ and $\mathbb{V}\textrm{ar}[\widehat{\pi(\varphi)}]<\infty$, 
where $\mathbb{E}$ and $\mathbb{V}\textrm{ar}$ denote expectation and variance w.r.t.~the law of the stochastic scheme to be developed, respectively.  
This scheme 
is of interest for several reasons:
\begin{enumerate}
\item{One can produce unbiased estimators of score functions which can be employed within stochastic gradient algorithms to perform parameter inference.}
\item{One can simulate i.i.d.~replicates of such unbiased estimators in parallel and combine them 
to construct lower variance estimators in a static context (sometimes referred to as strong parallel scaling).}
\item{The method provides a benchmark for other computations.}
\end{enumerate}
In terms of the first point, it is often simpler to verify the validity of stochastic gradient algorithms when the estimate of the noisy gradient is unbiased. 
The second point means that the variance can be reduced proportionally to the number of available processors,
for the same fixed expected cost per processor.
The third point means that one can check the precision of biased methods against the results. 

The approach that we follow is based upon an idea that was outlined in \cite{rssb_disc} and belongs to the class of \emph{doubly-randomized estimators} -- more specifically, estimators which arise from applying randomization of the type \cite{mcl,rhee} (see also \cite{vihola}) {\em twice}, in a nested fashion. The baseline version of the randomized estimators 
of \cite{mcl,rhee} place 
a probability distribution $\bbP_L$ over the level of discretization $l$. 
Given a simulation from this probability distribution, one way to obtain unbiased estimates of $\pi(\varphi)$ 
is to unbiasedly estimate $\pi_l(\varphi)-\pi_{l-1}(\varphi)$ for $l\in\mathbb{N}$, 
or an unbiased estimate of $\pi_0(\varphi)$ if one samples $l=0$. 
Denoting these estimators by $\xi_l$, the so-called ``single-term'' estimator is given by $\xi_L/\bbP_L(L)$, 
where $L\sim \bbP_L$.
In the inference context, it is challenging to obtain unbiased estimators, 
and this is where the second randomization comes into the picture.
For $\xi_0$, one can use the recently introduced
unbiased MCMC scheme of \cite{jacob1} (see also \cite{glynn2}). 
This estimator is built by truncating an 
infinite series of increments of coupled MCMCs once the chains meet.
The main complication is then to 
unbiasedly estimate $\pi_l(\varphi)-\pi_{l-1}(\varphi)$.
It will typically not suffice to estimate
$\pi_l(\varphi)$ and $\pi_{l-1}(\varphi)$ independently, 
because the resulting estimator would often have infinite variance. Therefore an additional technique is required.
The main contribution of this article is to develop several novel coupled MCMC schemes that can achieve 
unbiased estimates $\xi_l$ for $l>0$
such that the resulting estimator of $\pi(\varphi)$ is unbiased and of finite variance. 
The latter properties are proved mathematically under assumptions. 
We also implement our proposed algorithms on several challenging statistical applications.

The idea of using doubly-randomized estimators has appeared in several recent works. 
The work \cite{ub_bip} utilizes a ``coupled-sum'' estimator over sample size to debias multilevel
estimators of the type introduced in \cite{beskos}, which are then utilized in the framework described above. 
That method is applicable to the static/non-dynamic problems where one can evaluate the target distribution,
up to a normalizing constant, like the method we introduce here.
The work 
\cite{ub_pf} uses a ``single-term'' estimator over sample sizes in order debias estimators of the type
introduced in \cite{mlpf}. 
Those estimators are designed for online inference in dynamic problems, such as state space models,
and partially observed diffusion processes in particular.
The methodology in \cite{ub_bip} has infinite expected cost, 
whereas this is not always the case for the method introduced in this paper.
In a companion paper \cite{ub_grad},
we show how to extend the framework of this article to the context of partially observed diffusion processes. 
A possible alternative to our approach would be that of \cite{sergios}. 

This article is structured as follows. In Sections \ref{sec:mot_ex} and \ref{sec:strategy}, the precise problem is stated and our strategy outlined. 
We show in Section \ref{sec:theory}, under assumptions, that our general approach can produce unbiased and finite variance estimators 
with finite expected costs. In Section \ref{sec:algos}, we present some specific Markov kernels which fall under our general framework. 
We illustrate our methodology on several numerical examples in Section \ref{sec:numerics}. 
The proofs of our mathematical results are given in Appendix \ref{sec:proofs}.

\section{General Framework}\label{sec:frame}

\subsection{Notations}

Let $(\mathsf{X},\mathcal{X})$ be a measurable space.
For $\varphi:\mathsf{X}\rightarrow\mathbb{R}$ we write $\mathcal{B}_b(\mathsf{X})$, 
to denote the collection of bounded measurable functions and, if $\mathsf{X}\subseteq\mathbb{R}^d$, $\mathbb{L}^2(\mathsf{X})$ as the collection of square Lebesgue-integrable functions. For $\varphi\in\mathcal{B}_b(\mathsf{X})$, we write the supremum norm as $\|\varphi\|_{\infty}=\sup_{x\in\mathsf{X}}|\varphi(x)|$. 
We denote the Borel sets on
$\mathbb{R}^d$ as $B(\mathbb{R}^d)$. The $d-$dimensional Lebesgue measure is written as $dx$.
For a metric $\mathsf{d:\mathsf{X}\times\mathsf{X}}\rightarrow\mathbb{R}^+$ on $\mathsf{X}$ and a function $\varphi:\mathsf{X}\rightarrow\mathbb{R}$,
$\textrm{Lip}_{\mathsf{d}}(\mathsf{X})$ are the Lipschitz functions (with finite Lipschitz constants), that is for every $(x,w)\in\mathsf{X}\times\mathsf{X}$, $|\varphi(x)-\varphi(w)|\leq \|\varphi\|_{\textrm{Lip}}\mathsf{d}(x,w)$.
$\mathscr{P}(\mathsf{X})$  denotes the collection of probability measures on $(\mathsf{X},\mathcal{X})$.
For a finite measure $\mu$ on $(\mathsf{X},\mathcal{X})$
and a $\varphi\in\mathcal{B}_b(\mathsf{X})$, the notation $\mu(\varphi)=\int_{\mathsf{X}}\varphi(x)\mu(dx)$ is used.
For $(\mathsf{X}\times\mathsf{W},\mathcal{X}\vee\mathcal{W})$ a measurable space and $\mu$ a non-negative finite measure on this space,
we use the tensor-product of functions notation for $(\varphi,\psi)\in\mathcal{B}_b(\mathsf{X})\times\mathcal{B}_b(\mathsf{W})$,
$\mu(\varphi\otimes\psi)=\int_{\mathsf{X}\times\mathsf{Y}}\varphi(x)\psi(w)\mu(d(x,w))$.
Given a Markov kernel $K:\mathsf{X}\rightarrow\mathscr{P}(\mathsf{X})$ and a finite measure $\mu$, we use the notations
$
\mu K(dx') = \int_{\mathsf{X}}\mu(dx) K(x,dx')
$
and 
$
K(\varphi)(x) = \int_{\mathsf{X}} \varphi(x') K(x,dx'),
$
for $\varphi\in\mathcal{B}_b(\mathsf{X})$. 
The iterated kernel is $K^n(x_0,dx_n) = \int_{\mathsf{X}^{n-1}}\prod_{i=1}^n K(x_{i-1},dx_i)$.
For $A\in\mathcal{X}$, the indicator function is written as $\mathbb{I}_A(x)$. $\mathbb{Z}^+$ is the set of non-negative integers. For $(\mu,\nu)\in\mathscr{P}(\mathsf{X})\times\mathscr{P}(\mathsf{X})$, $\|\mu-\nu\|_{\textrm{tv}}:=\sup_{A\in\mathcal{X}}|\mu(A)-\nu(A)|$ is the total variation distance. $\mathcal{N}_d(\mu,\Sigma)$ is the $d-$dimensional Gaussian distribution with mean $\mu$ and covariance $\Sigma$, with corresponding Lebesgue density
written as $x\mapsto\phi_d(x;\mu,\Sigma)$. 
$\mathcal{U}_A$ denotes the uniform distribution on a measurable set $A$. $0_d$ denotes the $d$-dimensional column vector of zeros. $I_d$ denotes the $d\times d$ identity matrix. The transpose of a vector or matrix $x$ is denoted as $x^T$.
For a set $\mathsf{C} \subset\mathbb{R}^d$, 
the $\mathbb{L}^2(\mathsf{C})$ norm of $f$ is written as
$\|f\|_2 = \int_{\mathsf{C}} f(x)^2 dx$, 
and the space of square integrable functions on $\mathsf{C}$ is denoted by 
$\mathbb{L}^2(\mathsf{C}) = \{f:\mathsf{C} \rightarrow \bbR : \|f\|_2 < \infty \}$. 
For a vector $x\in \bbR^d$, its Euclidean norm is also written as $\|x\|_2$.

\subsection{Motivating Example}\label{sec:mot_ex}
\subsubsection{Problem Specification}
A particular Bayesian inverse problem associated to partial differential equations (PDEs) 
is now introduced as a motivating example.
The objective is to infer the permeability field associated to a porous medium, 
based on pressure measurements of the fluid flow governed by Darcy's law. 
This example is prototypical in the context of subsurface inversion, with applications
ranging from oil recovery to contaminant transport in groundwater \cite{tarantola, stuart}.

Let $\mathsf{C} \subset\mathbb{R}^D$ with the boundary $\partial \mathsf{C}$ 
convex and once continuously differentiable and suppose 
$f\in \mathbb{L}^2(\mathsf{C})$. 
Consider the following PDE for the pressure field $h$ on $\mathsf{C}$:
\begin{align}
-\nabla \cdot (\Phi \nabla h )  & = f, \quad \textrm{ on } \mathsf{C},      \label{eq:pde1} \\
h &= 0, \quad \textrm{ on } \partial \mathsf{C},
\nonumber 
\end{align}
where, for $t\in\mathsf{C}$, the permeability is 
$$
\Phi(t;X) := \bar{\Phi}(t) + \sum_{j=1}^d X_j\vartheta_j v_j(t).
$$
The known forcing $f$ can represent, e.g. injection and/or extraction of fluid from wells. 
In the above:
\begin{itemize}
\item{$X=(X_1,\dots,X_d)$, with $X_j \stackrel{\textrm{i.i.d.}}{\sim} \mathcal{U}_{[-1,1]}$. This determines the prior distribution for $X$ on the state space $\mathsf{X}=
[-1,1]^d$.}
\item{$\bar{\Phi}:\mathsf{C}\rightarrow\mathbb{R}$, and for $j\in\{1,\dots,d\}$, 
$v_j:\mathsf{C}\rightarrow\mathbb{R}$ with sup$_{t\in \mathsf{C}}|v_j(t)| \leq 1$, 
and $\vartheta_j\in\mathbb{R}^+$.}
\item{$h(\cdot;X)$ (or just $h(X)$) denotes the weak solution of \eqref{eq:pde1} for a given $X$. }
\end{itemize}
We remark that one can allow $d\rightarrow \infty$ in the above if
$\vartheta_j$ decay to zero sufficiently fast with $j$; 
see \cite{beskos2,beskos} and the references therein for further details.
The following will be assumed.
\begin{hypH}
\label{hyp:N}
There exists a $\Phi^{\star}>0$ such that 
$\inf_{t\in\mathsf{C}}\bar{\Phi}(t) \geq \sum_{j=1}^d \vartheta_j + \Phi^{\star}$. 
\end{hypH}
This assumption guarantees that for all $x \in \mathsf{X}$,
$\inf_{t\in\mathsf{C}}\Phi(t; x) \geq \Phi^{\star}$, 
hence there is a well-defined and unique weak solution $h(x) \in \mathbb{L}^2(\mathsf{C})$, 
and $\sup_{x\in \mathsf{X}} \|h(x)\|_2 \leq C$ for some $C>0$ \cite{ciarlet}.

Define the vector-valued function $G : \mathsf{X} \rightarrow \bbR^P$ by
\begin{align}\label{eqn:observation_function}
G: x \mapsto G(x) = [g_1(h(x)),\dots,g_P(h(x))],
\end{align}
where $g_p: \mathbb{L}^2(\mathsf{C}) \rightarrow \bbR$ 
are bounded linear functionals on $\mathbb{L}^2(\mathsf{C})$
for $p\in\{1,\dots, P\}$. 
It is assumed that the data $y\in\mathbb{R}^P$ take the form  
\begin{equation}\label{eq:data}
Y | (X=x) \sim \calN_P(G(x),\theta^{-1} I_P) \, ,
\end{equation}
where $\theta>0$ is a parameter that we will be interested in inferring. 
In fact, unbiased estimators are particularly useful in this context, 
and it will be considered in the numerical examples of Section \ref{sec:numerics}. 
We simplify notation by suppressing explicit dependence on parameter $\theta$ and data $y$, 
and write the un-normalized Lebesgue density of $X$ for fixed $y$ and $\theta$ as
\begin{align}\label{eq:unnol_exact}
\gamma(x) = \exp\Big\{-\frac{\theta}{2}\|y-G(x)\|_2^2\Big\}\mathbb{I}_{\mathsf{X}}(x),
\end{align}
and the normalized density as
$\pi(x) = \gamma(x)/ \int_{\mathsf X}\gamma(x)dx$. 
These densities will be written as $\gamma_{\theta}(x)$ and $\pi_{\theta}(x)$ 
when we consider inference for $\theta$. 
This posterior distribution is in general intractable due to the nonlinear dependence of $y$ on $x$,
even if the PDE were to admit an analytical solution, and one must resort to computationally intensive
inference methods such as MCMC.
A further complication is that the analytical solution of the PDE is in general not available, 
so one must resort to numerical approximations, which will be discussed in the next section.

\subsubsection{Discretization}\label{ssec:example}

For simplicity we present the case $D=1$ and $\mathsf{C}=[0,1]$,
but extension to higher dimensions is straightforward -- see e.g. \cite{beskos2, brenner}. 
The PDE problem at resolution level $l$ is solved using a finite element method (FEM) with piecewise linear shape functions on a uniform mesh of width $\Delta_l=2^{-(l+l_0)}$, 
for $l\in\{1,2,\dots\}$ and $l_0\geq 0$ a maximal mesh width. 
In particular, the finite element basis functions $\{\psi_i^l\}_{i=1}^{\Delta_l^{-1}-1}$ on level $l$ 
are defined as follows for $t_i = i~ 2^{-l}$:
$$
\psi_i^l(t) =  \left\{\begin{array}{ll}
\Delta_l^{-1}[t-(t_i-\Delta_l)],
& \textrm{if}~t\in[t_i-\Delta_l,t_i]\, , \\
\Delta_l^{-1}[t_i+\Delta_l-t], & \textrm{if}~t\in[t_i,t_i+\Delta_l]\, .
\end{array}\, \right  .
$$
To solve the PDE for a given $x\in \mathsf{X}$, 
$h_l(x)=\sum_{i=1}^{\Delta_l^{-1}-1}h_i^l(x) \psi_i^l$ is substituted into \eqref{eq:pde1}, 
and projected onto each basis element:
$$
-\Big\langle \nabla\cdot\Big({\Phi}\nabla\sum_{i=1}^{\Delta_l^{-1}-1}h_i^l(x) \psi_i^l \Big),\psi_j^l \Big\rangle = \langle f, \psi_j^l \rangle \, .
$$
We introduce the matrix $\bm{A}^l(x)$ with entries $A_{ij}^l(x) = \langle {\Phi}(x) \nabla\psi_i^l,\nabla\psi_j^l \rangle$, and vectors $\bm{h}^l(x), \bm{f}^l$ with entries $h_i^l(x)$ and $f_i^l=\langle f, \psi_i^l\rangle$, respectively.
Solving the discretized problem involves solving the following linear system 
\begin{equation}\label{eq:forward}
\bm{A}^l(x)\bm{h}^l(x) = \bm{f}^l \, .
\end{equation}

Define $G_l(x) = [g_1(h_l(x)),\dots,g_P(h_l(x))]$. 
We denote the corresponding approximated un-normalized density by 
\begin{equation}\label{eq:unnol}
\gamma_l(x) = \exp\Big\{-\frac{\theta}{2}\|y-G_l(x)\|_2^2\Big\}\mathbb{I}_{\mathsf{X}}(x),
\end{equation}
and the approximated normalized density by
$
\pi_l(x) = \gamma_l(x)/\int_{\mathsf{X}}\gamma_l(x)dx . 
$
We now present some fundamental convergence results relating to this approximation, 
which are crucial for the application of our proposed methodology.
\begin{prop}\label{prop:disc}
Assume (H\ref{hyp:N}). 
\begin{enumerate}
\item{For all $x \in \mathsf{X}$, 
$$
\lim_{l\rightarrow\infty}h_l(x)=h(x).
$$ 
In addition, there exists a $C\in(0,\infty)$ such that for every $(l,x)\in\mathbb{Z}^+\times\mathsf{X}$:
\begin{equation}\label{eq:hl}
\|h_l(x)\|_{\infty} \leq C \,  \qquad {\rm and} \qquad  
\|h_l(x) - h(x) \|_{2}^2 \leq C \Delta_l^{2\beta} \, 
\end{equation}
with $\beta=2$.}
\item{
For all $x \in \mathsf{X}$, 
$$
\lim_{l\rightarrow\infty}\gamma_l(x)=\gamma(x).
$$ 
In addition, there exists a $(C,\tilde{l})\in(0,\infty)\times\mathbb{Z}^+$ such that for every $(l,x)\in\{\tilde{l},\tilde{l}+1,\dots\}\times\mathsf{X}$:
\begin{equation}\label{eq:gaml}
|\gamma_l(x)-\gamma(x) |^2 \leq C \Delta_l^{2\beta} \, \qquad {\rm and} \qquad 
|\pi_l(x)-\pi(x) |^2 \leq C \Delta_l^{2\beta}
\end{equation}
with $\beta=2$.}
\end{enumerate}
\end{prop}

\begin{proof}
The first part in \eqref{eq:hl} is a standard result in finite element methods \cite{brenner, ciarlet}.
For the second, recall that for $p\in\{1,\dots,P\}$, each $g_p$ is a bounded linear functional.
Using this fact and \eqref{eq:hl}, it is straightforward to establish \eqref{eq:gaml}
-- see e.g.~\cite{beskos, ub_bip} and the references therein.
\end{proof}

\subsection{Unbiased Estimation}\label{sec:strategy}

\subsubsection{Overall Strategy}\label{sec:ov_app}
We now describe our strategy to construct unbiased estimators of $\pi(\varphi)$. 
Consider a positive probability mass function, $\mathbb{P}_L$, on $\mathbb{Z}^+$. It is known \cite{rhee,vihola} that if one can find a sequence of independent random variables $(\xi_l)_{l\geq 0}$ independent of $L\sim\mathbb{P}_L$ such that
\begin{align}
&\mathbb{E}[\xi_0]  = \pi_0(\varphi), \label{eq:ub_pi0}\\
&\mathbb{E}[\xi_l]  =  \pi_l(\varphi) - \pi_{l-1}(\varphi),\quad\forall l\in\mathbb{N}, \label{eq:ub_inc}\\
&\sum_{l\in\mathbb{Z}^+}\frac{\mathbb{E}[\xi_l^2]}{\mathbb{P}_L(l)}  < +\infty, \label{eq:main_cond}
\end{align}
then 
\begin{equation}\label{eq:single_term}
\widehat{\pi(\varphi)}_S := \frac{\xi_L}{\mathbb{P}_L(L)}
\end{equation}
is an unbiased and finite variance estimator of $\pi(\varphi)$. 
This is the `single term' estimator as discussed by \cite{rhee,vihola}, and alternatives such as the `independent sum' estimator are also possible.
In this latter case, if one can construct independent random variables $(\xi_l)_{l\geq 0}$ that are independent of $L\sim\mathbb{P}_L$,
which satisfy \eqref{eq:ub_pi0}-\eqref{eq:ub_inc} and additionally that
\begin{equation}\label{eq:coup_sum_cond}
\sum_{l\in\mathbb{Z}^+}\frac{\mathbb{V}\textrm{ar}[\xi_l] + (\pi_{l}(\varphi)-\pi(\varphi))^2}{\overline{\mathbb{P}}_L(l)} < +\infty \, ,
\end{equation}
where $\overline{\mathbb{P}}_L(l)=\sum_{k\geq l}\mathbb{P}_L(k)$, 
then
\begin{equation}\label{eq:coupled_sum}
\widehat{\pi(\varphi)}_I := \sum_{l=0}^L\frac{\xi_l}{\overline{\mathbb{P}}_L(l)}
\end{equation}
is also an unbiased estimator of $\pi(\varphi)$ with finite variance. 
Typically, one will run $N\in\mathbb{N}$ independent replicates of either \eqref{eq:single_term} or \eqref{eq:coupled_sum} and then use the average
$$
\frac{1}{N}\sum_{i=1}^N \Big(\widehat{\pi(\varphi)}_k\Big)^i,
$$
where $k\in\{S,I\}$ and $\Big(\widehat{\pi(\varphi)}_k\Big)^i$ represents the $i^{th}-$independent replicate of the estimate. 

We note that this idea was mentioned in \cite{rssb_disc} 
and also used in various different ways in \cite{ub_vihola, ub_pf, ub_bip}. 
The main point of these schemes is that one can completely remove the discretization bias (represented by $l$)
associated to for example Euler discretizations of stochastic differential equations or
FEM discretizations 
of PDEs, 
whilst only working with biased versions of $\pi$, denoted as $\pi_l$, with $l<\infty$. 
In addition, the method is completely parallelizable,
as one can run each replicate independently.

We remark that the condition  \eqref{eq:ub_inc} is not a necessary one (in the case of \eqref{eq:single_term}, but is needed for \eqref{eq:coupled_sum}), but is certainly sufficient. In many contexts, satisfying \eqref{eq:ub_inc} is not trivial as exact simulation from any of the distributions $(\pi_l)_{l\in\mathbb{Z}^+}$ is often not possible. 
In the work of \cite{glynn2,jacob1} (see also \cite{heng_jacob_2019,jacob2}), 
the authors consider a methodology to unbiasedly estimate $\pi_l(\varphi)$ for each $l$,
which we shall build upon. 
In order to satisfy \eqref{eq:main_cond} (or \eqref{eq:coup_sum_cond}), it will typically not be sufficient (or at least efficient) to run two \emph{independent} unbiased MCMC algorithms to estimate $\pi_l(\varphi)$ and $\pi_{l-1}(\varphi)$ respectively. 
To see this, let $l\in\mathbb{N}$ be given and 
consider the easier situation where one can sample exactly from $\pi_l$ and $\pi_{l-1}$. 
Then an unbiased estimator 
of $[\pi_l-\pi_{l-1}](\varphi)$ is given by 
$$
\xi_l = \frac{1}{N}\sum_{n=1}^N\{\varphi(X_n)-\varphi(W_{n})\},
$$
where $(X_n)_{n\in\mathbb{N}}$ are i.i.d.~from $\pi_l$ and $(W_n)_{n\in\mathbb{N}}$ are i.i.d.~from $\pi_{l-1}$. 
In this case, we have
$$
\mathbb{E}[\xi_l^2] = \frac{\mathbb{V}\textrm{ar}_{\pi_l}[\varphi(X)]+\mathbb{V}\textrm{ar}_{\pi_{l-1}}[\varphi(W)]}{N} + [\pi_l-\pi_{l-1}](\varphi)^2.
$$
In practice, it will typically be difficult to choose $\mathbb{P}_L$ so that the estimator would satisfy \eqref{eq:main_cond}
and have finite expected cost, unless one takes $N=\mathcal{O}(2^l)$. 
We shall introduce a novel solution to circumvent this difficulty. 
To motivate our approach, we begin by recalling the methodology in \cite{glynn2,jacob1}.

\subsubsection{Unbiased Markov chain Monte Carlo}\label{sec:ub_mcmc}

We consider the unbiased estimation of $\pi_l(\varphi)$ with $l\in\mathbb{Z}^+$ fixed. 
Suppose we have a $\pi_l-$invariant, ergodic MCMC kernel $K_{l}:\mathsf{X}\rightarrow\mathscr{P}(\mathsf{X})$ 
and an initial distribution $\nu_l\in\mathscr{P}(\mathsf{X})$. 
Let $\tilde{\nu}_l\in\mathscr{P}(\mathsf{X}^2)$ be a coupling of $\nu_l$, i.e. 
$\int_{w\in\mathsf{X}}\tilde{\nu}_l(d(x,w))=\nu_l(dx)$ and $\int_{x\in\mathsf{X}}\tilde{\nu}_l(d(x,w))=\nu_l(dw)$. 
The idea is to run a Markov chain $(Z_{n,l})_{n\in\mathbb{Z}^+}=(X_{n,l},W_{n,l})_{n\in\mathbb{Z}^+}$ 
on $\mathsf{X}\times\mathsf{X}$, initialized from  
\begin{align}\label{eqn:initial_shifted}
\check{\nu}_{l}(d(x,w)):=\int_{\mathsf{X}^2}\tilde{\nu}_l(d(x',w'))K_{l}(x',dx)\delta_{\{w'\}}(dw),
\end{align}
and evolving according to a coupled transition kernel 
$\check{K}_{l}:\mathsf{X}\times\mathsf{X}\rightarrow\mathscr{P}(\mathsf{X}\times\mathsf{X})$,
satisfying 
$$\int_{w' \in \mathsf{X}}\check{K}_{l}((x,w),d(x',w')) = K_{l}(x,dx'),\quad 
\int_{{x}' \in \mathsf{X}}\check{K}_{l}((x,w),d(x',w')) = K_{l}(w,dw').$$
The process is simulated as follows.
\begin{enumerate}
\item{Sample $Z_{0,l}'=(X_{0,l}',W_{0,l}')$ from $\tilde{\nu}_l(\cdot)$.}
\item{Generate $X_{0,l} | X_{0,l}'$ according to $K_l(X_{0,l}',\cdot)$ and set $W_{0,l}=W_{0,l}'$.}
\item{For $n\geq 1$, generate $Z_{n,l} | Z_{n-1,l}$ according to $\check{K}_{l}(Z_{n-1,l},\cdot)$.}
\end{enumerate}
Marginally, the sequences of random variables 
$(X_{n,l})_{n\in\mathbb{Z}^+}$ and $(W_{n,l})_{n\in\mathbb{Z}^+}$ 
are coupled time-homogenous Markov chains with initial distributions 
$\nu K_{l}$ and $\nu$, respectively, and the same transition kernel $K_{l}$. 
Define the meeting time
$$
\tau_{l} := \inf\{n\geq 1:X_{n,l}=W_{n,l}\}.
$$
The coupled chains $(Z_{n,l})_{n\in\mathbb{Z}^+}$ should be constructed 
so that $\tau_{l}$ is almost surely finite, at the very least. 
In addition, we require that the chains remain faithful after meeting, that is 
$X_{n,l}=W_{n,l}$ for all $n\geq\tau_{l}$.  

Under fairly weak assumptions (e.g.~\cite{glynn2}), the following is an unbiased estimator of $\pi_l(\varphi)$ for any $k\in\mathbb{Z}^+$
\begin{equation}\label{eq:stan_est}
\widehat{\pi_l(\varphi)}_k := \varphi(X_{k,l}) + \sum_{n=k+1}^{\tau_{l}-1}\{\varphi(X_{n,l})-\varphi(W_{n,l})\}.
\end{equation}
We remark that a time-averaged extension is also possible; let 
$(m,k)\in\mathbb{Z}^+\times \mathbb{Z}^+$, with $m\geq k$, then one can also use (see \cite{jacob1}) the time-averaged estimator
\begin{equation}\label{eq:time_ave}
\widehat{\pi_l(\varphi)}_{T,k,m} := \frac{1}{m-k+1}\sum_{n=k}^m\varphi(X_{n,l}) + \sum_{n=k+1}^{\tau_{l}-1}\Big(1\wedge\frac{n-k}{m-k+1}\Big)\{\varphi(X_{n,l})-\varphi(W_{n,l})\},
\end{equation}
which recovers \eqref{eq:stan_est} in the case $m=k$.

We now describe a constructive procedure to generate such a coupled Metropolis--Hastings (MH) kernel $\check{K}_{l}$ \cite{jacob1}. 
Let $Q_l:\mathsf{X}\rightarrow\mathscr{P}(\mathsf{X})$ be a proposal Markov kernel, 
such that $\pi_l(dx)Q_l(x,dx')$ has a positive density $\pi_l(x)q_l(x,x')$ w.r.t.~a dominating measure. 
We define the MH acceptance probability for $(x,x')\in\mathsf{X}\times\mathsf{X}$ as 
$$\alpha_l(x,x') :=  1\wedge \frac{\pi_l(x')q_l(x',x)}{\pi_l(x)q_l(x,x')}.$$
For $(x,w)\in\mathsf{X}\times\mathsf{X}$, 
we define a maximal coupling of the proposal Markov kernels $Q_l(x,dx')$ and $Q_l(w,dw')$
\begin{align}\label{eqn:maximal_level}
\check{Q}_l((x,w),d(x',w')) &= 
S_l(x,w) \int_{u\in\mathsf{X}}
\frac{O_l((x,w),du)}{S_l(x,w)}\delta_{\{u\}^2}(d(x',w'))
+ (1-S_l(x,w))\times\notag \\  &\Big(
\frac{Q_l(x,dx')-O_l((x,w),dx')}
{1-S_l(x,w)}
\Big)\otimes\Big(
\frac{Q_l(w,dw')-O_l((x,w),dw')}
{1-S_l(x,w)}
\Big),
\end{align}
where $O_l((x,w),du):=Q_l(x,du)\wedge Q_l(w,du)$ denotes the overlapping kernel on $\mathsf{X}$
and $S_l(x,w):=\int_{\mathsf{X}}O_l((x,w),du)$ is the size of the overlap. 
Under this transition kernel, with probability $S_l(x,w)$, one simulates 
$X'=W'$ from the overlap $O_l((x,w),du)/S_l(x,w)$,
and with probability $1-S_l(x,w)$, $X'$ and $W'$ are simulated independently 
from the residuals required to ensure they retain the appropriate marginals. 
As this transition achieves the maximum probability of having $X'=W'$, 
$\check{Q}_l$ is known as a maximal coupling of $Q_l(x,dx')$ and $Q_l(w,dw')$. 
This coupling can be simulated using the algorithm of \cite{thor},  
assuming one can sample from the proposal kernels and evaluate their densities. 
One can then obtain a sample from the coupled MH kernel $\check{K}_{l}$ 
by accepting the proposals $X'$ and $W'$ with a common uniform random variable $U\sim\mathcal{U}_{[0,1]}$. 
If the proposal is a Gaussian random walk, then the idea of using maximal couplings within MH goes back to at least \cite{johnson_1996}.
We note that alternatives to a maximal coupling are possible and described in Section \ref{sec:algos}. 

We now outline the procedure required to compute the time-averaged estimator $\widehat{\pi_l(\varphi)}_{T,k,m}$ in \eqref{eq:time_ave}. 
\begin{enumerate}
\item{Sample $Z_{0,l}'=(X_{0,l}',W_{0,l}')$ from $\tilde{\nu}_{l}(\cdot)$.}
\item{Generate $X_{l} | X_{0,l}'$ according to $Q_l(X_{0,l}',\cdot)$ and $U\sim\mathcal{U}_{[0,1]}$. 
If $U<\alpha_l(X_{0,l}',X_l)$, set $X_{0,l}=X_l$, otherwise, set $X_{0,l}=X_{0,l}'$. 
Set $W_{0,l}=W_{0,l}'$ and $n=1$.}
\item{Generate $(X_l,W_l) | (X_{n-1,l}, W_{n-1,l})$ according to $\check{Q}_l((X_{n-1,l}, W_{n-1,l}),\cdot)$ and $U\sim\mathcal{U}_{[0,1]}$. 
\begin{itemize}
\item{
If
$U<\alpha_l(X_{n-1,l},X_l)$, set $X_{n,l}=X_l$, otherwise, set $X_{n,l}=X_{n-1,l}$.}
\item{
If
$U<\alpha_l(W_{n-1,l},W_l)$, set $W_{n,l}=W_l$, otherwise, set $W_{n,l}=W_{n-1,l}$.}
\end{itemize}
}
\item{If $X_{n,l}=W_{n,l}$ and $n\geq m$ stop, otherwise set $n=n+1$ and return to step 3.}
\end{enumerate}
Assuming that the resulting coupled kernel $\check{K}_{l}$ costs twice as much as $K_l$, 
the above procedure requires $(2\tau_{l}-1)\vee (m+\tau_l-1)$ applications of $K_l$.

\subsubsection{Unbiased Estimation of Increments}\label{sec:ub_inc_gen}

We now describe, abstractly, how one can obtain a sequence of independent random variables $(\xi_l)_{l\geq 0}$ with the properties prescribed in Section \ref{sec:ov_app}. Concrete approaches are detailed in Section \ref{sec:algos}.

In the case of $\xi_0$, one can simply use the unbiased MCMC methodology described in Section \ref{sec:ub_mcmc}. 
Hence we consider how to unbiasedly estimate $[\pi_l-\pi_{l-1}](\varphi)$, for a fix $l\in\mathbb{N}$, using MCMC. 
For $s\in\{l,l-1\}$, suppose we have a $\pi_s-$invariant, ergodic MCMC kernel $K_{s}:\mathsf{X}\rightarrow\mathscr{P}(\mathsf{X})$ and 
an initial distribution $\nu_s\in\mathscr{P}(\mathsf{X})$. 
Let $\check{\nu}_{l,l-1}$ be a coupling of the distributions $\check{\nu}_{l}$ and $\check{\nu}_{l-1}$ defined in \eqref{eqn:initial_shifted}. 
We will generate a Markov chain $(Z_{n,l,l-1})_{n\in\mathbb{Z}^+}$ with
$$
Z_{n,l,l-1}=((X_{n,l},W_{n,l}),(X_{n,l-1},W_{n,l-1}))\in(\mathsf{X}\times\mathsf{X})\times(\mathsf{X}\times\mathsf{X})=:\mathsf{Z}, 
$$
for each $n\in\mathbb{Z}^+$. 
The Markov chain is such that, marginally, for each $s\in\{l,l-1\}$, 
the sequence of random variables $(X_{n,s})_{n\in\mathbb{Z}^+}$ and $(W_{n,s})_{n\in\mathbb{Z}^+}$ 
are time-homogenous Markov chains with initial distribution 
$\nu_s K_{s}$ and $\nu_s$, respectively, and the same transition kernel $K_{s}$. 
The four sequences will be constructed in a dependent manner 
in order to satisfy \eqref{eq:main_cond} or \eqref{eq:coup_sum_cond}. 
We denote the transition kernel for $(Z_{n,l,l-1})_{n\in\mathbb{Z}^+}$, as $\check{K}_{l,l-1}:\mathsf{Z}\rightarrow\mathscr{P}(\mathsf{Z})$.
Define the meeting time 
\begin{eqnarray*}
\tau_{s} & =  & \inf\{n\geq 1:X_{n,s}=W_{n,s}\} 
\end{eqnarray*}
 for $s\in\{l,l-1\}$. 
It is explicitly assumed that (at the very least) $(Z_{n,l,l-1})_{n\in\mathbb{Z}^+}$ is constructed so that the stopping time $\check{\tau}_{l,l-1}:=\tau_l\vee\tau_{l-1}$ is almost surely finite. In addition, the pair of chains on each level should be faithful, i.e. for $s\in\{l,l-1\}$, we have 
\begin{align}\label{eqn:faithfulness_bothlevels}
	X_{n,s}=W_{n,s}, \mbox{ for all } n\geq\tau_{s}.
\end{align}
Hence for time $n\geq\check{\tau}_{l,l-1}$, $Z_{n,l,l-1}$ only has a distinct state on each level. 
We will give explicit examples of Markov kernels which satisfy these constraints in Section \ref{sec:algos}. 
Note that we do not require the pairs $(Z_{n,l})_{n\in\mathbb{Z}^+}=(X_{n,l},W_{n,l})_{n\in\mathbb{Z}^+}$ 
or $(Z_{n,l-1})_{n\in\mathbb{Z}^+}=(X_{n,l-1},W_{n,l-1})_{n\in\mathbb{Z}^+}$ 
to be Markov chains with exactly the properties considered in Section \ref{sec:ub_mcmc}.

One can estimate $[\pi_l-\pi_{l-1}](\varphi)$ as follows, for any $k\in\mathbb{Z}^+$:
\begin{equation}\label{eq:stan_est_inc}
\reallywidehat{[\pi_l-\pi_{l-1}](\varphi)}_k := 
\widehat{\pi_l(\varphi)}_k - \widehat{\pi_{l-1}(\varphi)}_k,
\end{equation}
where $\widehat{\pi_s(\varphi)}_k$ is computed using \eqref{eq:stan_est} 
based on the pair of chains $(X_{n,s},W_{n,s})_{n\in\mathbb{Z}^+}$ on level $s\in\{l,l-1\}$. 
One can also employ time-averaging, for $(m,k)\in\mathbb{Z}^+\times \mathbb{Z}^+$ satisfying $m\geq k$:
\begin{align}\label{eq:time_ave_inc}
\reallywidehat{[\pi_l-\pi_{l-1}](\varphi)}_{T,k,m} &:= 
\widehat{\pi_l(\varphi)}_{T,k,m} - \widehat{\pi_{l-1}(\varphi)}_{T,k,m},
\end{align}
where $\widehat{\pi_s(\varphi)}_{T,k,m}$ is computed using \eqref{eq:time_ave} 
based on the pair of chains $(X_{n,s},W_{n,s})_{n\in\mathbb{Z}^+}$ on level $s\in\{l,l-1\}$. 
The steps to compute the estimator \eqref{eq:time_ave_inc} are outlined below.
\begin{enumerate}
\item{Sample $Z_{0,l,l-1}$ from $\check{\nu}_{l,l-1}(\cdot)$. Set $n=1$.}
\item{Generate $Z_{n,l,l-1}|Z_{n-1,l,l-1}$ according to $\check{K}_{l,l-1}(Z_{n-1,l,l-1},\cdot)$.
}
\item{If $X_{n,l}=W_{n,l}$, $X_{n,l-1}=W_{n,l-1}$ and $n\geq m$ stop, otherwise set $n=n+1$ and return to step 2.}
\end{enumerate}
Assuming the cost of $\check{K}_{l,l-1}$ is two times that of running both $K_l$ and $K_{l-1}$, 
and $K_l$ costs twice as much as $K_{l-1}$, then the cost of the above procedure requires 
\begin{align}\label{eqn:cost_increment}
\frac{(2\tau_{l-1}-1)\vee(m+\tau_{l-1}-1)}{2} + (2\tau_{l}-1)\vee(m+\tau_{l}-1)
\leq 2\{(2\check{\tau}_{l,l-1}-1)\vee (m+\check{\tau}_{l,l-1}-1)\}
\end{align}
applications of the kernel $K_{l}$.

\subsubsection{Summary of Proposed Methodology}
We now consolidate the above discussion by summarizing our proposed methodology to unbiasedly estimate $\pi(\varphi)$.
We begin with the single term estimator $\widehat{\pi(\varphi)}_S$ in \eqref{eq:single_term}.
\begin{enumerate}
\item{Sample $L\sim\mathbb{P}_L$.}
\item{If $L=0$, generate a Markov chain $(Z_{n,0})_{n\in\mathbb{Z}^+}$ according to $\check{K}_0$ 
as described in Section \ref{sec:ub_mcmc} and compute the estimator 
$\widehat{\pi_0(\varphi)}_k$ in \eqref{eq:stan_est} or 
$\widehat{\pi_0(\varphi)}_{T,k,m}$ in \eqref{eq:time_ave}. }
\item{If $L>0$, generate a Markov chain $(Z_{n,L,L-1})_{n\in\mathbb{Z}^+}$ according to 
$\check{K}_{L,L-1}$ as described in Section \ref{sec:ub_inc_gen} and compute the estimator 
$\reallywidehat{[\pi_L-\pi_{L-1}](\varphi)}_k $ in \eqref{eq:stan_est_inc} or 
$\reallywidehat{[\pi_L-\pi_{L-1}](\varphi)}_{T,k,m}$ in \eqref{eq:time_ave_inc}.
}
\end{enumerate}
We then return the single term estimator
\begin{equation}\label{eq:basic_ub_est}
\widehat{\pi(\varphi)}_{S,k} := \frac{1}{\mathbb{P}_L(L)}\Big(\mathbb{I}_{\{0\}}(L)\widehat{\pi_0(\varphi)}_k + \mathbb{I}_{\mathbb{N}}(L)\reallywidehat{[\pi_L-\pi_{L-1}](\varphi)}_k \Big)
\end{equation}
or
\begin{equation}\label{eq:basic_ub_est_time}
\reallywidehat{\pi(\varphi)}_{S,T,k,m} := \frac{1}{\mathbb{P}_L(L)}\Big(\mathbb{I}_{\{0\}}(L)
\reallywidehat{\pi_0(\varphi)}_{T,k,m} + \mathbb{I}_{\mathbb{N}}(L)
\reallywidehat{[\pi_L-\pi_{L-1}](\varphi)}_{T,k,m} \Big),
\end{equation}
depending on whether one chooses the time-averaged estimator or not. 

For the independent sum estimator $\widehat{\pi(\varphi)}_I$ in \eqref{eq:coupled_sum}, the steps are quite similar. 
\begin{enumerate}
\item{Sample $L\sim\mathbb{P}_L$.}
\item{If $L=0$, generate a Markov chain $(Z_{n,0})_{n\in\mathbb{Z}^+}$ according to $\check{K}_0$ 
as described in Section \ref{sec:ub_mcmc} and compute the estimator 
$\widehat{\pi_0(\varphi)}_k$ in \eqref{eq:stan_est} or 
$\widehat{\pi_0(\varphi)}_{T,k,m}$ in \eqref{eq:time_ave}. }
\item{If $L>0$, for all $l\in\{1,\ldots,L\}$, 
generate a Markov chain $(Z_{n,l,l-1})_{n\in\mathbb{Z}^+}$ according to 
$\check{K}_{l,l-1}$ as described in Section \ref{sec:ub_inc_gen} and compute the estimator 
$\reallywidehat{[\pi_l-\pi_{l-1}](\varphi)}_k$ in \eqref{eq:stan_est_inc} or 
$\reallywidehat{[\pi_l-\pi_{l-1}](\varphi)}_{T,k,m}$ in \eqref{eq:time_ave_inc}.
}
\end{enumerate}
We then return the independent sum estimator
\begin{equation}\label{eq:coup_ub_est}
\reallywidehat{\pi(\varphi)}_{I,k} := \reallywidehat{\pi_0(\varphi)}_{k} + \sum_{l=1}^L \frac{1}{\overline{\mathbb{P}}_L(l)}\Big(\reallywidehat{[\pi_l-\pi_{l-1}](\varphi)}_k\Big)
\end{equation}
or
\begin{equation}\label{eq:coup_ub_est_time}
\reallywidehat{\pi(\varphi)}_{I,T,k,m} := \reallywidehat{\pi_0(\varphi)}_{T,k,m} + \sum_{l=1}^L \frac{1}{\overline{\mathbb{P}}_L(l)}\Big(\reallywidehat{[\pi_l-\pi_{l-1}](\varphi)}_{T,k,m}\Big).
\end{equation}
As noted, in practice, one can also run the given procedure $N$ times and use an average. 
For instance, in the context of the single term estimator \eqref{eq:basic_ub_est}, one would use 
\begin{align}\label{eqn:averaging_singleterm}
\reallywidehat{\pi(\varphi)}(N) = \frac{1}{N}\sum_{i=1}^N\Big(\widehat{\pi(\varphi)}_{S,k}\Big)^i
\end{align}
to estimate $\pi(\varphi)$, where $\Big(\widehat{\pi(\varphi)}_{S,k}\Big)^i$ denotes 
the $i^{th}-$independent replicate of the estimate. 

\subsection{Theoretical Results}\label{sec:theory}
\subsubsection{Assumptions and Results}
The main objective of this section is to establish, under assumptions, that 
\eqref{eq:basic_ub_est}-\eqref{eq:coup_ub_est_time} are unbiased and finite variance estimators of $\pi(\varphi)$. 
We first state the assumptions that we will rely on.  
In the following, we suppose that the quality of approximation of $\pi$ by $\pi_l$ is controlled by a scalar parameter $\Delta_l=2^{-l}$, 
which is consistent with the examples that are to be considered. 
For a constant $0<C<+\infty$ and a metric $\mathsf{d}$ on $\mathsf{X}$, we define the set
\begin{eqnarray*}
B(C,\Delta_l,\mathsf{d}) := \{(x_1,x_2,x_3,x_4)\in\mathsf{X}^4:\forall(i,j)\in\{1,\dots,4\}, \mathsf{d}(x_i,x_j)\leq C\Delta_l\}.
\end{eqnarray*}
We will make the following assumptions with $\mathsf{X}$ compact.

\begin{hypA}\label{ass:1}
There exist $(C,\rho)\in(0,\infty)\times(0,1)$ such that for any $n\in\mathbb{N}$
$$
\sup_{l\in\mathbb{Z}^+} \sup_{x\in \mathsf{X}} \|K_l^n(x,\cdot)-\pi_l(\cdot)\|_{\textrm{tv}} \leq C \rho^n.
$$
\end{hypA}
\begin{hypA}\label{ass:5}
There exist $(C,\rho)\in(0,\infty)\times(0,1)$ such that
for any $(l,n)\in\mathbb{Z}^+\times\mathbb{N}$
$$
\mathbb{E}[\mathbb{I}_{\{\tau_l> n\}}]\leq C\rho^n.
$$
\end{hypA}
\begin{hypA}\label{ass:3}
There exist a $C<\infty$ and a metric $\tilde{\mathsf{d}}:\mathsf{X}\times\mathsf{X}\rightarrow\mathbb{R}^+$ on $\mathsf{X}$, 
such that for any $(l,\varphi,(x,w))\in\mathbb{Z}^+\times\mathcal{B}_b(\mathsf{X})\cap\textrm{Lip}_{\tilde{\mathsf{d}}}(\mathsf{X})\times\mathsf{X}\times\mathsf{X}$
$$
|K_l(\varphi)(x)-K_l(\varphi)(w)| \leq C (\|\varphi\|_{\infty}\vee \|\varphi\|_{\textrm{Lip}})\tilde{\mathsf{d}}(x,w).
$$
\end{hypA}
\begin{hypA}\label{ass:2}
There exist $(C,\beta_1)\in(0,\infty)\times(0,\infty)$ such that for any $(l,\varphi)\in\mathbb{N}\times\mathcal{B}_b(\mathsf{X})$
\begin{enumerate}
\item{
$
|[\pi_l-\pi](\varphi)| \leq C \|\varphi\|_{\infty} \Delta_l^{\beta_1}.
$
}
\item{
$
\sup_{x\in\mathsf{X}}| K_l(\varphi)(x)-K_{l-1}(\varphi)(x) | \leq C  \|\varphi\|_{\infty} \Delta_l^{\beta_1}.
$
}
\end{enumerate}
\end{hypA}
\begin{hypA}\label{ass:4}
There exist $(C,\beta_2,\epsilon)\in(0,\infty)^3$, such that for the metric  $\tilde{\mathsf{d}}$ in (A\ref{ass:3}) and any $(l,n)\in\mathbb{N}\times\mathbb{N}$
\begin{eqnarray*}
\mathbb{E}[\mathbb{I}_{B(C,\Delta_l^{\beta_2},\tilde{\mathsf{d}})^c}(Z_{0,l,l-1})] & \leq & C\Delta_l^{\beta_2(2+\epsilon)}, \\
\mathbb{E}[\mathbb{I}_{B(C,\Delta_l^{\beta_2},\tilde{\mathsf{d}})^c\times B(C,\Delta_l^{\beta_2},\tilde{\mathsf{d}})}(Z_{n,l,l-1},Z_{n-1,l,l-1})] & \leq & C\Delta_l^{\beta_2(2+\epsilon)}.\\
\end{eqnarray*}
\end{hypA}

Our main result focusses upon \eqref{eq:basic_ub_est} as the proof of the other results are more-or-less a direct corollary of the first result.
The proofs of all results are in Appendix \ref{sec:proofs}.

\begin{theorem}\label{theo:main_res}
Assume (A\ref{ass:1}-\ref{ass:4}). Then there exists a choice of positive probability mass function $\mathbb{P}_L$, such that for 
the metric $\tilde{\mathsf{d}}$ in (A\ref{ass:3}) and
any $\varphi\in \mathcal{B}_b(\mathsf{X})\cap\textrm{\emph{Lip}}_{\tilde{\mathsf{d}}}(\mathsf{X})$,  
\eqref{eq:basic_ub_est} is an unbiased and finite variance estimator of $\pi(\varphi)$. 
\end{theorem}

\begin{rem}
It is straightforward to establish that one can find a $\mathbb{P}_L$, so that \eqref{eq:basic_ub_est_time} is also unbiased
with finite variance. The proof can be constructed via the technical results in the appendix.
\end{rem}

The following result can be deduced by observing \eqref{eq:coup_sum_cond} and using the technical results in the appendix.

\begin{cor}
Assume (A\ref{ass:1}-\ref{ass:4}). Then there exists a choice of positive probability mass function $\mathbb{P}_L$, such that for 
the metric $\tilde{\mathsf{d}}$ in (A\ref{ass:3}) and any $\varphi\in \mathcal{B}_b(\mathsf{X})\cap\textrm{\emph{Lip}}_{\tilde{\mathsf{d}}}(\mathsf{X})$,  
\eqref{eq:coup_ub_est} is an unbiased and finite variance estimator of $\pi(\varphi)$. 
\end{cor}

\begin{rem}
As for \eqref{eq:basic_ub_est_time},  one can find a $\mathbb{P}_L$, so that \eqref{eq:coup_ub_est_time} is also unbiased
with finite variance. 
\end{rem}

The main strategy of the proof is to establish a martingale plus remainder type decomposition for $\reallywidehat{\pi_0(\varphi)}_{k}$ and $\reallywidehat{[\pi_l-\pi_{l-1}](\varphi)}_k$. Given this, one can rely on the optional sampling theorem to establish that
the former quantities are unbiased estimates of $\pi_0(\varphi)$ and $[\pi_l-\pi_{l-1}](\varphi)$. One is then left to control the second moments of the decomposition, which can be achieved in a variety of ways; we rely on martingale methods.

\subsubsection{Discussion of Assumptions} 
\label{ssec:ass}

Assumption (A\ref{ass:1}) is a strong assumption, although reasonable as we only consider compact $\mathsf{X}$. It is verified for an example in \cite{mlmcmc}. Assumption (A\ref{ass:5}) has been considered in \cite{ub_grad} and it shown to hold
for a related context; see the results in \cite[Lemmata 14 \& 22]{ub_grad}. 
Assumption (A\ref{ass:3}) has been verified for an example in \cite{mlmcmc}.
Assumption (A\ref{ass:2}) 1.~relates to the rate at which the bias converges and can hold for inverse problems (see \cite{beskos}). 
Assumption (A\ref{ass:2}) 2.~can be achieved by considering an appropriate coupling of $(K_l(x,\cdot),K_{l-1}(x,\cdot))$ and is problem specific; 
see \cite{mlmcmc} for an example where this is verified.

Assumption (A\ref{ass:4}) appears to be quite non-standard. 
We first remark that assumptions of these type (not identical), 
can be verified in complex settings \cite{ub_grad}: 
(A\ref{ass:4}) is shown in \cite[Lemma 16]{ub_grad}. 
Secondly, if one can establish that the pairs $(X_{n,l},X_{n,l-1})_{n\in\mathbb{Z}^+}$ and $(W_{n,l},W_{n,l-1})_{n\in\mathbb{Z}^+}$ 
are (marginally) uniformly ergodic Markov chains with an invariant measure $\check{\pi}_{l,l-1}$, then it is sufficient to assume that
$$
\check{\pi}_{l,l-1}\Big((\varphi \otimes 1 - 1\otimes\varphi)^2\Big) \leq C\Delta_l^{2\beta_2}.
$$
This is because the proof that is used, turns these one-step type properties of the coupled Markov chains in (A\ref{ass:4}) into similar properties of the Markov chain at any time step (see Lemma \ref{lem:good_prob}). 
However, these properties are inherited directly from the invariant measure $\check{\pi}_{l,l-1}$ if such a quantity exists and 
the chain converges sufficiently fast to it.


\subsubsection{Implication of results and choice of $\bbP_L$}
\label{ssec:pl}

The discussion below relates to the single term estimator $\widehat{\pi(\varphi)}_{S,k}$ in \eqref{eq:basic_ub_est} 
with increments estimated using $\xi_l = \reallywidehat{[\pi_l-\pi_{l-1}](\varphi)}_{k}$ defined in \eqref{eq:stan_est_inc}, and is easily extended to the time-averaged estimator \eqref{eq:time_ave_inc}. 
The case of the independent sum estimators \eqref{eq:coup_ub_est}
and \eqref{eq:coup_ub_est_time} follows along the same lines and is thus omitted. 
The variance and expected cost of $\widehat{\pi(\varphi)}_{S,k}$ can be bounded as follows: 
\begin{eqnarray}\label{eq:var}
\mathbb{V}\textrm{ar}\left[\widehat{\pi(\varphi)}_{S,k}\right]&\leq & \sum_{l\in\mathbb{Z}^+}\frac{\bbE[\xi_l^2]}{\mathbb{P}_L(l)} \, , \\
{\rm Cost}\left[\widehat{\pi(\varphi)}_{S,k}\right] &\leq & C\sum_{l\in\mathbb{Z}^+} 
\bbE [\check{\tau}_{l,l-1}] {\rm Cost}(K_l) \mathbb{P}_L(l) \, ,
\label{eq:cost}
\end{eqnarray}
where $C>0$ is some constant (see e.g. \eqref{eqn:cost_increment}),    
${\rm Cost}(K_l)$ denotes the cost of an application of the marginal kernel $K_l$, 
$\check{\tau}_{l,l-1} =\tau_l\vee\tau_{l-1}$ is the stopping time of the Markov chain 
$(Z_{n,l,l-1})_{n\in\mathbb{Z}^+}$ with $\check{\tau}_{0,-1}:=\tau_0$. 
Averaging $N$ single term estimators as in \eqref{eqn:averaging_singleterm} would yield 
a variance of $\mathbb{V}\textrm{ar}[\widehat{\pi(\varphi)}_{S,k}] N^{-1}$ 
and expected cost of ${\rm Cost}[\widehat{\pi(\varphi)}_{S,k}] N$.

For the second moment of $\xi_l$ featuring in \eqref{eq:var}, Lemma \ref{lem:sec_mom} provides the bound 
\begin{equation}\label{eq:varbound}
\bbE[\xi_l^2] \leq \Delta_l^{2\beta},
\end{equation}
with
$\beta = \beta_1 \wedge\beta_2 > 0$, where 
$\beta_1$ and $\beta_{2}$ are given in Assumptions
(A\ref{ass:2}) and (A\ref{ass:4}), respectively, 
and $\tilde{\mathsf{d}}$ in Assumption (A\ref{ass:3}) is given by the standard Euclidean distance.
Using Assumption (A\ref{ass:5}), one can upper-bound the expected stopping time 
$\bbE [\check{\tau}_{l,l-1}]$ appearing in \eqref{eq:cost} by a constant that is independent of $l$. 
We assume furthermore that there exists $\omega<2\beta$ and $C_K>0$ such that the cost of $K_l$ satisfies
\begin{equation}\label{eq:costbound}
{\rm Cost}(K_l) \leq C_K\Delta_l^{-\omega} \, .
\end{equation}
Suppose the probability mass function on $\mathbb{Z}^+$ is of the form $\mathbb{P}_L(l) \propto \Delta_l^\eta$. 
Following \eqref{eq:varbound} and \eqref{eq:costbound}, 
there are constants $C_V,C_C >0$ such that 
the right-hand sides of 
\eqref{eq:var} and \eqref{eq:cost} are bounded above by 
$C_V \sum_{l\in \bbZ_+} \Delta_l^{2\beta-\eta}$ and 
$C_C \sum_{l\in \bbZ_+} \Delta_l^{\eta-\omega}$, respectively. 
Both quantities are finite for any $\eta \in (\omega,2\beta)$,
e.g. one can let $\eta = (2\beta + \omega)/2$.
The above discussion requires $\omega<2\beta$; 
see e.g. \cite{rhee,ub_bip,ub_vihola} for the case $\omega \geq 2\beta$.

We now consider the above discussion in the particular context of the 
motivating example in Section \ref{sec:mot_ex}. 
In this setting, Proposition \ref{prop:disc} provides $\beta_1=2$ in Assumption (A\ref{ass:2}).1.
As discussed in Section \ref{ssec:ass},
the rates for Assumptions (A\ref{ass:2}).2 and (A\ref{ass:4}) depend upon the particular kernels used
and are more difficult to establish theoretically, 
but the value of $\beta$ in \eqref{eq:varbound} can be estimated numerically. 
Evaluation of the density \eqref{eq:unnol} at level $l$ requires the solution of the 
tridiagonal linear system \eqref{eq:forward}, which has  
$\mathcal{O}(\Delta_l^{-1})$ degrees of freedom. 
Therefore \eqref{eq:costbound} holds with $\omega=1$, and we choose $\eta=5/2$.

\section{Specific Kernels}\label{sec:algos}

We consider various strategies to construct the kernels $\check{K}_{l,l-1}$, $l\in\mathbb{N}$ described in Section \ref{sec:ub_inc_gen}. We begin with Metropolis--Hastings algorithms in Sections \ref{sec:coupled_MH} and \ref{sec:coupled_proposal}, 
and consider the case of Hamiltonian Monte Carlo methods in Section \ref{sec:hmc}. 

\subsection{Coupled Metropolis--Hastings Kernels}\label{sec:coupled_MH}
We consider a collection of coupled MH kernels $\check{K}_{l,l-1}$ that are defined by the following simulation procedure. 
\begin{enumerate}
	\item Given current state $z_{l,l-1}=((x_{l},w_l),(x_{l-1},w_{l-1}))\in\mathsf{Z}$, generate proposal 
	$Z_{l,l-1}'=((X_{l}',W_l'),(X_{l-1}',W_{l-1}'))$ according to $\check{Q}_{l,l-1}(z_{l,l-1},\cdot)$.
	
	\item Generate $U\sim\mathcal{U}_{[0,1]}$ and for level $s\in\{l,l-1\}$:
	\begin{itemize}
		\item{If $U<\alpha_s(x_s,X_s')$, set $X_s^{\star}=X_s'$, otherwise set $X_s^{\star}=x_s$.}
		\item{If $U<\alpha_s(w_s,W_s')$, set $W_s^{\star}=W_s'$, 
		otherwise set $W_s^{\star}=w_s$.}
	\end{itemize}
	
	\item Return $Z_{l,l-1}^{\star}=((X_{l}^{\star},W_l^{\star}),(X_{l-1}^{\star},W_{l-1}^{\star}))$ as a sample according to  
	$\check{K}_{l,l-1}(z_{l,l-1},\cdot)$.	
\end{enumerate}

The notation $\check{Q}_{l,l-1}:\mathsf{Z}\rightarrow\mathscr{P}(\mathsf{Z})$ refers to a coupling of the proposal kernels $Q_{l}(x_{l},dx_l')$, $Q_{l}(w_{l},dw_l')$, $Q_{l-1}(x_{l-1},dx_{l-1}')$ and $Q_{l-1}(w_{l-1},dw_{l-1}')$, in the sense that generating 
$Z_{l,l-1}'=((X_{l}',W_l'),(X_{l-1}',W_{l-1}'))$ according to $\check{Q}_{l,l-1}(z_{l,l-1},\cdot)$, 
is marginally equivalent to 
\begin{align}\label{eqn:four_marginals}
X_l'\sim Q_l(x_l,\cdot), \quad W_l'\sim Q_l(w_l,\cdot), \quad  
X_{l-1}'\sim Q_{l-1}(x_{l-1},\cdot), \quad W_{l-1}'\sim Q_{l-1}(w_{l-1},\cdot).
\end{align}
We also require the coupled proposal kernel $\check{Q}_{l,l-1}$ to satisfy 
\begin{align}\label{eqn:positiveprob_meeting}
\check{Q}_{l,l-1}(z_{l,l-1},D\times D)>0,
\end{align}
where $D := \{(x,w)\in\mathsf{X}\times\mathsf{X}:x=w\}$ denotes the diagonal set, 
and the following faithfulness property 
\begin{align}\label{eqn:faithfulness_level}
x_s = w_s ~\Longrightarrow~ X_s' = W_s'\mbox{ almost surely}
\end{align}
for each level $s\in\{l,l-1\}$. 
The condition \eqref{eqn:positiveprob_meeting} requires the coupling mechanism to generate identical proposals 
on each level with positive probability. 
Under \eqref{eqn:faithfulness_level} and the use of a common uniform random variable in Step 2, 
the pair of MH chains on each level would be faithful, as required in \eqref{eqn:faithfulness_bothlevels}. 

In the following, we will consider $\mathsf{X}=\mathbb{R}^d$ and describe concrete examples 
of $\check{Q}_{l,l-1}$ for proposal kernels associated to random walk Metropolis--Hastings (RWMH)
\begin{align}\label{eqn:randomwalk_proposal}
	Q_s(x,dx') = \phi_d(x';x,\Sigma_s)dx', \quad (x,s)\in\mathsf{X}\times\mathbb{Z}^+,
\end{align}
and pre-conditioned Crank-Nicolson (pCN) \cite{pcn,neal}
\begin{align}\label{eqn:pcn_proposal}
	Q_s(x,dx') = \phi_d(x';\rho_sx,(1-\rho_s^2)\Sigma_s)dx', \quad(x,s)\in\mathsf{X}\times\mathbb{Z}^+,
\end{align}
where $\Sigma_s=\sigma_s\sigma_s^T$, $\sigma_s$ is an invertible $d\times d$ matrix 
and $\rho_s\in(-1,1)$. 

\subsection{Coupled Proposal Kernels}\label{sec:coupled_proposal}
\subsubsection{Independent Maximal Couplings}
A naive approach is to employ 
$$
\check{Q}_{l,l-1}^{(I)}(z_{l,l-1},dz_{l,l-1}') := \check{Q}_l(z_l,dz_l')\otimes \check{Q}_{l-1}(z_{l-1},dz_{l-1}'),
$$ 
which independently samples from the maximal coupling of the proposal kernels on each level given in \eqref{eqn:maximal_level}. 
Although it does satisfy the requirements \eqref{eqn:four_marginals}, \eqref{eqn:positiveprob_meeting} and \eqref{eqn:faithfulness_level}, this choice is unlikely to ensure that the conditions 
in \eqref{eq:main_cond} or \eqref{eq:coup_sum_cond} hold, as the proposals on levels $l$ and $l-1$ are sampled independently. 


\subsubsection{Four-Marginal Maximal Couplings}\label{sec:4max}
We now introduce an extension of the maximal coupling in \eqref{eqn:maximal_level} to the case of four marginals. 
For $z_{l,l-1}=((x_{l},w_l),(x_{l-1},w_{l-1}))\in\mathsf{Z}$, define the overlapping kernel on $\mathsf{X}$ as
$$
O_{l,l-1}(z_{l,l-1},du) := Q_l(x_l,du)\wedge Q_l(w_l,du) \wedge Q_{l-1}(x_{l-1},du)\wedge Q_{l-1}(w_{l-1},du),
$$
and the size of the overlap 
$
S_{l,l-1}(z_{l,l-1}) := \int_{\mathsf{X}}O_{l,l-1}(z_{l,l-1},du).
$
For each level $s\in\{l,l-1\}$, we define the two residual probability measures on $\mathsf{X}$
\begin{align*}
r_{s}(z_{l,l-1},du) := \frac{Q_s(x_s,du) - O_{l,l-1}(z_{l,l-1},du)}{1-S_{l,l-1}(z_{l,l-1})},\quad 
\widetilde{r}_{s}(z_{l,l-1},du) := \frac{Q_s(w_s,du) - O_{l,l-1}(z_{l,l-1},du)}{1-S_{l,l-1}(z_{l,l-1})}.
\end{align*}
To accommodate the event when the pair of chains on a level have met, i.e. $(x_s,w_s)\in D$, we define 
$$
R_s(z_{l,l-1},d(x_s',w_s')) := r_{s}(z_{l,l-1},dx_s')\Big(\mathbb{I}_D(x_s,w_s)\delta_{\{x_s'\}}(dw_s')
+ \mathbb{I}_{D^c}(x_s,w_s)\widetilde{r}_{s}(z_{l,l-1},dw_s')\Big).
$$
Writing 
\begin{equation}\label{eq:4way_residuals}
\check{R}_{l,l-1}(z_{l,l-1},dz_{l,l-1}')  :=  R_l(z_{l,l-1},d(x_l',w_l'))R_{l-1}(z_{l,l-1},d(x_{l-1}',w_{l-1}')),
\end{equation}
we may then define our coupled proposal kernel $\check{Q}_{l,l-1}^{(M)}:\mathsf{Z}\rightarrow\mathscr{P}(\mathsf{Z})$ as 
\begin{align}
&\check{Q}_{l,l-1}^{(M)}(z_{l,l-1},dz_{l,l-1}') \notag\\
:= ~&S_{l,l-1}(z_{l,l-1})\int_{\mathsf{X}}
\frac{O_{l,l-1}(z_{l,l-1},du)}{S_{l,l-1}(z_{l,l-1})}
\delta_{\{u\}^4}(dz_{l,l-1}') +  
(1-S_{l,l-1}(z_{l,l-1}))\check{R}_{l,l-1}(z_{l,l-1},dz_{l,l-1}').\label{eq:4way_maximal}
\end{align}
One can check that this satisfies the requirements in \eqref{eqn:four_marginals}, \eqref{eqn:positiveprob_meeting} and \eqref{eqn:faithfulness_level}. 
Moreover, $\check{Q}_{l,l-1}^{(M)}$ is a maximal coupling as it achieves the maximum probability of having identical proposals
$X_{l}'=W_l'=X_{l-1}'=W_{l-1}'$, 
which is given by the size of the overlap $S_{l,l-1}(z_{l,l-1})$.

Algorithm \ref{alg:4way_maximalcoupling} provides a method to
sample from \eqref{eq:4way_maximal}, assuming that we can sample
from the proposal transition kernels and evaluate their transition
densities. These assumptions clearly hold for the Gaussian proposal kernels in 
\eqref{eqn:randomwalk_proposal} and \eqref{eqn:pcn_proposal}. 
Step 1 can be seen as an attempt to sample $X_{l}'=W_l'=X_{l-1}'=W_{l-1}'$ 
from the overlap $O_{l,l-1}(z_{l,l-1},du)/S_{l,l-1}(z_{l,l-1})$, 
and if this fails, Step 2 corresponds to a rejection sampler 
to sample from the residuals \eqref{eq:4way_residuals}. 
The four cases considered in Step 2 are needed to ensure faithfulness property in \eqref{eqn:faithfulness_level}.

The coupling $\check{Q}_{l,l-1}^{(M)}$ in \eqref{eq:4way_maximal} can be readily 
employed on RWMH and pCN proposals, and more general proposal transitions outside the Gaussian 
family. Next we present an alternative coupling that is specific to the Gaussian case. 

\begin{algorithm}
\textbf{Input}: transition kernels $Q_{s}(x_{s},dx_{s}')$ and $Q_{s}(w_{s},dw_{s}')$ 
for level $s\in\{l,l-1\}$.
\begin{enumerate}
\item Sample $U\sim Q_{l}(x_{l},\cdot)$. With probability 
\[
\min\left\{ 1,\frac{q_{l}(w_{l},U)}{q_{l}(x_{l},U)},\frac{q_{l-1}(x_{l-1},U)}{q_{l}(x_{l},U)},\frac{q_{l-1}(w_{l-1},U)}{q_{l}(x_{l},U)}\right\} ,
\]
output $((X_{l}',W_{l}'),(X_{l-1}',W_{l-1}'))=((U,U),(U,U))$.
\item Otherwise 
\begin{enumerate}
\item Set $X_l'=U$. If $x_{l}=w_{l}$, set $W_l'=U$.
\item If $x_{l}\neq w_{l}$, propose $U_{l}\sim Q_{l}(w_{l},\cdot)$ and with probability 
\[
1-\min\left\{ 1,\frac{q_{l}(x_{l},U_l)}{q_{l}(w_{l},U_{l})},\frac{q_{l-1}(x_{l-1},U_{l})}{q_{l}(w_{l},U_{l})},\frac{q_{l-1}(w_{l-1},U_{l})}{q_{l}(w_{l},U_{l})}\right\},
\]
set $W_{l}'=U_{l}$; otherwise repeat until acceptance. 

\item Propose $U_{l-1}\sim Q_{l-1}(x_{l-1},\cdot)$ and with probability 
\[
1-\min\left\{ 1,\frac{q_{l}(x_{l},U_{l-1})}{q_{l-1}(x_{l-1},U_{l-1})},\frac{q_{l}(w_{l},U_{l-1})}{q_{l-1}(x_{l-1},U_{l-1})},\frac{q_{l-1}(w_{l-1},U_{l-1})}{q_{l-1}(x_{l-1},U_{l-1})}\right\},
\]
set $X_{l-1}'=U_{l-1}$; otherwise repeat until acceptance. If $x_{l-1}=w_{l-1}$, set $W_{l-1}'=X_{l-1}'$.

\item If $x_{l-1}\neq w_{l-1}$, propose $U_{l-1}\sim Q_{l-1}(w_{l-1},\cdot)$ and with probability 
\[
1-\min\left\{ 1,\frac{q_{l}(x_{l},U_{l-1})}{q_{l-1}(w_{l-1},U_{l-1})},\frac{q_{l}(w_{l},U_{l-1})}{q_{l-1}(w_{l-1},U_{l-1})},\frac{q_{l-1}(x_{l-1},U_{l-1})}{q_{l-1}(w_{l-1},U_{l-1})}\right\},
\]
set $W_{l-1}'=U_{l-1}$; otherwise repeat until acceptance. 
\end{enumerate}
\end{enumerate}

\textbf{Output}: Sample $Z_{l,l-1}'=((X_{l}',W_{l}'),(X_{l-1}',W_{l-1}'))$ from $\check{Q}_{l,l-1}^{(M)}(z_{l,l-1},\cdot)$.
\caption{A maximal coupling of four transition kernels} 
\label{alg:4way_maximalcoupling}
\end{algorithm}

\subsubsection{Synchronous Pairwise Reflection Maximal Couplings}\label{sec:4refmax}
We consider the case $\mathsf{X}=\mathbb{R}^d$ and proposal kernels of the form 
\begin{align}\label{eqn:gaussian_proposal}
	Q_s(x,dx') = \phi_d(x';\mu_s(x),\Sigma_s)dx', \quad (x,s)\in\mathsf{X}\times\mathbb{Z}^+,
\end{align}
where $\mu_s:\mathsf{X}\rightarrow\mathsf{X}$, $\Sigma_s=\sigma_s\sigma_s^T$ and $\sigma_s$ is an invertible $d\times d$ matrix. 
The following will exploit the fact that a sample $X'$ from $Q_s(x,\cdot)$ can be represented as 
$X'=\mu_s(x) + \sigma_sv_s$ with $v_{s}\sim\mathcal{N}_d(0_{d},I_{d})$. 
As noted in \cite{jacob1}, the case of $v_{s}$ following a spherically symmetric distribution 
can also be accommodated.

The coupled proposal kernel $\check{Q}_{l,l-1}^{(R)}:\mathsf{Z}\rightarrow\mathscr{P}(\mathsf{Z})$ 
that we construct here is based on a synchronous coupling of the reflection maximal coupling 
in \cite{bourabee_eberle_zimmer} for the pair of proposals on each level. 
Algorithm \ref{alg:coupled4_reflectionmaximal} details how to obtain a sample 
$Z_{l,l-1}'=((X_{l}',W_{l}'),(X_{l-1}',W_{l-1}'))$ from $\check{Q}_{l,l-1}^{(R)}(z_{l,l-1},\cdot)$, 
which will satisfy the requirements in \eqref{eqn:four_marginals}, \eqref{eqn:positiveprob_meeting} and \eqref{eqn:faithfulness_level}. 
The synchronous use of $v_l=v_{l-1}=v$ in Step 1 induces a coupling 
between the proposals across levels. 
In Step 3, for each level $s\in\{l,l-1\}$, 
note that the event $\widetilde{v}_{s}=v_{s}+u_{s}$ yields identical 
proposals $X_s'=W_s'$, and this occurs with maximal probability given by the size of the overlap 
$S_s(x_s,w_s)$. 
When identical proposals are not possible, we take $\widetilde{v}_{s}$ as the reflection of 
$v_s$ with respect to the hyperplane orthogonal to $e_s$, and right between 
$\sigma_s^{-1}x_s$ and $\sigma_s^{-1}w_s$. 
Under this reflection coupling \cite{lindvall_rogers_1986}, one can show that 
$X_s' - W_s' = \varrho(v_s)(\mu_s(x_s) - \mu_s(w_s))$ with 
$\varrho(v_s) = 1 + 2 (v_s^T e_s)e_s$. 
Since $\varrho(v_s)\sim\mathcal{N}(1,4\|u_s\|_2^{-2})$, one has contraction of the proposals 
on each level with probability of almost $1/2$ when $\|u_s\|$ is large. 
If the proposals on each level are both accepted, contraction is desirable as it leads to states 
$X_{s}'$ and $W_{s}'$ that are closer. This in turn yields a higher probability of generating 
identical proposals in the next application of $\check{Q}_{l,l-1}^{(R)}$. 

As \eqref{eqn:gaussian_proposal} clearly includes 
\eqref{eqn:randomwalk_proposal} and \eqref{eqn:pcn_proposal} as special cases, 
the coupling $\check{Q}_{l,l-1}^{(R)}$ is applicable to 
both RWMH and pCN proposals. 
In the next section, we consider another construction for pCN that always induces 
contractive proposals.



\begin{algorithm}
\textbf{Input}: transition kernels $Q_{s}(x_{s},dx_{s}')=\phi_d(x_{s}';\mu_s(x_{s}),\Sigma_{s})dx_{s}',$ 
and $Q_{s}(w_{s},dw_s')=\phi_{d}(w_{s}';\mu_s(w_{s}),\Sigma_{s})dw_{s}'$ for level $s\in\{l,l-1\}$.

Sample $v\sim\mathcal{N}_d(0_{d},I_{d})$ and for level $s\in\{l,l-1\}$:
\begin{enumerate}
\item Set $v_s=v$ and $X_{s}'=\mu_s(x_{s})+\sigma_{s}v_s$.
\item Set $u_{s}=\sigma_{s}^{-1}(\mu_s(x_{s})-\mu_s(w_{s}))$ 
and $e_{s}=u_{s}/\|u_{s}\|_2$.
\item With probability $\min\left\{ 1,\phi
_d(v_{s}+u_{s};0_{d},I_{d})/\phi_d(v_{s};0_{d},I_{d})\right\}$, 
set $\widetilde{v}_{s}=v_{s}+u_{s}$; otherwise set $\widetilde{v}_{s}=v_{s}-2(v_{s}^Te_{s})e_{s}$. Set $W_{s}'=\mu_s(w_{s})+\sigma_{s}\widetilde{v}_s$.
\end{enumerate}
\textbf{Output}: Sample $Z_{l,l-1}'=((X_{l}',W_{l}'),(X_{l-1}',W_{l-1}'))$ from $\check{Q}_{l,l-1}^{(R)}(z_{l,l-1},\cdot)$.
\caption{Synchronous pairwise reflection maximal couplings}
\label{alg:coupled4_reflectionmaximal}
\end{algorithm}

\subsubsection{Synchronous Pre-conditioned Crank Nicolson}\label{sec:sync_coup}

We consider a coupled proposal kernel $\check{Q}_{l,l-1}^{(S)}:\mathsf{Z}\rightarrow\mathscr{P}(\mathsf{Z})$ for pCN \eqref{eqn:pcn_proposal}. 
To obtain a sample $Z_{l,l-1}'=((X_{l}',W_{l}'),(X_{l-1}',W_{l-1}'))$ from $\check{Q}_{l,l-1}^{(S)}(z_{l,l-1},\cdot)$,
we simulate $v\sim\mathcal{N}_d(0_d,I_d)$ and take 
\begin{align*}
	X_{s}'=\rho_{s}x_{s}+\sqrt{1-\rho_{s}^{2}}\sigma_{s}v,\quad 
	W_{s}'=\rho_{s}w_{s}+\sqrt{1-\rho_{s}^{2}}\sigma_{s}v,
\end{align*}
for both levels $s\in\{l,l-1\}$. The synchronous use of $v$ guarantees contraction of the proposals on each level, 
i.e. $\|X_{s}' - W_{s}'\|_2 = \rho_s \| x_s - w_s \|_2$ for $\rho_s\in(-1,1)$, 
and induces dependencies between the pairs of proposals across levels. 
The latter is crucial in our context as we want \eqref{eq:main_cond} or \eqref{eq:coup_sum_cond} to hold. 
Although this synchronous coupling satisfies the requirements in \eqref{eqn:four_marginals} and \eqref{eqn:faithfulness_level}, 
it is not possible to have identical proposals as needed in \eqref{eqn:positiveprob_meeting}. 
A solution is to consider a mixture of $\check{Q}_{l,l-1}^{(S)}$ and 
either $\check{Q}_{l,l-1}^{(M)}$ or $\check{Q}_{l,l-1}^{(R)}$ that allow identical proposals to occur. 
More precisely, we take the coupled proposal kernel in Section \ref{sec:coupled_MH} as the mixture kernel
\begin{align}\label{eqn:pcn_coupledkernel}
	\check{Q}_{l,l-1}(z_{l,l-1},dz_{l,l-1}') = \kappa\check{Q}_{l,l-1}^{(S)}(z_{l,l-1},dz_{l,l-1}') + 
	(1-\kappa)\check{Q}_{l,l-1}^{(p)}(z_{l,l-1},dz_{l,l-1}'), \quad p\in\{M,R\},
\end{align}
which will satisfy \eqref{eqn:four_marginals}, \eqref{eqn:positiveprob_meeting} and \eqref{eqn:faithfulness_level} 
for any choice of $\kappa\in(0,1)$. 

\subsection{Hamiltonian Monte Carlo}\label{sec:hmc}
We restrict ourselves to the case $\mathsf{X}=\mathbb{R}^d$. Hamiltonian Monte Carlo (HMC) \cite{duane} 
considers the following auxiliary target distribution on $(\mathbb{R}^{2d},B(\mathbb{R}^{2d}))$ 
for each $l\in\mathbb{Z}^+$
\begin{align}\label{eqn:extended_target}
 \bar{\pi}_l(d(x,v)) := \pi_l(x)\phi_d(v;0_d,I_d)d(x,v),
\end{align}
where $d(x,v)$ denotes the Lebesgue measure on $\mathbb{R}^{2d}$. 
We will assume that the target density $x\mapsto\pi_l(x)$ has a well-defined gradient, 
and write the Hamiltonian corresponding to \eqref{eqn:extended_target} as 
$H_l(x,v)=-\log\pi_{l}(x)+\frac{1}{2}\|v\|_2^{2}$. 

Given a current position $x\in\mathbb{R}^d$, HMC samples an initial velocity $v\sim\mathcal{N}_d(0_{d},I_{d})$ and 
generates a proposal by discretizing the Hamiltonian dynamics associated to $H_l$ 
using a leapfrog integrator. Given a stepsize $\varepsilon_l>0$ and a number of steps $M_l\in\mathbb{N}$, this
numerical scheme initializes at $(x_{0},v_{0})=(x,v)\in\mathbb{R}^{2d}$ and iterates for $m\in\{0,\ldots,M_l-1\}$
\begin{align}\label{eqn:leapfrog}
v_{m+1/2} = v_{m}+\frac{\varepsilon_l}{2}\nabla\log\pi_{l}(x_{m}),\quad
x_{m+1} = x_{m}+\varepsilon_l v_{m+1/2},\quad
v_{m+1} = v_{m+1/2}+\frac{\varepsilon_l}{2}\nabla\log\pi_{l}(x_{m+1}).
\end{align}
As the Hamiltonian is not exactly conserved under the leapfrog integrator, 
the proposal $(x',v')=(x_{M_l},v_{M_l})$ is then subjected to a Metropolis--Hastings accept-reject step, i.e. 
for $U\sim\mathcal{U}_{[0,1]}$, if  
\begin{align}\label{eqn:acceptreject_hmc}
U \leq \bar{\alpha}_{l}\left\{ (x,v),(x',v')\right\} :=1\wedge\exp\{H_l(x,v)-H_l(x',v')\},
\end{align}
we output $(X^{\star},V^{\star})=(x',v')$, otherwise we output $(X^{\star},V^{\star})=(x,v)$. 
The resulting transition $X^{\star}\sim P_l(x,\cdot)$ on the position coordinate defines a HMC kernel $P_l$ at level $l$.

Following \cite{heng_jacob_2019}, one can construct a faithful coupling of $P_l(x,\cdot)$ and $P_l(w,\cdot)$ 
for $(x,w)\in\mathbb{R}^d\times\mathbb{R}^d$, by employing a common velocity $v$ to initialize \eqref{eqn:leapfrog} and 
a common uniform random variable $U$ in the accept-reject step \eqref{eqn:acceptreject_hmc}. 
If the initial positions $x$ and $w$ are in a region where $\pi_l$ is log-concave,  
and the integration time $\varepsilon_l M_l$ is appropriately chosen, 
this can lead to contractive proposals which are then accepted with high probability for small $\varepsilon_l$. 
The resulting coupled HMC kernel $\check{P}_l((x,w),\cdot)$ cannot be employed within the framework of
Section \ref{sec:ub_mcmc}, as it does not allow chains to meet. 
To circumvent this issue, \cite{heng_jacob_2019} considered a mixture kernel
\begin{align}\label{eqn:coupled2_hmc}
	\check{K}_{l}((x,w),d(x',w')) = 
	\kappa\check{P}_l((x,w),d(x',w')) 
	+ (1-\kappa)\check{K}_{l}^{(M)}((x,w),d(x',w')),
\end{align}
where $\kappa\in(0,1)$ and $\check{K}_{l}^{(M)}$ denotes a coupled RWMH kernel based on 
the maximal coupling in \eqref{eqn:maximal_level}. 
Although the latter enables meetings, the marginal kernel induced by \eqref{eqn:coupled2_hmc} 
is not the HMC kernel $P_l$, but remains close if $\kappa$ is close to one.

We now extend the work of \cite{heng_jacob_2019} to our context. 
For $z_{l,l-1}=((x_{l},w_l),(x_{l-1},w_{l-1}))\in\mathsf{Z}$, 
one can also construct a faithful coupling of $P_l(x_l,\cdot)$, $P_l(w_l,\cdot)$, 
$P_{l-1}(x_{l-1},\cdot)$ and $P_{l-1}(w_{l-1},\cdot)$ by using 
a common initial velocity in the leapfrog integrators and a common uniform random variable for all 
accept-reject steps. Let $\check{P}_{l,l-1}(z_{l,l-1},\cdot)$ denote the resulting coupled HMC kernel 
on levels $l$ and $l-1$. In addition to the above-mentioned behaviour for the pairs on each level, 
this coupling also induces dependencies between the pairs across levels, 
which are crucial for the conditions in \eqref{eq:main_cond} or \eqref{eq:coup_sum_cond} to hold. 
Analogous to \eqref{eqn:coupled2_hmc}, we take 
\begin{align}\label{eqn:coupled4_hmc}
	\check{K}_{l,l-1}(z_{l,l-1},dz_{l,l-1}') = 
	\kappa\check{P}_{l,l-1}(z_{l,l-1},dz_{l,l-1}') 
	+ (1-\kappa)\check{K}_{l,l-1}^{(M)}(z_{l,l-1},dz_{l,l-1}'),
\end{align}
where $\kappa\in(0,1)$ and $\check{K}_{l,l-1}^{(M)}$ denotes a coupled RWMH kernel based on 
the maximal couplings in Section \ref{sec:4max} or Section \ref{sec:4refmax}. 
We note that the mixture kernel \eqref{eqn:coupled4_hmc} satisfies the properties stated in Section \ref{sec:ub_inc_gen}.  
An algorithmic description of how to sample from it is provided in Algorithm \ref{alg:coupled4_mixturehmc}. 

\begin{algorithm}
\textbf{Input}: current state $z_{l,l-1}=((x_{l},w_l),(x_{l-1},w_{l-1}))\in\mathsf{Z}$, 
leapfrog stepsize $\varepsilon_l>0$, number of leapfrog steps $M_l\in\mathbb{N}$, mixing probability $\kappa\in(0,1)$ 
and coupled RWMH kernel $\check{K}_{l,l-1}^{(M)}$. 
\\

With probability $\kappa$, 
\begin{enumerate}
	\item Sample initial velocity $v\sim\mathcal{N}_d(0_{d},I_{d})$ and $U\sim\mathcal{U}_{[0,1]}$.
	
	\item For $s\in\{l,l-1\}$, run leapfrog integrator \eqref{eqn:leapfrog} on level $s$ with initial condition $(x_s,v)$ to obtain 
	proposal $(x_s',v_s')$.
	
	\item For $s\in\{l,l-1\}$, run leapfrog integrator \eqref{eqn:leapfrog} on level $s$ with initial condition $(w_s,v)$ to obtain 
	proposal $(w_s',\widetilde{v}_s')$.
	
	\item For $s\in\{l,l-1\}$, if $U \leq \bar{\alpha}_{l}\{ (x_s,v),(x_s',v_s') \}$, set $X_s^{\star}=x_s'$; otherwise set $X_s^{\star}=x_s$.
	
	\item For $s\in\{l,l-1\}$, if $U \leq \bar{\alpha}_{l}\{ (w_s,v),(w_s',\widetilde{v}_s') \}$, 
	set $W_s^{\star}=w_s'$; otherwise set $W_s^{\star}=w_s$.
\end{enumerate}
Otherwise, generate $Z_{l,l-1}^{\star}=((X_{l}^{\star},W_l^{\star}),(X_{l-1}^{\star},W_{l-1}^{\star}))$ 
according to $\check{K}_{l,l-1}^{(M)}(z_{l,l-1},\cdot)$.

\textbf{Output}: Sample $Z_{l,l-1}^{\star}=((X_{l}^{\star},W_l^{\star}),(X_{l-1}^{\star},W_{l-1}^{\star}))$ 
from $\check{K}_{l,l-1}(z_{l,l-1},\cdot)$ in \eqref{eqn:coupled4_hmc}.
\caption{Mixture of coupled HMC and coupled RWMH}
\label{alg:coupled4_mixturehmc}
\end{algorithm}

\section{Numerical Results}\label{sec:numerics}

Three numerical examples will be used to illustrate 
the properties of various algorithms and our theoretical results.  
The elliptic PDE problem introduced in Section \ref{sec:mot_ex}
will be considered in Section \ref{sec:bip}. 
In Section \ref{sec:analytical_case}, we begin with a case where an analytical solution of the PDE is tractable. 
Subsequently, in Section \ref{sec:bipfull} a particular example of 
the problem described in Section \ref{sec:mot_ex} is considered.
Finally, we examine a model from epidemiology in Section \ref{sec:epidemic}, 
to analyze COVID-19 infections in the UK.

Two quantities of interest will be used to illustrate our methodology. 
The first is the expected value corresponding to the choice of function $\varphi(x) = x$. 
The next quantity of interest is motivated by the estimation of 
parameters $\theta$, such as the precision of the observation model in \eqref{eq:data}. 
One approach is based on maximizing the marginal likelihood $Z(\theta)$, 
defined as the normalizing constant of \eqref{eq:unnol_exact}. 
We will compute the maximum likelihood estimator (MLE) $\theta_{\rm MLE}\in\arg\max Z(\theta)$ 
by employing a stochastic gradient algorithm \cite{gower, ub_bip,kushner}, given by the iterative scheme
\begin{align}\label{eq:sgd}
	\theta^{(i)} = \theta^{(i-1)} + \alpha_{i} \reallywidehat{\nabla_{\theta}\log Z(\theta^{(i-1)})},\quad i\geq 1,
\end{align}
where $(\alpha_i)_{i\in\mathbb{N}}$ is a sequence of learning rates and 
$\reallywidehat{\nabla_{\theta}\log Z(\theta)}$ denotes an unbiased estimator of 
\begin{align}\label{eqn:score_function}
	\nabla_{\theta}\log Z(\theta) = \int_\mathsf{X} ( \nabla_\theta \log \gamma_\theta(x) )\pi_\theta(x) dx.
\end{align}
Following convergence results in \cite{gower,kushner}, we select $\alpha_i=\alpha_1/i$ and choose $\alpha_1$ appropriately. 
We will rely on our methodology to obtain unbiased estimators of the score function \eqref{eqn:score_function} 
by choosing the function $\varphi_\theta(x) = \nabla_\theta \log \gamma_\theta(x)$.
Finally, to deal with parameters that have positivity constraints, 
we apply a logarithmic transformation before employing \eqref{eq:sgd}. 

\subsection{Elliptic Bayesian Inverse Problem}
\label{sec:bip}

\subsubsection{An Analytically Tractable Case}\label{sec:analytical_case}
We first consider an example where an analytical solution is available. 
The PDE on $\mathsf{C}=[0,2\pi]$ 
is defined by \eqref{eq:pde1} with constant diffusion coefficient $\Phi=1$,  
and forcing $f(t; X) = X_1 \sin(2t) + X_2 \sin(t)$.
The analytical solution is given by $h(t; X) = \frac{1}{4} X_1 \sin(2t) + X_2 \sin(t)$.
Furthermore, we assume a prior of $X \sim \calN_2(0_2, 16 I_2)$ on the state space $\mathsf X = \bbR^2$.
Although this setting extends beyond the theoretical framework we have considered, 
we expect our results to generalize. 
The observation functions \eqref{eqn:observation_function} are given by 
the Dirac delta functions $g_p = \delta_{t_p}$, where $t_p = 2\pi (2 p-1)/2P$ for $p\in\{1,\dots, P\}$ with $P=50$. 
We simulate observations $y\in\mathbb{R}^P$ from \eqref{eq:data} 
using $x=(2,-2)$ and $\theta=100$. 

Given a value of $\theta$ and data $y$, the posterior of $X$ in \eqref{eq:unnol_exact},
denoted as $\pi_\theta$, has the form $\calN_2(\mu_\theta, \Sigma_\theta)$, where 
$$
\mu_\theta = \theta \Sigma_\theta G^T y, \quad \Sigma_\theta^{-1} = \theta G^T G + 16^{-1} I_2. 
$$
In the preceding line, $G \in \bbR^{P\times 2}$ 
is the forward model matrix (such that $G(x) = G x$)
with entries $G_{p,1} = \frac{1}{4} \sin(2t_p)$ and $G_{p,2} = \sin(t_p)$ 
for $p\in\{1,\dots, P\}$. 
The above quantities of interest are also analytically tractable and used as ground truth.  
Firstly, for $\varphi(x)=x$, the expected value is $\pi_\theta(\varphi) = \mu_\theta$. 
Secondly, the score function \eqref{eqn:score_function} can be computed using the 
fact that the marginal likelihood satisfies $Z(\theta)=\phi_2(y;0_P,  16 GG^T + \theta^{-1} I_P)$. 

\begin{figure}[!htbp]
	\centering\includegraphics[width=0.49\textwidth]{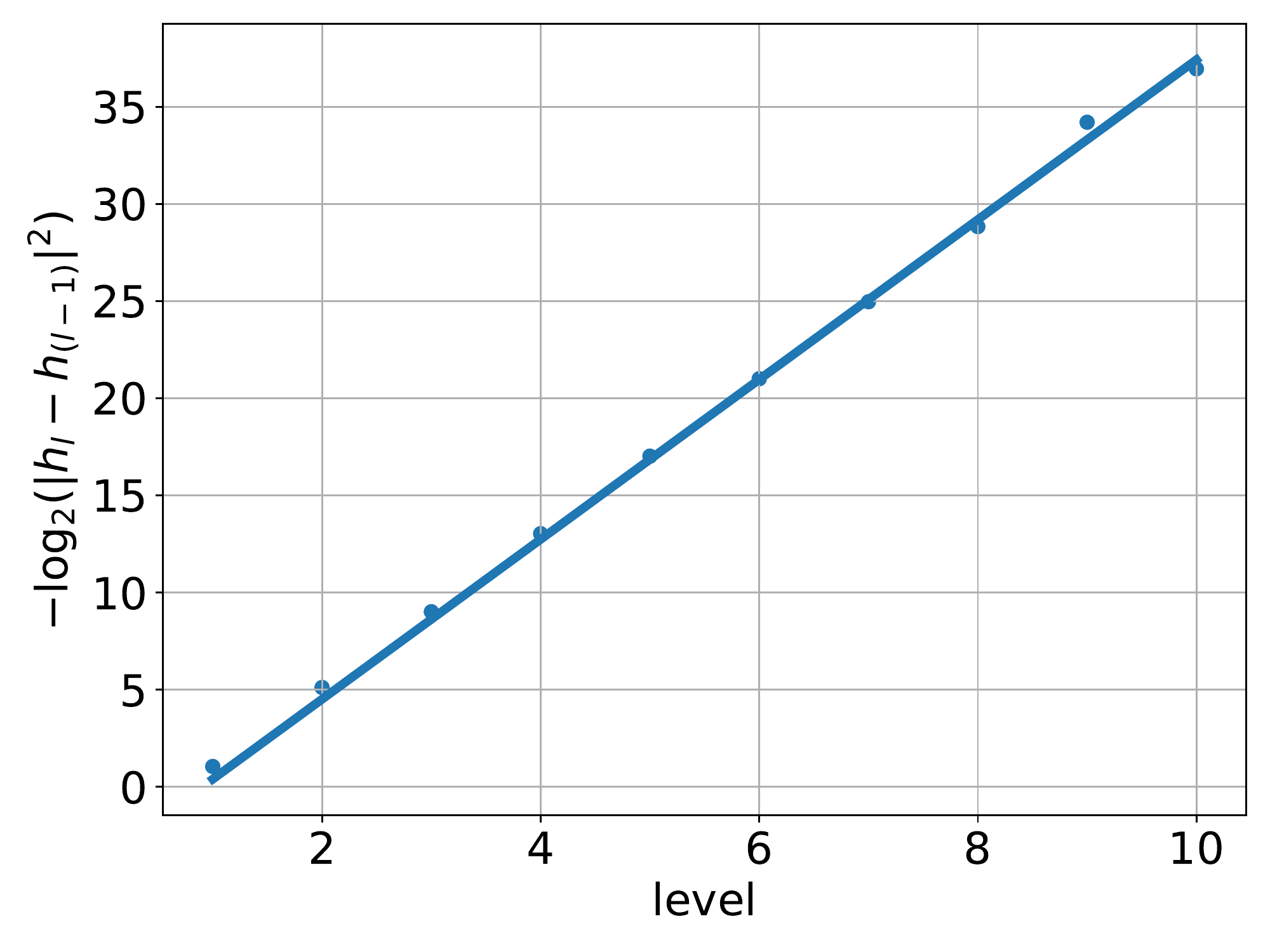}
	\includegraphics[width=0.49\textwidth]{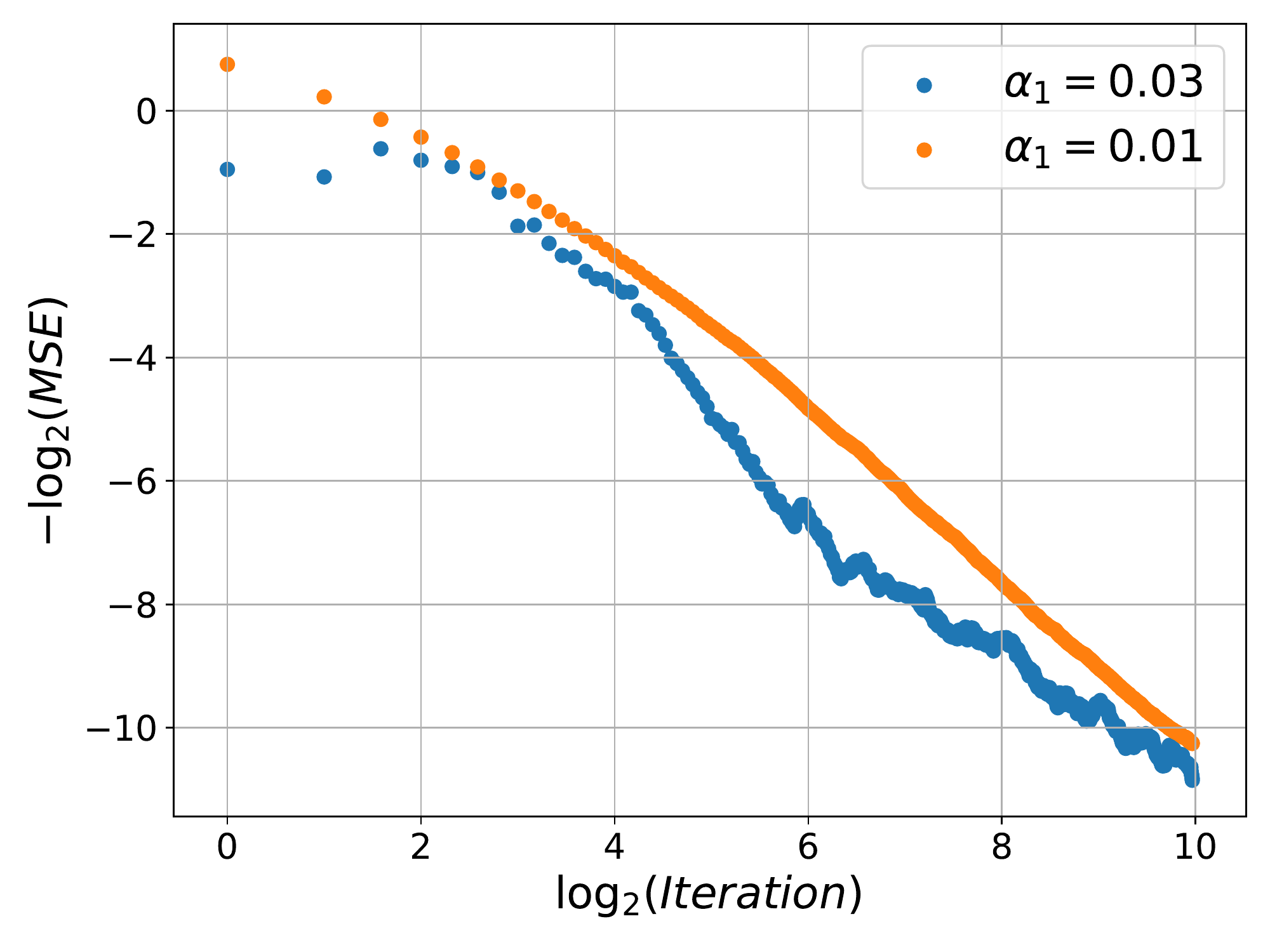}
	\caption{Toy model of Section \ref{sec:analytical_case}. 
	Left: error of forward model approximation $\|h_l - h_{l-1}\|_2^2$ against discretization level $l$  
	satisfies $\|h_l - h_{l-1}\|_2^2 \leq C \Delta_l^{2\beta}$ with a rate of $\beta=2$. 
	Right: convergence of stochastic gradient iterates $(\theta^{(i)})_{i\in\mathbb{Z}^+}$, 
	defined in \eqref{eq:sgd}, to the maximum likelihood estimator $\theta_{\rm MLE}$, 
	measured in terms of the mean squared error $\bbE[(\theta^{(i)}-\theta_{\rm MLE})^2]$ 
	that was estimated using $100$ independent realizations. 
	The learning rates considered here are $\alpha_i=\alpha_1/i$ with $\alpha_1\in\{0.01,0.03\}$. 
	}
	\label{fig:toy_conv}
\end{figure}

\begin{figure}[!htbp]
	\centering
	\includegraphics[width=0.49\textwidth]{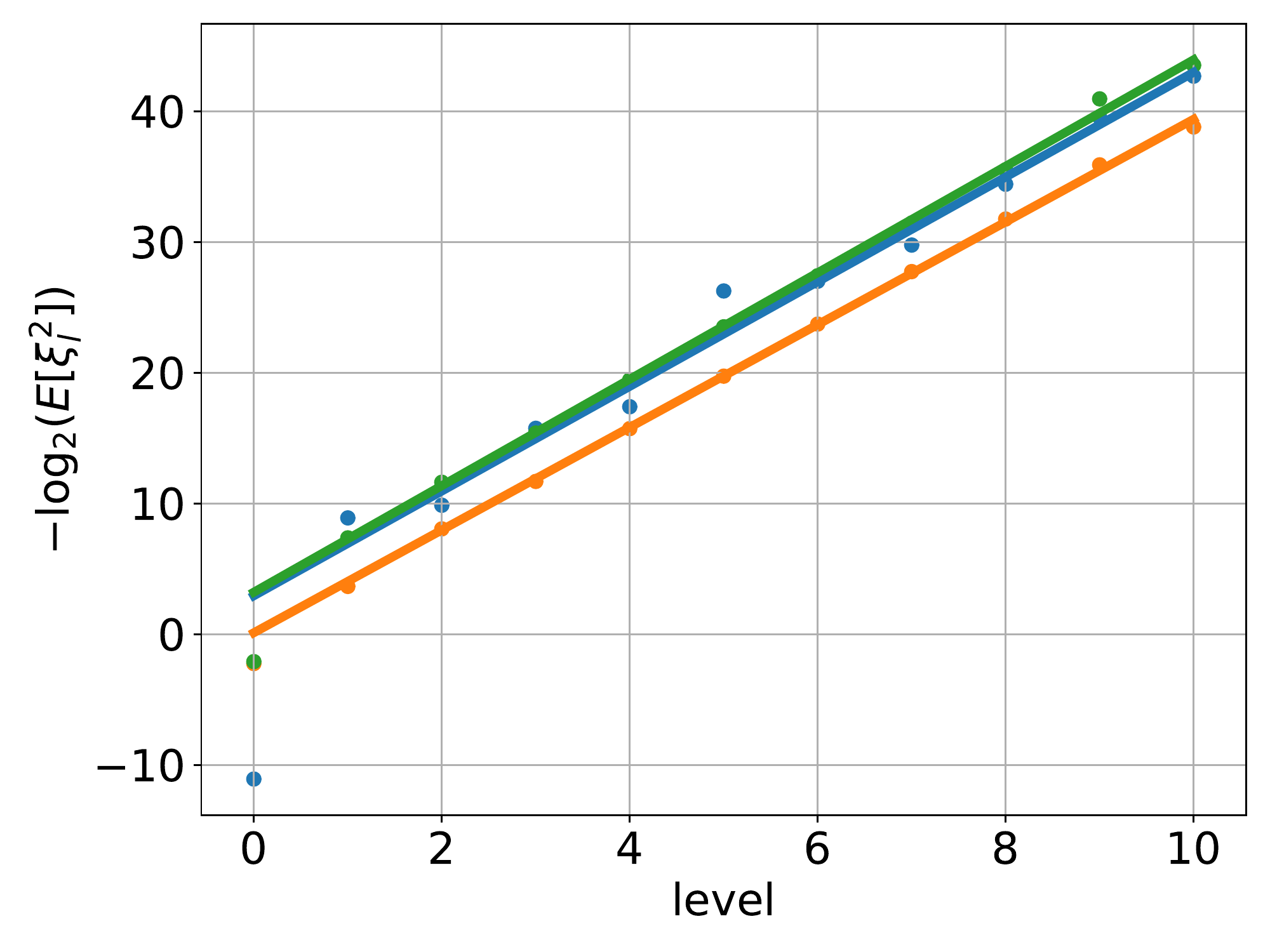}
	\includegraphics[width=0.49\textwidth]{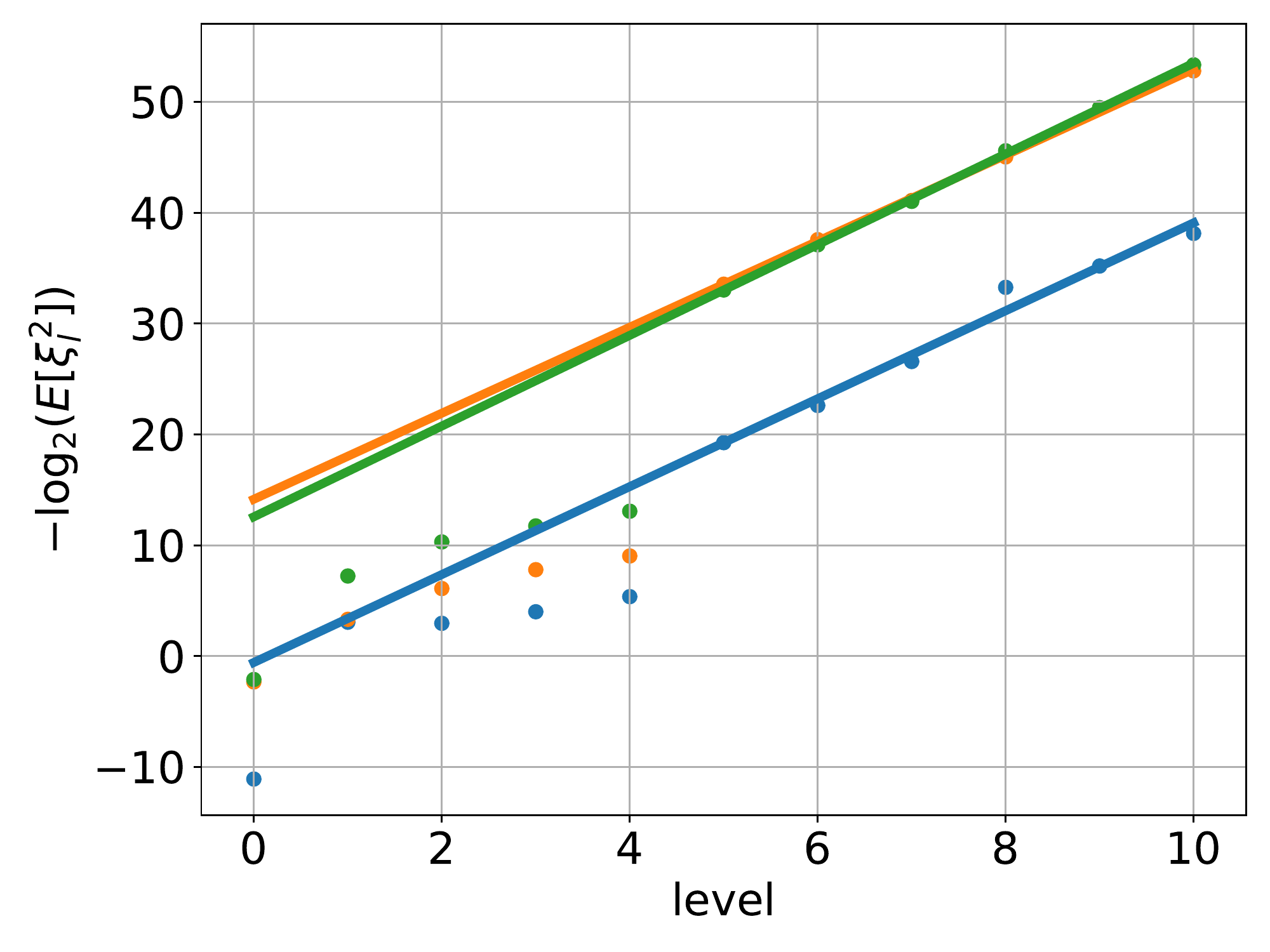}
	\caption{Toy model of Section \ref{sec:analytical_case} at data generating parameter $\theta=1$. 
	Second moment of the time-averaged estimator $\xi_l = \reallywidehat{[\pi_l-\pi_{l-1}](\varphi)}_{T,k,m}$ 
	in \eqref{eq:time_ave_inc} with $k=100$ and $m=1000$ 
	against discretization level $l$ 
	using either the mixture of coupled RWMH and HMC kernels in \eqref{eqn:coupled4_hmc} (\emph{left}), 
	or the coupled pCN kernel based on the reflection maximal coupling of Section \ref{sec:4refmax} 
	(\emph{right}). 
	Functions considered here are $\varphi(x) = \nabla_\theta \log \gamma_\theta(x)$ (\emph{blue}), 
	$\varphi(x)=x_1$ (\emph{orange}) and $\varphi(x)=x_2$ (\emph{green}). 
	In both cases, we have $\bbE [\xi_l^2] \leq C \Delta_l^{2\beta}$ with a rate of $\beta=2$.
	The second moment was estimated using $100$ independent realizations. 
	}
	\label{fig:toy_mse}
\end{figure}

\begin{figure}[!htbp]
	\centering	
	\includegraphics[width=0.49\textwidth]{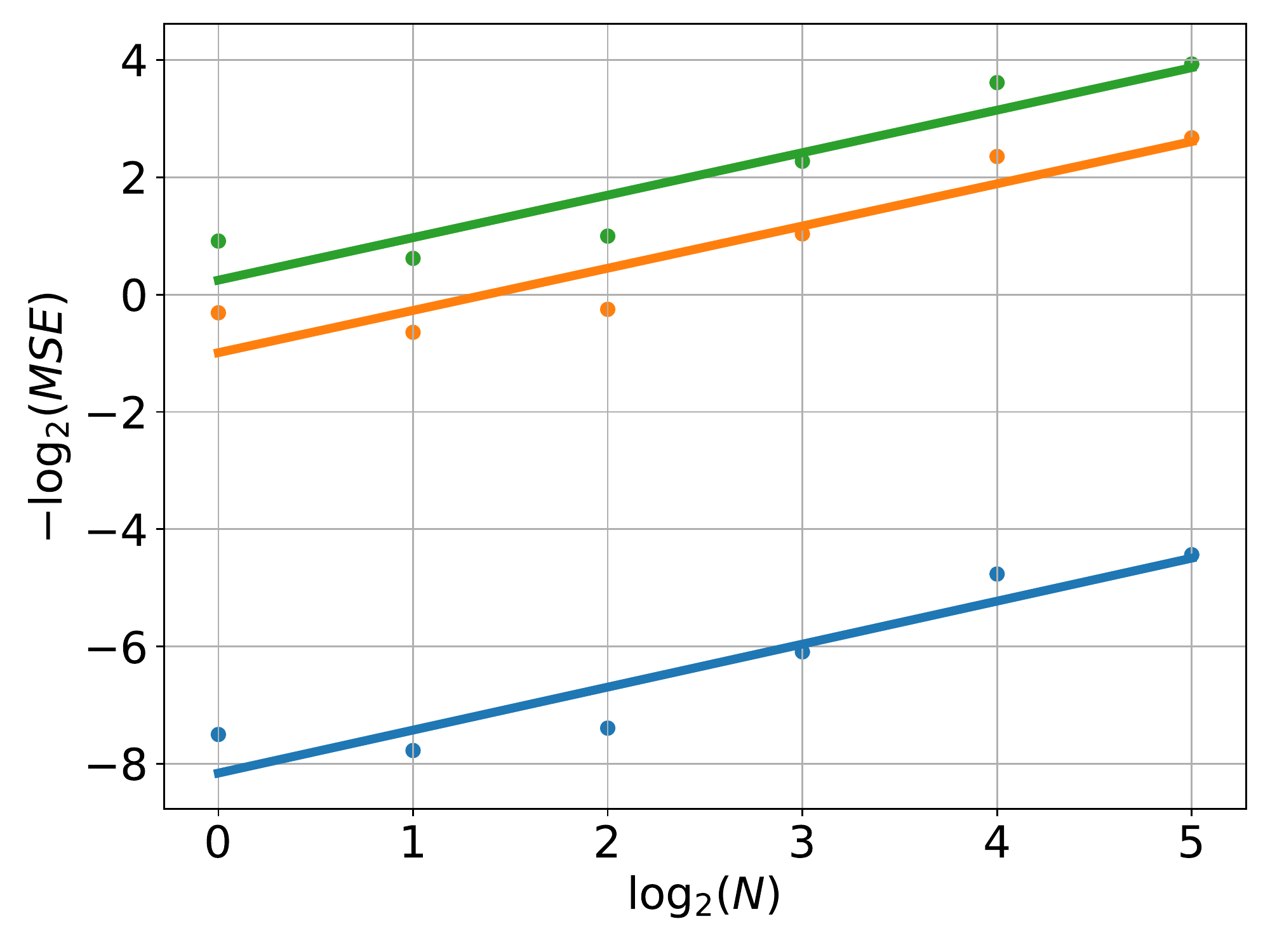}
	\includegraphics[width=0.49\textwidth]{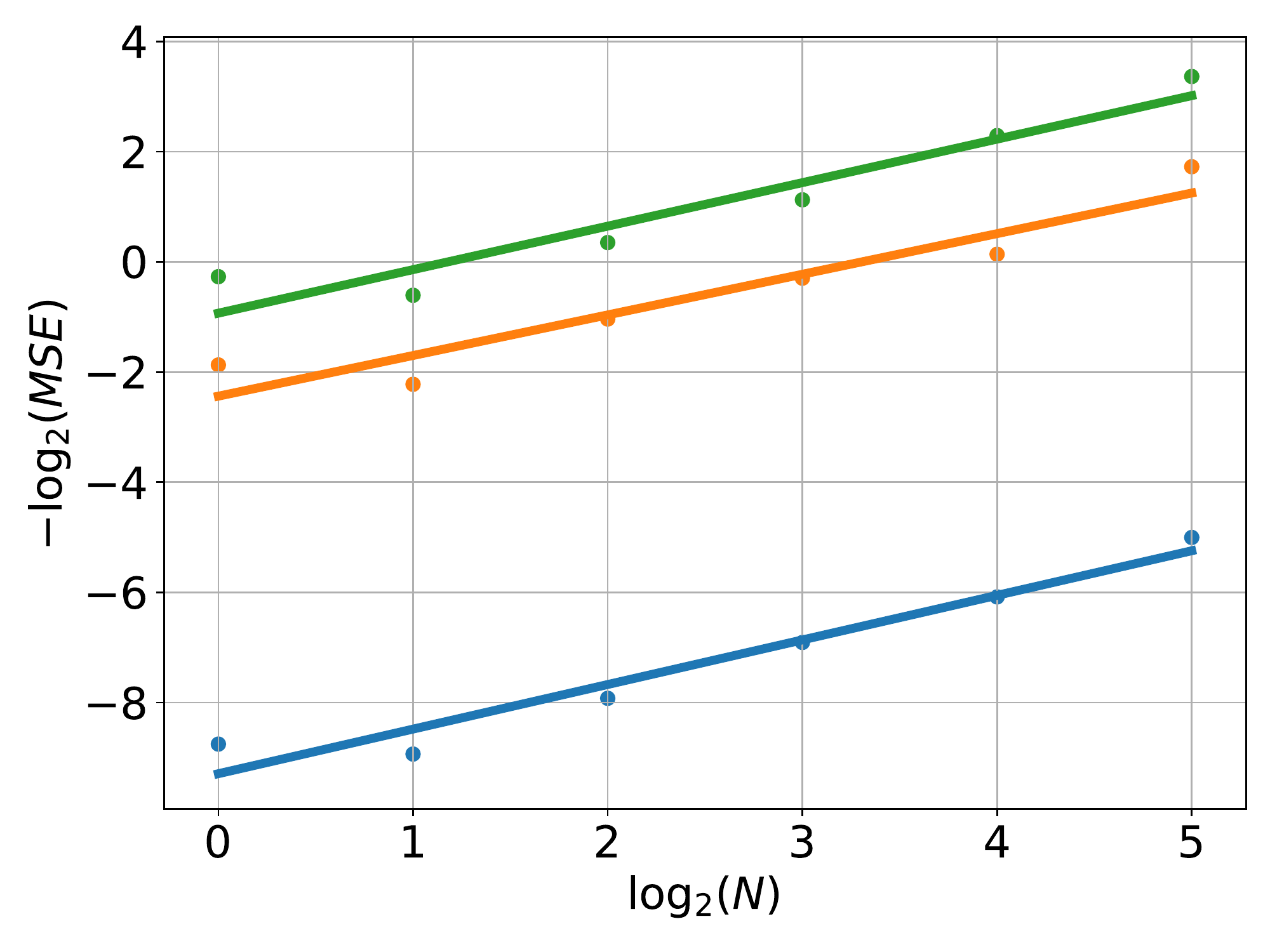}
	\caption{Toy model of Section \ref{sec:analytical_case} at data generating parameter $\theta=1$. 
	Mean squared error of averaged estimator $\bbE[(\reallywidehat{\pi(\varphi)}(N) - \pi(\varphi))^2]$
	against number of single term replicates $N$ using 
	either the mixture of coupled RWMH and HMC kernels in \eqref{eqn:coupled4_hmc} (\emph{left}), 
	or the coupled pCN kernel based on the reflection maximal coupling of Section \ref{sec:4refmax} 
	(\emph{right}). 
	Functions considered here are $\varphi(x) = \nabla_\theta \log \gamma_\theta(x)$ (\emph{blue}), 
	$\varphi(x)=x_1$ (\emph{orange}) and $\varphi(x)=x_2$ (\emph{green}). 	
	In both cases, we have the standard Monte Carlo rate 
	$\bbE[(\reallywidehat{\pi(\varphi)}(N) - \pi(\varphi))^2]\leq CN^{-1}$. 
	The mean squared error was estimated using the exact $\pi(\varphi)$ as ground truth and 
	$160$ independent realizations.}
	\label{fig:extraconv}
\end{figure}

To suit the domain under consideration, we take the mesh width of the FEM scheme 
in Section \ref{ssec:example} as $\Delta_l = 2 \pi \times 2^{-(l + l_0)}$ with $l_0=5$. 
Firstly, in the left panel of Figure \ref{fig:toy_conv}, 
we numerically verify that our approximation of the forward model indeed 
converges at the rate of $\beta = 2$. 
Next, we consider the time-averaged estimator 
$\xi_l = \reallywidehat{[\pi_l-\pi_{l-1}](\varphi)}_{T,k,m}$ in \eqref{eq:time_ave_inc} 
with $k=100$ and $m=1000$, 
and examine the rate at which its second moment converges to zero as $l$ increases 
in Figure \ref{fig:toy_mse}. 
These estimators are computed using a Markov chain $(Z_{n,l,l-1})_{n\in\mathbb{Z}^+}$ 
that is simulated using either the mixture of coupled RWMH and HMC kernels 
(Algorithms \ref{alg:coupled4_reflectionmaximal} and \ref{alg:coupled4_mixturehmc}) 
in \eqref{eqn:coupled4_hmc} (left panel), 
or a coupled pCN kernel based on the reflection maximal coupling of Section \ref{sec:4refmax} (right panel). 
The algorithmic settings of \eqref{eqn:coupled4_hmc} include a mixing probability of $\kappa=0.9$, 
stepsize of $\varepsilon_l = 0.1$ and $M_l = 10$ leapfrog steps for all $l\in\mathbb{Z}^+$ in the HMC kernels;  
and proposal covariance of $10^{-8}I_2$ for all $l\in\mathbb{Z}^+$ in the RWMH kernels. 
For pCN kernels, we took $\rho_l=0.95$ and $\sigma_l=4.0I_2$ for all $l\in\mathbb{Z}^+$.
To satisfy Assumption A5, we initialize $Z_{0,l,l-1}=((X_{0,l},W_{0,l}),(X_{0,l-1}, W_{0,l-1}))$ from 
a coupling $\check{\nu}_{l,l-1}$ that can be described by the following steps: 
1) sample $X_{0,l-1}'$ and $W_{0,l-1}$ from the prior $\calN_2(0, 16 I_2)$ independently;  
2) generate $X_{0,l}'|X_{0,l-1}' \sim \mathcal{N}_2(X_{0,l-1}', 2^{-(2l+1)}I_2)$ and 
$W_{0,l}|W_{0,l-1} \sim \mathcal{N}_2(W_{0,l-1}, 2^{-(2l+1)}I_2)$ independently; 
3) generate $(X_{0,l},X_{0,l-1})|(X_{0,l}',X_{0,l-1}')$ according to $K_{l,l-1}((X_{0,l}',X_{0,l-1}'),\cdot)$, the 
marginal kernel on $\mathsf{X}\times\mathsf{X}$ induced by $\check{K}_{l,l-1}$ for a pair across levels.

We can infer from both plots in Figure \ref{fig:toy_mse} that 
$\bbE [\xi_l^2] \leq C \Delta_l^{2\beta}$ with a rate of $\beta=2$, which matches that of the forward model approximation.  
Hence the condition in \eqref{eq:main_cond} ensuring unbiased and finite variance properties 
of the single term estimator $\reallywidehat{\pi(\varphi)}_{S,T,k,m}$ in \eqref{eq:basic_ub_est_time} 
can be verified. 
As the cost of the marginal kernel at level $l$ is of order $\Delta_l^{-\omega}$ with $\omega=1$, 
following the discussion in Section \ref{ssec:pl}, 
we select $\bbP_L(l) \propto \Delta_l^{\eta}$ with $\eta=5/2$ to ensure finite expected cost. 
Figure \ref{fig:extraconv} shows that by averaging $N\in\mathbb{N}$ independent replicates of the single term estimator, 
we obtain an unbiased estimator $\reallywidehat{\pi(\varphi)}(N)$ in \eqref{eqn:averaging_singleterm} 
that satisfies the standard Monte Carlo rate as $N\rightarrow\infty$. 

Lastly, in the right panel of Figure \ref{fig:toy_conv}, we illustrate convergence of the stochastic gradient algorithm 
\eqref{eq:sgd}, initialized at $\theta^{(0)}=1$, to the maximum likelihood estimator $\theta_{\rm MLE}$, for two sequences of learning rates.  
The MLE was computed numerically by maximizing the marginal likelihood $Z(\theta)$ 
with the exact score function \eqref{eqn:score_function}. 

\subsubsection{Example of Section \ref{sec:mot_ex}} 
\label{sec:bipfull}

The general case of Section \ref{sec:mot_ex} is now considered,
with unknown diffusion coefficient $\Phi$ and forcing $f(t)=100t$. 
The prior specification of $X=(X_1,X_2)$ is taken as $d=2$,
$\bar{\Phi}=0.15$, $\vartheta_1=1/10$, $\vartheta_2=1/40$, 
$v_1(t)=\sin(\pi t)$ and $v_2(t)=\cos(2\pi t)$. 
For this particular setting, the solution $h$ is continuous and 
hence pointwise observations are well-defined. 
The observation function $G(x)$ in \eqref{eqn:observation_function} 
is chosen as $g_p(h(X))=h(0.01+0.02(p-1);X)$ for $p\in\{1,\ldots,P\}$ with $P=50$. 
We employ the FEM scheme in Section \ref{ssec:example} with mesh width of 
$\Delta_l = 2^{-(l + l_0)}$ where $l_0=3$. 
Using a discretization level of $l=10$ to approximate $G(x)$ with $G_l(x)$, 
$x=(0.6,-0.4)$ and $\theta=1$,  
we simulate observations $y\in\mathbb{R}^P$ from \eqref{eq:data}. 

Figure \ref{fig:nl_conv} shows that the forward model approximation 
and the second moment of 
time-averaged estimator 
$\xi_l = \reallywidehat{[\pi_l-\pi_{l-1}](\varphi)}_{T,k,m}$ in \eqref{eq:time_ave_inc} 
with $k=100$ and $m=1000$ 
converges at the same rate of $\beta=2$ as $l$ increases. 
The estimators $(\xi_l)$ are computed using the reflection maximal coupling of 
pCN kernels in Section \ref{sec:4refmax}, with algorithmic parameters of 
$\rho_l=0.95$ and $\sigma_l=I_d$ for all $l\in\mathbb{Z}^+$. 
The Markov chain is initialized in a similar manner to Section \ref{sec:analytical_case}, 
with the exception of having $\mathcal{U}[-1,1]^2$ as the prior in this case. 
As before, we can select $\bbP_L(l) \propto \Delta_l^{\eta}$ with  $\eta=5/2$ to ensure that 
the single term estimator 
$\reallywidehat{\pi(\varphi)}_{S,T,k,m}$ in \eqref{eq:basic_ub_est_time} 
has finite variance and finite expected cost. 
The left panel of Figure \ref{fig:nl_mse} illustrates that averaging single term estimators 
yields a consistent estimator that converges at the standard Monte Carlo rate. 

Finally, we consider inference for $\theta$ in the Bayesian framework, 
under a prior $p(\theta)$ specified as a standard Gaussian prior on $\log \theta$. 
By adding the gradient of the prior density $\nabla_{\theta}\log p(\theta)$ to \eqref{eq:sgd}, 
we can run a stochastic gradient algorithm initialized at $\theta^{(0)}=0.1$ to compute 
the maximum a posteriori probability (MAP) estimator $\theta_{\rm MAP}\in \arg\max p(\theta)Z(\theta)$. 
The left panel of Figure \ref{fig:nl_mse} displays convergence of the stochastic iterates 
to $\theta_{\rm MAP}$. 
As competing algorithm, we consider the approach of \cite{ub_bip} that can also compute
unbiased estimators of the score function \eqref{eqn:score_function} 
using the algorithm in \cite{beskos2} instead of MCMC. 
The plot shows some gains over \cite{ub_bip} when the same learning rates are employed. 

\begin{figure}[!htbp]
	\centering\includegraphics[width=0.49\textwidth]{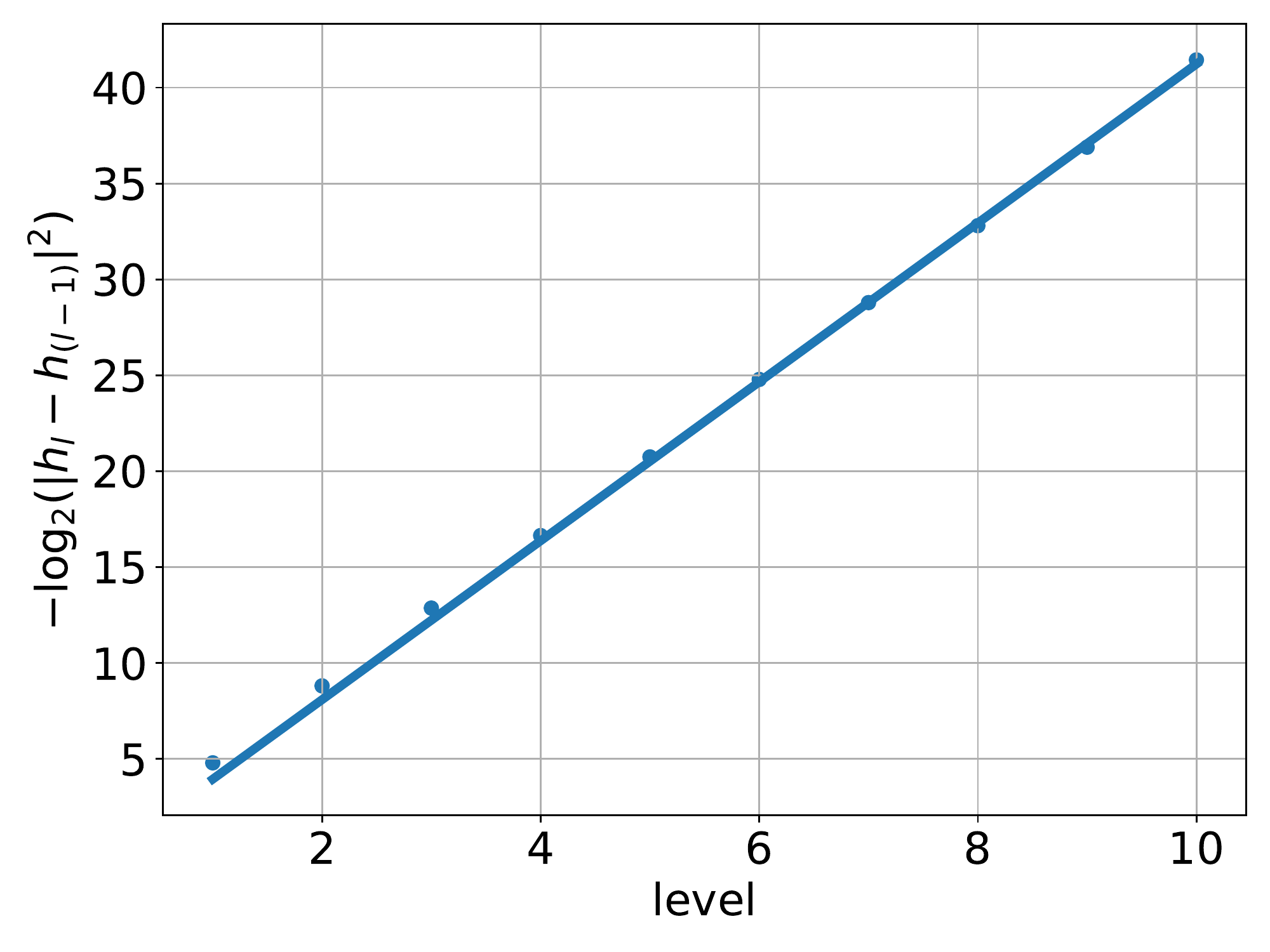}
	\includegraphics[width=0.49\textwidth]{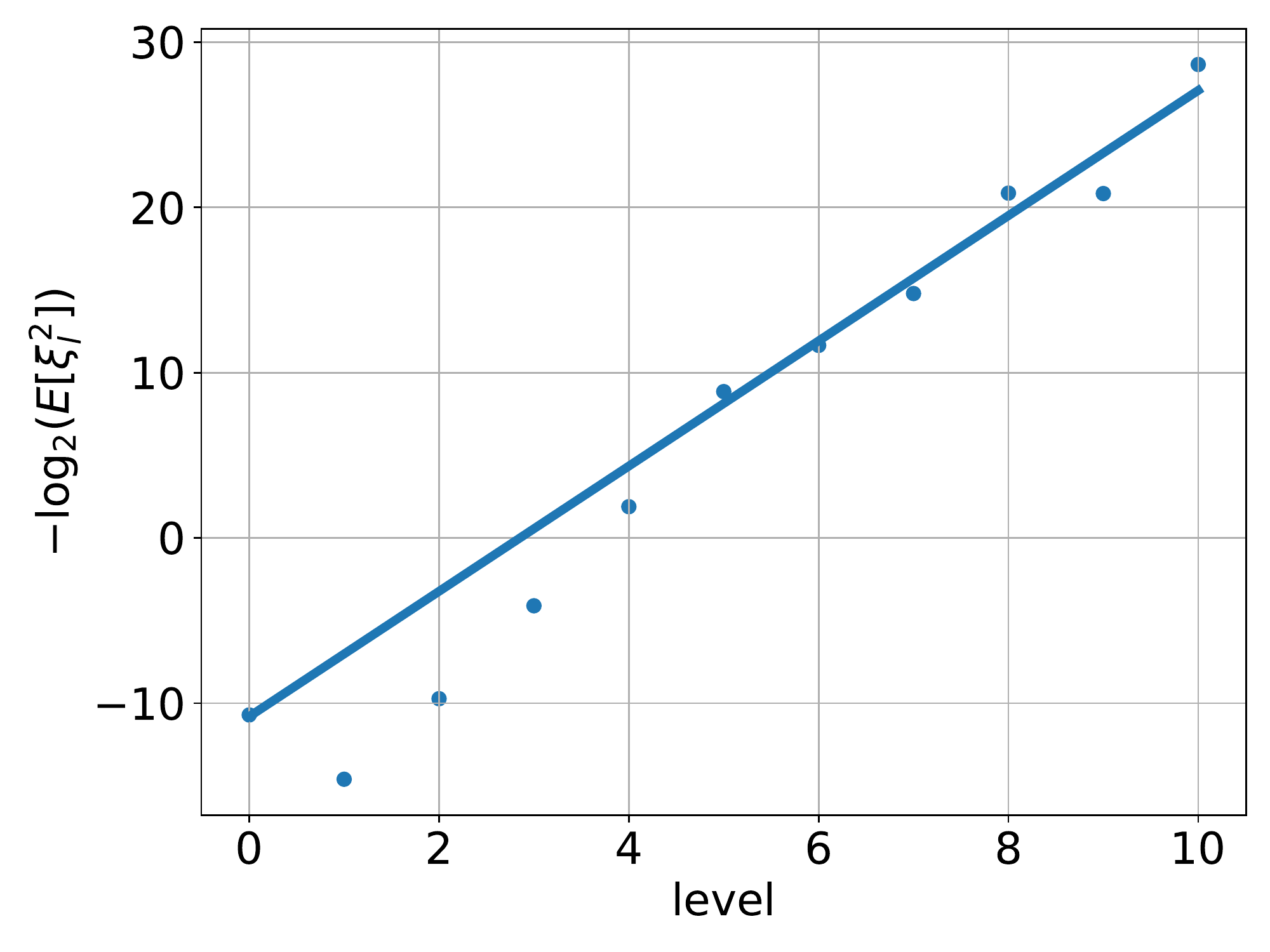}
	\caption{Elliptic Bayesian inverse problem of Sections \ref{sec:mot_ex} and \ref{sec:bipfull}. 
	Left: error of forward model approximation $\|h_l - h_{l-1}\|_2^2$ against discretization level $l$  
	satisfies $\|h_l - h_{l-1}\|_2^2 \leq C \Delta_l^{2\beta}$ with a rate of $\beta=2$. 
	Right: second moment of the time-averaged estimator  
	$\xi_l = \reallywidehat{[\pi_l-\pi_{l-1}](\varphi)}_{T,k,m}$ in \eqref{eq:time_ave_inc} 
	with $k=100$ and $m=1000$ 
	against discretization level $l$ satisfies $\bbE [\xi_l^2] \leq C \Delta_l^{2\beta}$ with a rate of $\beta=2$.
	The function considered here is $\varphi(x) = \nabla_\theta \log \gamma_\theta(x)$ with $\theta=0.1$.  
	The second moment was estimated using $100$ independent realizations. 
	}
	\label{fig:nl_conv}
\end{figure}

\begin{figure}[!htbp]
	\centering\includegraphics[width=0.49\textwidth]{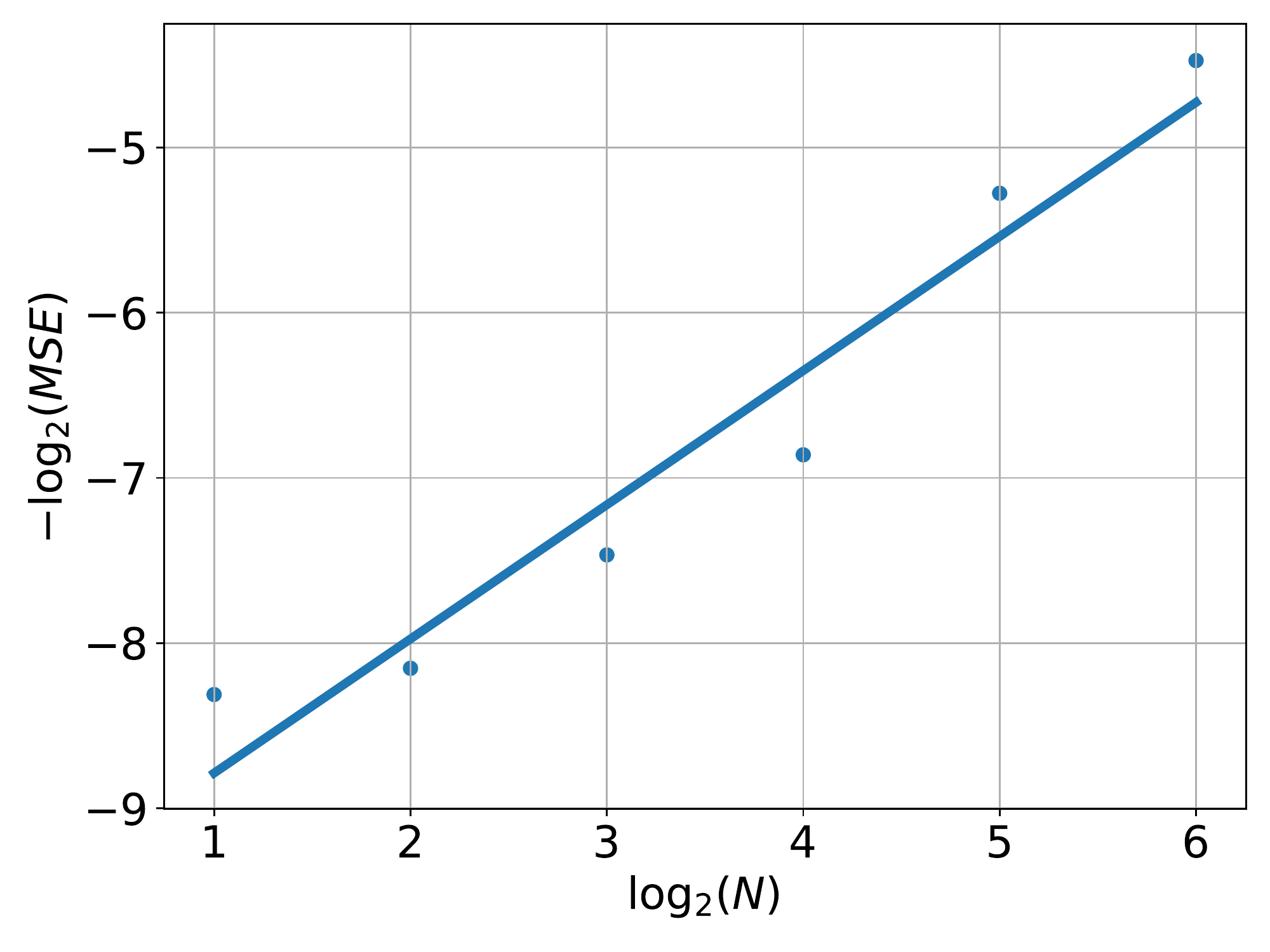}
	\includegraphics[width=0.49\textwidth]{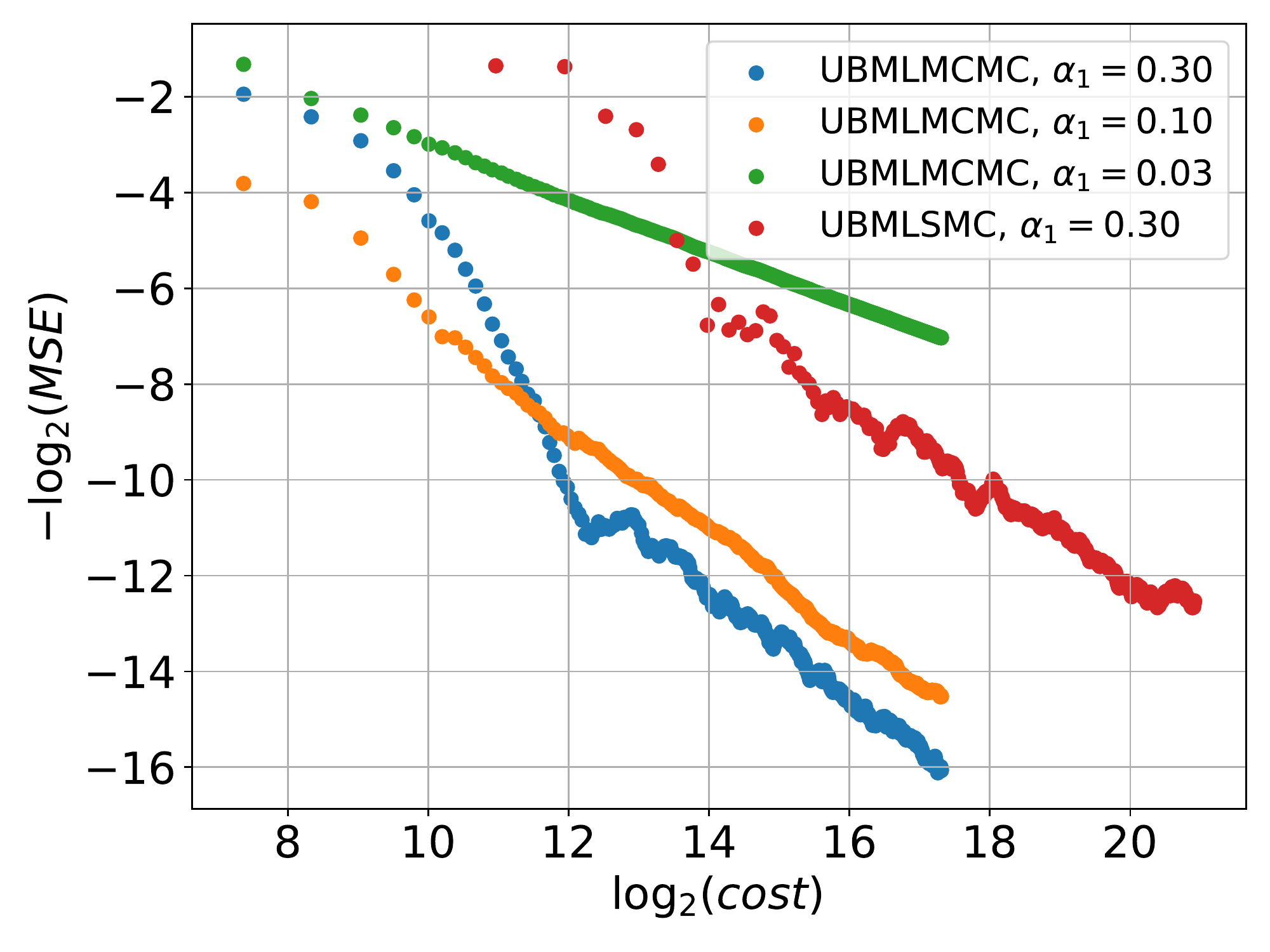}
	\caption{Elliptic Bayesian inverse problem of Sections \ref{sec:mot_ex} and \ref{sec:bipfull}. 
	Left: mean squared error of averaged estimator $\bbE[(\reallywidehat{\pi(\varphi)}(N) - \pi(\varphi))^2]$
	against number of single term replicates $N$ for the  
	function $\varphi(x) = \nabla_\theta \log \gamma_\theta(x)$ with $\theta = 0.3$.  
	The mean squared error was estimated using 
	an average of $10,000$ independent replicates with a minimal discretization level of $l_0=8$ 
	as ground truth for $\pi(\varphi)$, and $160$ independent realizations. 
	Right: convergence of stochastic gradient iterates $(\theta^{(i)})_{i\in\mathbb{Z}^+}$ 
	to the maximum a posteriori probability estimator $\theta_{\rm MAP}$.   
	The mean squared error $\bbE[(\theta^{(i)}-\theta_{\rm MAP})^2]$ was 
	estimated using a longer run of the stochastic gradient algorithm as ground truth for $\theta_{\rm MAP}$, 
	and $160$ independent realizations. 
	The learning rates considered here are $\alpha_i=\alpha_1/i$ with $\alpha_1\in\{0.03,0.10,0.30\}$. 	
	The red curve corresponds to a comparison with the unbiased MLSMC algorithm of \cite{ub_bip}. 	
	Here the cost is measured in terms of the number of forward model simulations, 
	weighted by the complexity associated to each discretization level, 
	and averaged over realizations.}
	\label{fig:nl_mse}
\end{figure}

\subsection{A Compartmental Model for COVID-19 in the UK}
\label{sec:epidemic}
Our final application concerns parameter inference for an epidemiological model to 
analyze COVID-19 infections in the UK. 
We consider a recently developed compartmental model \cite{sirx} for a closed population where 
$S(t)$ denotes the proportion that is susceptible to the disease, 
$I(t)$ denotes the proportion of infected individuals, 
$R(t)$ denotes the proportion of individuals who have recovered and are no longer part 
of the transmission process, 
and $\Xi(t)$ denotes the proportion of symptomatic and infected individuals who have been quarantined. 
The model dynamics are governed by the following system of ordinary differential equations
\begin{align}\label{eqn:ode_epi}
	&\frac{d}{dt}S(t) = -a S(t)I(t) - x_1 S(t), \quad 
	\frac{d}{dt}I(t) = a S(t)I(t) - (b+x_1+x_2) I(t),\\
	&\frac{d}{dt}R(t) = bI(t) + x_1 S(t), \quad 
	\frac{d}{dt}\Xi(t) = (x_1 + x_2) I(t).\notag
\end{align}
A unit of time in the model will represent the duration of a day. 
In Equation \eqref{eqn:ode_epi}, $a>0$ and $b>0$ denote the transmission rate and recovery rate, 
respectively. 
Following \cite{sirx}, we adopt the values $a = 0.775$ and $b = 0.125$ for COVID-19.
The parameter $x_1>0$ captures public containment policies or individual behavioural changes 
in response to the epidemic. 
Quarantine measures for symptomatic and infected individuals are described by 
the parameter $x_2>0$.  
We will also infer the time lapsed between the first infection and its reporting $x_3>0$. 
Letting time $t=0$ correspond to the first reported case on January 24, 2020, the initial condition 
is $(S(-x_3),I(-x_3),R(-x_3),\Xi(-x_3))=(1-1/N_{\rm pop}, 1/N_{\rm pop}, 0, 0)$, where 
$N_{\rm pop}= 66,650,000$ denotes the size of the UK population. 
Given parameters $x=(x_1,x_2,x_3)$, we will write the solution of \eqref{eqn:ode_epi} 
at time $t$ as $(S(t;x),I(t;x),R(t;x),\Xi(t;x))$. 
As our prior specification is given by $X_1\sim\cU_{[0.001,0.003]}$, $X_2\sim\cU_{[0.2,0.4]}$ and 
$X_3\sim\cU_{[5,25]}$ independently,  
we will work on the state-space $\mathsf X = [0.001,0.003] \times [0.2,0.4] \times [5,25]$.

To account for under-reporting, the observed proportion of daily confirmed cases $(Y_i)_{i=1}^P$ is modelled as 
\begin{align}\label{eq:obsmod}
	\log (Y_i) = \log (G_i(x)) - \Gamma_i,
\end{align}
for $i\in\{1,\ldots,P\}$, where 
\begin{align}\label{eqn:epi_observation_func}
	G_i(x) = a\int_{n-1+i}^{n+1} S(t;x)I(t;x)dt,
\end{align}
denotes the number of daily new infections under model \eqref{eqn:ode_epi}, 
and $(\Gamma_i)_{i=1}^P$ are independent gamma random variables with shape 
parameter $\theta_1>0$ and scale parameter $\theta_2>0$. 
We set $n=29$ to consider only $P=40$ observations $y = (y_i)_{i=1}^P$ 
from February 12, 2020, as earlier data seem to be unreliable. 
We note that the observation model in \eqref{eq:obsmod} differs from \cite{sirx} which 
adopted a least squares approach to infer parameters. 
Under the gamma likelihood, the unnormalized posterior density of $X$ given $y$ and $\theta=(\theta_1,\theta_2)$ is
\begin{align}\label{eqn:epi_posterior}
	\gamma_{\theta}(x)=\prod_{i=1}^P\frac{1}{\Gamma(\theta_1)\theta_2^{\theta_1}}
	(\log (G_i(x)/y_i))^{\theta_1-1} \exp(-\log (G_i(x)/y_i)/\theta_2)\mathbb{I}_{\mathsf{A}}(x),
\end{align}
where $\mathsf{A} = \{x \in \mathsf{X}: G_i(x) \geq y_i, i=1,\dots,P\}$. 

Any practical implementation of MCMC targeting \eqref{eqn:epi_posterior} 
would require an approximation of $G_i(x)$ in \eqref{eqn:epi_observation_func}. 
As it suffices to approximate $h(t;x)$ satisfying $(d/dt)h(t) = a S(t)I(t)$, 
we augment the system in \eqref{eqn:ode_epi} and 
employ a fourth-order Runge--Kutta numerical integrator \cite{rk} with 
stepsize $\Delta_l=0.1\times 2^{-1}$, for $l\in\mathbb{Z}^+$. 
The left panel of Figure \ref{fig:sir_conv} shows that 
the resulting approximation $h_l(t;x)$ converges to $h(t;x)$ at the expected rate of $\beta=4$.  
We can now apply our proposed methodology to approximate expectations. 
The right panel shows that the second moment of the time-averaged estimator 
$\xi_l = \reallywidehat{[\pi_l-\pi_{l-1}](\varphi)}_{T,k,m}$ in \eqref{eq:time_ave_inc}  
with $k=200$ and $m=2000$ also converges at the same rate. 
To compute the estimators $(\xi_l)$, we used the reflection maximal coupling of 
pCN kernels in Section \ref{sec:4refmax}, with algorithmic parameters of 
$\rho_l=0.95$ and $\sigma_l=I_d$ for all $l\in\mathbb{Z}^+$. 
The Markov chain is initialized in a similar manner to Section \ref{sec:analytical_case}, 
with the exception of truncating the prior to the subset $\mathsf{A}$ in this case. 
Since the cost of the marginal pCN kernel at level $l$ is of order $\Delta_l^{-\omega}$ with $\omega=1$, 
we choose $\bbP_L(l) \propto \Delta_l^{\eta}$ with $\eta=9/2$ to ensure that the single term estimator 
$\reallywidehat{\pi(\varphi)}_{S,T,k,m}$ in \eqref{eq:basic_ub_est_time} 
has finite variance and finite expected cost. 
The left panel of Figure \ref{fig:sir_mse} illustrates the impact of averaging independent replicates. 
As before, we compute the MLE of $\theta$ using the stochastic gradient algorithm \eqref{eq:sgd}, 
with initialization from $\theta^{(0)}=(1,1)$. 
Finally, we exploit the fitted model to infer the extent of under-reporting during 
the time period under consideration. 
In Figure \ref{fig:punchline}, we display the ratio of the total number of reported cases 
to the expected number of total infections under the posterior distribution \eqref{eqn:ode_epi} 
with $\theta=\theta_{\rm MLE}$.

\begin{figure}[!htbp]
	\centering\includegraphics[width=0.49\textwidth]{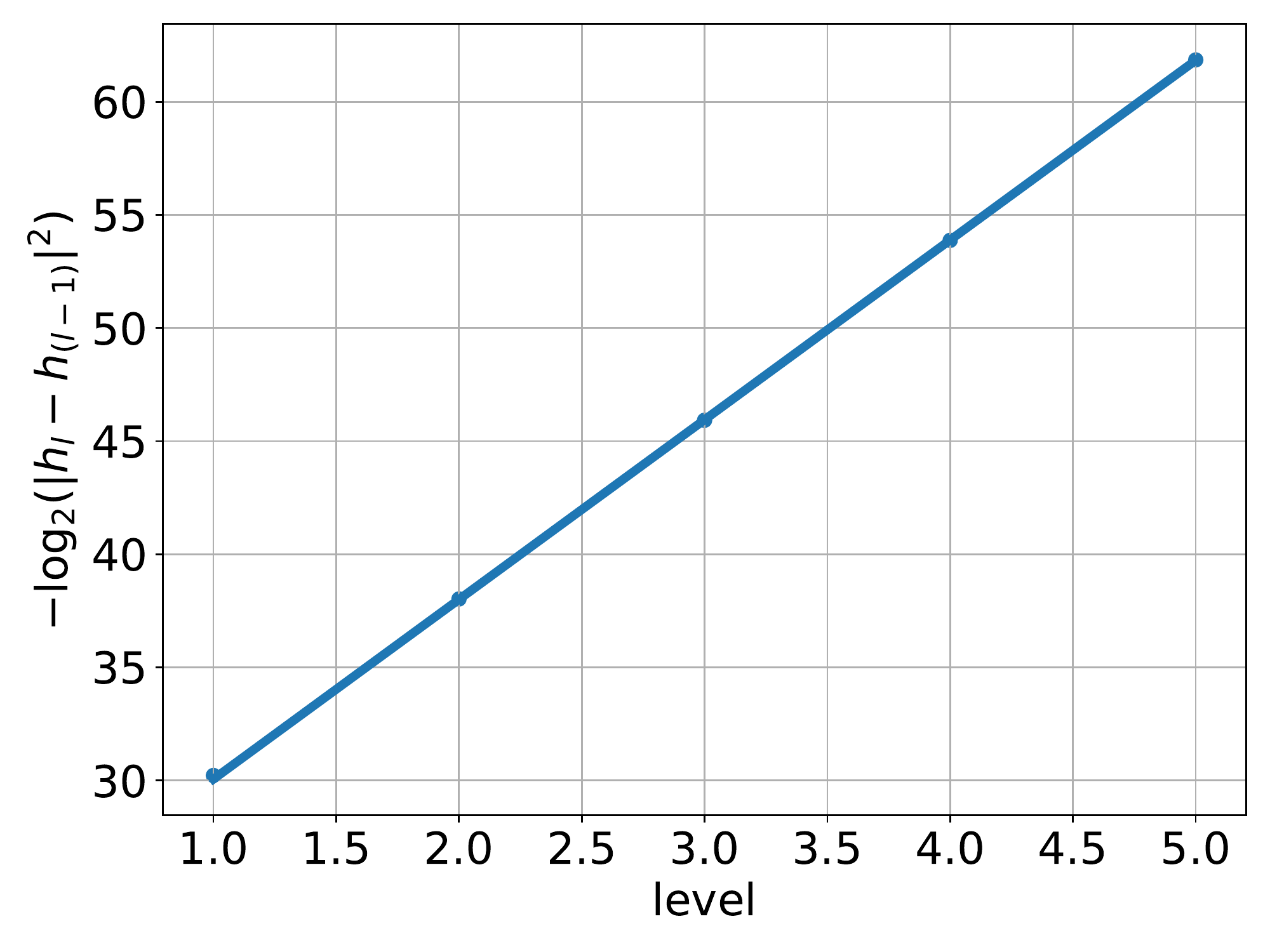}
	\includegraphics[width=0.49\textwidth]{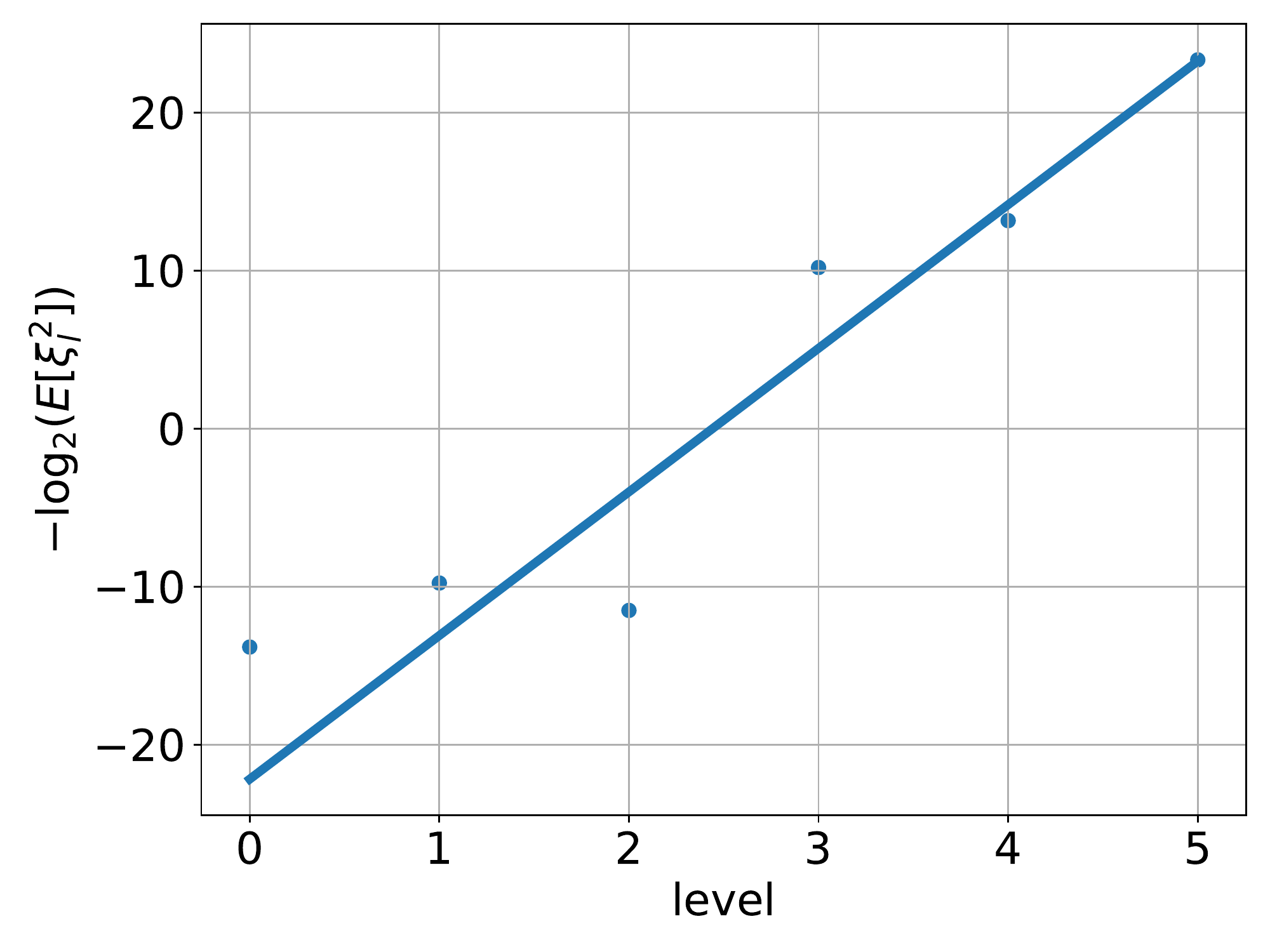}
	\caption{Compartmental model of Section \ref{sec:epidemic}. 
	Left: error of forward model approximation $|h_l(t;x) - h_{l-1}(t;x)|^2$ against discretization level $l$  
	satisfies $|h_l(t;x) - h_{l-1}(t;x)|^2 \leq C \Delta_l^{2\beta}$ with a rate of $\beta=4$ 
	at observation times $t$ and parameter $x=(0.002,0.3,15)$. 
	Right: second moment of the time-averaged estimator  
	$\xi_l = \reallywidehat{[\pi_l-\pi_{l-1}](\varphi)}_{T,k,m}$ in \eqref{eq:time_ave_inc} 
	with $k=200$ and $m=2000$ against discretization level $l$ satisfies $\bbE [\xi_l^2] \leq C \Delta_l^{2\beta}$ with a rate of $\beta=4$.
	The function considered here is $\varphi(x) = \nabla_\theta \log \gamma_\theta(x)$ with $\theta=(1,1)$.  
	The second moment was estimated using $100$ independent realizations.}
	\label{fig:sir_conv}
\end{figure}

\begin{figure}[!htbp]
	\centering\includegraphics[width=0.49\textwidth]{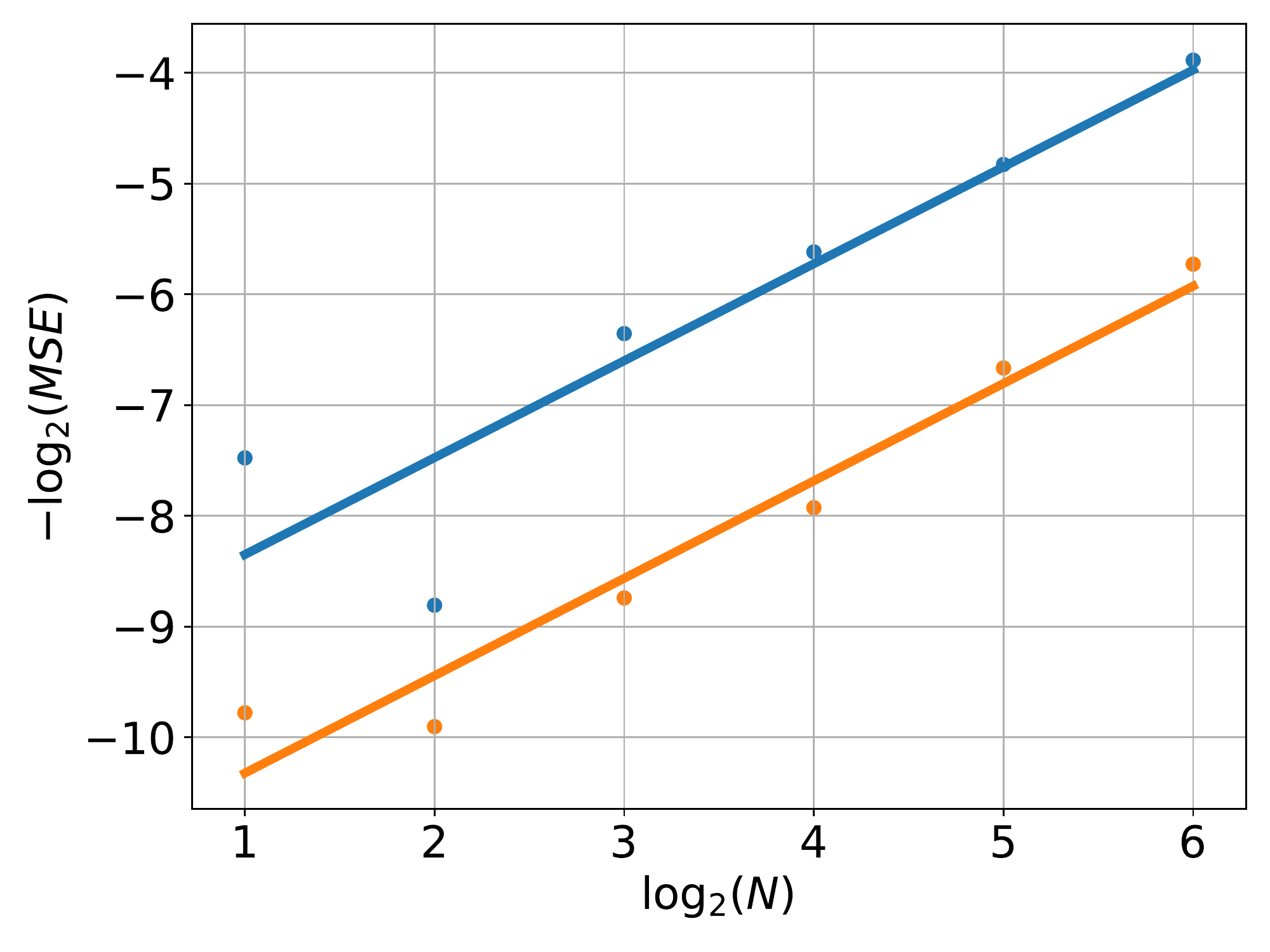}
	\includegraphics[width=0.49\textwidth]{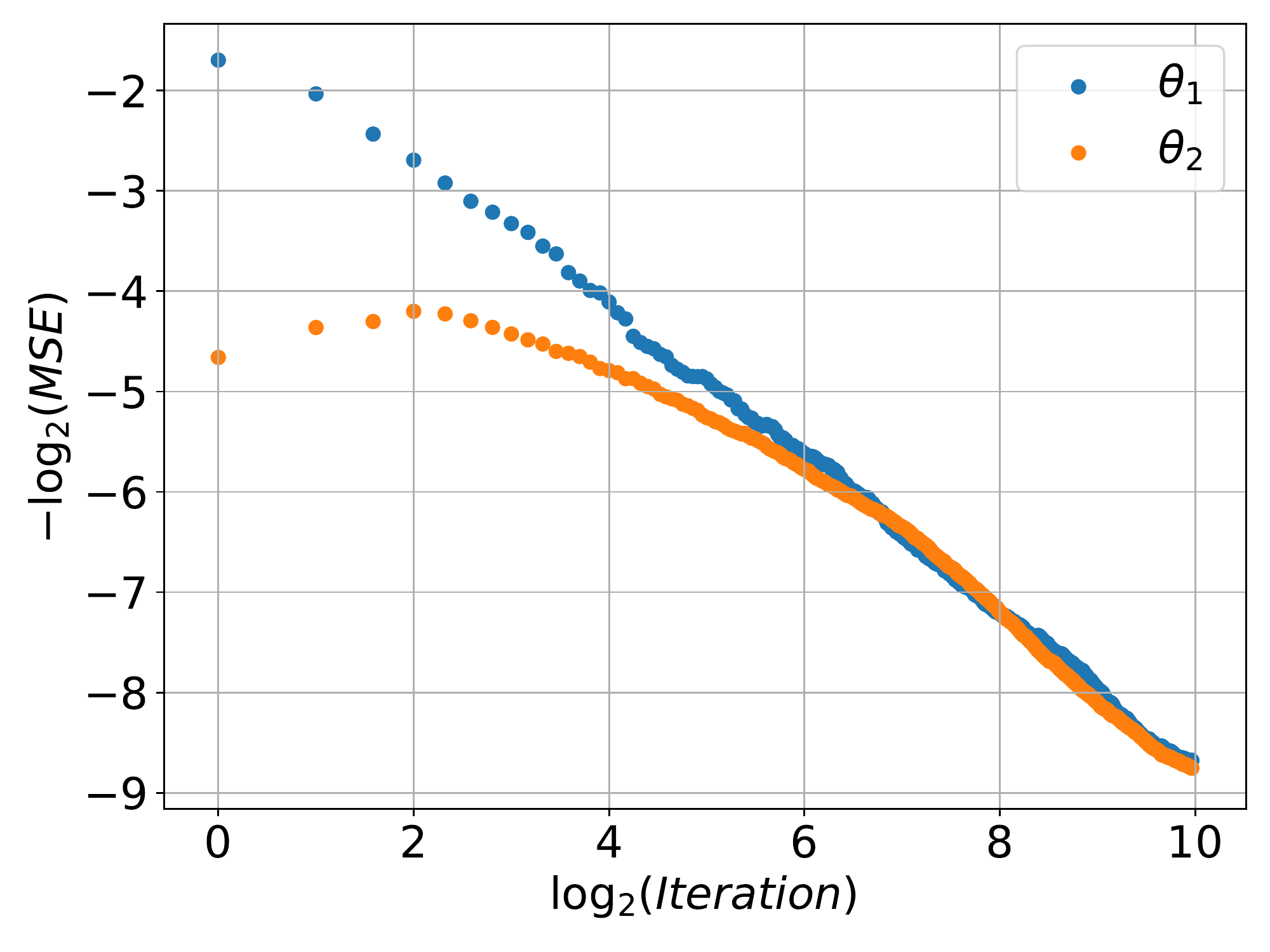}
	\caption{Compartmental model of Section \ref{sec:epidemic}. 
	Left: mean squared error of averaged estimator $\bbE[(\reallywidehat{\pi(\varphi)}(N) - \pi(\varphi))^2]$
	against number of single term replicates $N$ for the  
	functions $\varphi(x) = \partial_{\theta_1} \log \gamma_\theta(x)$ (\emph{blue}) and 
	$\varphi(x) = \partial_{\theta_2} \log \gamma_\theta(x)$ (\emph{orange}) with $\theta = (1,1)$. 
	The mean squared error was estimated using 
	an average of $10,000$ independent replicates with a discretization of $\Delta_l = 0.025 \times 2^{-l}$ 
	as ground truth for $\pi(\varphi)$, and $160$ independent realizations. 
	Right: convergence of stochastic gradient iterates $(\theta^{(i)})_{i\in\mathbb{Z}^+}$ 
	to the maximum likelihood estimator $\theta_{\rm MLE}$. 
	The mean squared error $\bbE[(\theta^{(i)}-\theta_{\rm MLE})^2]$ 
	was estimated using a longer run of the stochastic gradient algorithm as ground truth for $\theta_{\rm MLE}$, 
	and $100$ independent realizations. 
	The sequence of learning rates considered here is $\alpha_i=\alpha_1/i$ with $\alpha_1=0.01$. 
	}
	\label{fig:sir_mse}
\end{figure}

\begin{figure}[!htbp]
	\centering\includegraphics[width=0.49\textwidth]{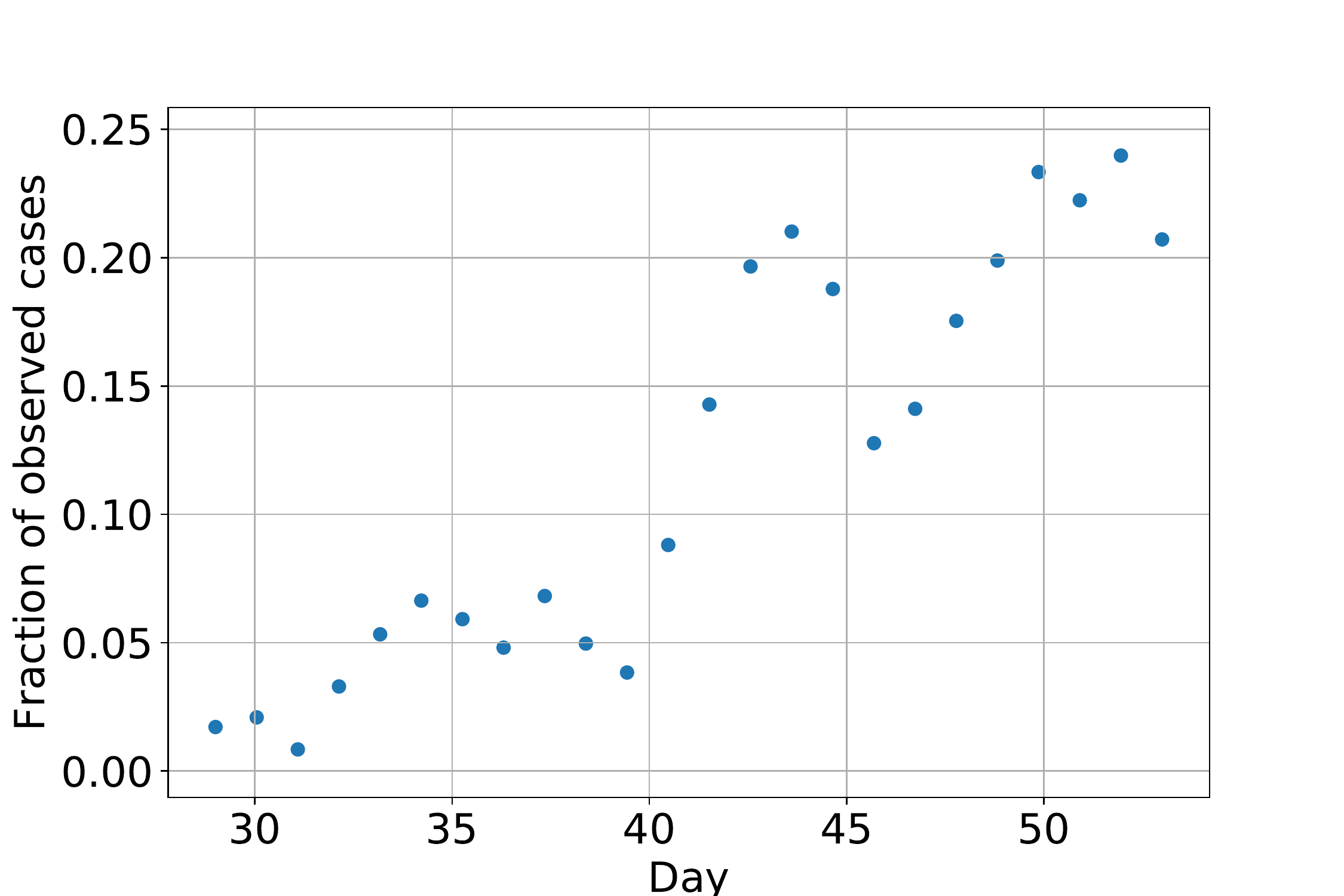}
	\caption{Compartmental model of Section \ref{sec:epidemic}. 
	Ratio of the total number of reported cases 
	$C_i = C_0 + \sum_{j=1}^i y_j$ ($C_0$ denotes the number of cases before February 12, 2020) 
	to the expected number of total infections $\bbE[  a \int_{-X_3}^{n+i}S(t; X)I(t; X) dt | y, \theta_{\rm MLE}]$ 
	on day $n+i$.}
	\label{fig:punchline}
\end{figure}

\subsubsection*{Acknowledgements}
JH was funded by CY Initiative of Excellence (grant ``Investissements
d'Avenir'' ANR-16-IDEX-0008). 
AJ was supported by KAUST baseline funding.
AT and KJHL were supported by The Alan Turing Institute under the 
EPSRC grant EP/N510129/1.

\appendix

\section{Proofs}\label{sec:proofs}

Throughout the appendix, $C$ is a finite constant that does not depend upon $l$ nor the time parameter of the Markov chain. The value
may change upon each appearance. 
For ease of notation only, we will
set $k=0$ from herein; the proofs, with some minor modifications, will hold for any $k\geq 0$. The appendix first gives the proof of Theorem \ref{theo:main_res} and then a collection of technical results which are used to achieve the proof.

\begin{proof}[Proof of Theorem \ref{theo:main_res}]
To prove the result we must verify for any $(l,\varphi)\in\mathbb{N}\times\mathcal{B}_b(\mathsf{X})\cap
\textrm{Lip}_{\tilde{\mathsf{d}}}(\mathsf{X})$
\begin{eqnarray}
\mathbb{E}[\reallywidehat{\pi_0(\varphi)}_{0}] & = & \pi_0(\varphi), \label{eq:theo1}\\
\mathbb{E}[\reallywidehat{[\pi_l-\pi_{l-1}](\varphi)}_0] & = & [\pi_l-\pi_{l-1}](\varphi),\label{eq:theo2}
\end{eqnarray}
and that
\begin{eqnarray}
\frac{1}{\mathbb{P}_L(0)}\mathbb{E}[\reallywidehat{\pi_0(\varphi)}_{0}^2] + 
\sum_{l\in\mathbb{N}} \frac{1}{\mathbb{P}_L(l)}\mathbb{E}[\{\reallywidehat{[\pi_l-\pi_{l-1}](\varphi)}_0\}^2] & < & +\infty.\label{eq:theo3}
\end{eqnarray}
Equations \eqref{eq:theo1} and \eqref{eq:theo2} are verified in Proposition \ref{prop:ub_inc}. For \eqref{eq:theo3}, $\mathbb{E}[\reallywidehat{\pi_0(\varphi)}_{0}^2]$ can be controlled by using Lemma \ref{lem:pi_0} and then by Lemma \ref{lem:sec_mom}, we have that
$$
\sum_{l\in\mathbb{N}} \frac{1}{\mathbb{P}_L(l)}\mathbb{E}[\{\reallywidehat{[\pi_l-\pi_{l-1}](\varphi)}_0\}^2] \leq C(\|\varphi\|_{\infty}\vee \|\varphi\|_{\textrm{Lip}})\sum_{l\in\mathbb{N}} \frac{1}{\mathbb{P}_L(l)}
\Delta_l^{2\big(\beta_1\wedge \beta_2 \big)}.
$$
The proof can now easily be completed.
\end{proof}

\begin{prop}\label{prop:ub_inc}
Assume (A\ref{ass:1}-\ref{ass:5}). Then we have for any $(l,\varphi)\in\mathbb{N}\times\mathcal{B}_b(\mathsf{X})$
\begin{eqnarray*}
\mathbb{E}[\reallywidehat{\pi_0(\varphi)}_{0}] & = & \pi_0(\varphi), \\
\mathbb{E}[\reallywidehat{[\pi_l-\pi_{l-1}](\varphi)}_0] & = & [\pi_l-\pi_{l-1}](\varphi).
\end{eqnarray*}
\end{prop}

\begin{proof}
The case of $\reallywidehat{\pi_0(\varphi)}_{0}$ is essentially that in \cite{glynn2}, so we focus on the second expectation.
We have the standard Martingale plus remainder decomposition:
$$
\reallywidehat{[\pi_l-\pi_{l-1}](\varphi)}_0 = \varphi(X_{0,l}) + \sum_{n=1}^{\check{\tau}_{l,l-1}-1}\{\varphi(X_{n,l})-\varphi(W_{n,l})\} - \Big(
\varphi(X_{0,l-1}) + \sum_{n=1}^{\check{\tau}_{l,l-1}-1}\{\varphi(X_{n,l-1})-\varphi(W_{n,l-1})\}\Big) 
$$
$$
= \varphi(X_{0,l}) + \sum_{n=1}^{\check{\tau}_{l,l-1}}\zeta_{l}(X_{n-1:n,l},W_{n-1:n,l}) + K_{l}(\reallywidehat{\varphi}_{l})(X_{0,l}) - K_{l}(\reallywidehat{\varphi}_{l})(W_{0,l}) -
$$
$$
\Big(\varphi(X_{0,l-1}) + \sum_{n=1}^{\check{\tau}_{l,l-1}}\zeta_{l-1}(X_{n-1:n,l-1},W_{n-1:n,l-1}) + K_{l-1}(\reallywidehat{\varphi}_{l-1})(X_{0,l-1}) - K_{l-1}(\reallywidehat{\varphi}_{l-1})(W_{0,l-1})\Big)
$$
where, for $s\in\{l,l-1\}$ and $(x_{n-1:n},w_{n-1:n})\in\mathsf{X}^2\times\mathsf{X}^2$
$$
\zeta_{s}(x_{n-1:n},w_{n-1:n}) := \reallywidehat{\varphi}_{s}(x_n) - K_{s}(\reallywidehat{\varphi}_{s})(x_{n-1}) -  
\Big(\reallywidehat{\varphi}_{s}(w_n) - K_{s}(\reallywidehat{\varphi}_{s})(w_{n-1})\Big)
$$
and 
\begin{equation}\label{eq:pois_eq}
\reallywidehat{\varphi}_{s}(x) = \sum_{n\in\mathbb{Z}^+}[K_{s}^n(\varphi)(x)-\pi_s(\varphi)]
\end{equation}
is well-defined for each $x\in\mathsf{X}$ and solves the Poisson equation, 
$$
\reallywidehat{\varphi}_{s}(x) - K_{s}(\reallywidehat{\varphi}_{s})(x) = \varphi(x) - \pi_s(\varphi).
$$
Set $(\mathcal{F}_n)_{n\in\mathbb{Z}^+}$ as the natural filtration generated by $(Z_{n,l,l-1})_{n\geq 0}$ and for $(n,s)\in\mathbb{N}\times\{l,l-1\}$
\begin{eqnarray*}
M_{n,s} & := &  \sum_{j=1}^{n}\zeta_{s}(X_{j-1:j,s},W_{j-1:j,s}) 
\end{eqnarray*}
with $M_{0,l}=M_{0,l-1}=0$. Then for each $s\in\{l,l-1\}$, $(M_{n,s},\mathcal{F}_n)_{n\in\mathbb{Z}^+}$ is a Martingale. 
Thus, we have
$$
\reallywidehat{[\pi_l-\pi_{l-1}](\varphi)}_0 
= \varphi(X_{0,l}) + M_{\check{\tau}_{l,l-1},l} + K_{l}(\reallywidehat{\varphi}_{l})(X_{0,l}) - K_{l}(\reallywidehat{\varphi}_{l})(W_{0,l}) -
$$
\begin{equation}\label{eq:mart_decomp}
\Big(\varphi(X_{0,l-1}) + M_{\check{\tau}_{l,l-1},l-1}+ K_{l-1}(\reallywidehat{\varphi}_{l-1})(X_{0,l-1}) - K_{l-1}(\reallywidehat{\varphi}_{l-1})(W_{0,l-1})\Big).
\end{equation}
Now taking expectations on both sides of the equation and applying the optional sampling theorem, we have
$$
\mathbb{E}[\reallywidehat{[\pi_l-\pi_{l-1}](\varphi)}_0] = 
\nu_l K_{l}(\varphi) + \nu_l K_{l}^2(\reallywidehat{\varphi}_{l}) - \nu _lK_{l}(\reallywidehat{\varphi}_{l}) -
\Big(\nu_{l-1} K_{l-1}(\varphi) + \nu_{l-1} K_{l-1}^2(\reallywidehat{\varphi}_{l-1} )- \nu_{l-1} K_{l-1}(\reallywidehat{\varphi}_{l-1})\Big).
$$
Hence, we have that
$$
\mathbb{E}[\reallywidehat{[\pi_l-\pi_{l-1}](\varphi)}_0]  =  [\pi_l-\pi_{l-1}](\varphi).
$$
\end{proof}

\begin{rem}
In the proof we have essentially derived the first Wald equality for Markov chains (see \cite{wald1}). The main point is to provide the Martingale plus remainder
decomposition \eqref{eq:mart_decomp}.
\end{rem}

The following two results are Lemmata A.1 and A.2.~in \cite{mlmcmc}. There is a slight addition, which can be deduced from the calculations in 
\cite[Lemma A.4.]{mlmcmc}.

\begin{lem}\label{lem:pois_lip}
Assume (A\ref{ass:1},\ref{ass:3}). Then there exists a $C\in(0,\infty)$, such that for any $(l,\varphi)\in\mathbb{Z}^+\times\mathcal{B}_b(\mathsf{X})\cap\textrm{\emph{Lip}}_{\tilde{\mathsf{d}}}(\mathsf{X})$
we have:
$$
|\reallywidehat{\varphi}_l(x)-\reallywidehat{\varphi}_l(w)|\vee 
|K_l(\reallywidehat{\varphi}_l)(x)-K_l(\reallywidehat{\varphi}_l)(w)|
\leq C(\|\varphi\|_{\infty}\vee \|\varphi\|_{\textrm{\emph{Lip}}}) \tilde{\mathsf{d}}(x,w).
$$
\end{lem}

\begin{lem}\label{lem:poisson_eq_cont}
Assume (A\ref{ass:1},\ref{ass:2}). Then there exists a $C\in(0,\infty)$, such that for any $(l,\varphi)\in\mathbb{N}\times\mathcal{B}_b(\mathsf{X})$ we have:
$$
\sup_{x\in\mathsf{X}} |\reallywidehat{\varphi}_l(x)-\reallywidehat{\varphi}_{l-1}(x)| \vee
\sup_{x\in\mathsf{X}} |K_l(\reallywidehat{\varphi}_l)(x)-K_{l-1}(\reallywidehat{\varphi}_{l-1})(x)|
\leq C\|\varphi\|_{\infty} \Delta_l^{\beta_1},
$$
where $\beta_1$ is as in (A\ref{ass:2}).
\end{lem}

\begin{lem}\label{lem:good_prob}
Assume (A\ref{ass:4}). Then there exists a $C\in(0,\infty)$, such that for any $(l,n)\in\mathbb{N}\times\mathbb{Z}^+$ we have:
\begin{eqnarray*}
\mathbb{E}[\mathbb{I}_{B(C,\Delta_l^{\beta_2},\tilde{\mathsf{d}})^c}(Z_{n,l,l-1})] 
& \leq & C(n+1)\Delta_l^{\beta_2(2+\epsilon)}, \\
\mathbb{E}[\tilde{\mathsf{d}}(X_{n,l},X_{n,l-1})^{2+\epsilon}]\vee
\mathbb{E}[\tilde{\mathsf{d}}(W_{n,l},W_{n,l-1})^{2+\epsilon}]
& \leq & C(n+1)\Delta_l^{\beta_2(2+\epsilon)},
\end{eqnarray*}
where $\beta_{2}$ and $\epsilon$ are as in (A\ref{ass:4}).
\end{lem}

\begin{proof}
The proof of the first statement is by induction. The first statement, which holds at step zero by assumption, so assuming the result for $n-1$:
$$
\mathbb{E}[\mathbb{I}_{B(C,\Delta_l^{\beta_2},\tilde{\mathsf{d}})^c}(Z_{n,l,l-1})] = 
$$
$$
\mathbb{E}[\mathbb{I}_{B(C,\Delta_l^{\beta_2},\tilde{\mathsf{d}})^c\times B(C,\Delta_l^{\beta_2},\tilde{\mathsf{d}})}(Z_{n,l,l-1},Z_{n-1,l,l-1})] +
\mathbb{E}[\mathbb{I}_{B(C,\Delta_l^{\beta_2},\tilde{\mathsf{d}})^c\times B(C,\Delta_l^{\beta_2},\tilde{\mathsf{d}})^c}(Z_{n,l,l-1},Z_{n-1,l,l-1})].
$$
Then applying (A\ref{ass:4}) along with the induction hypothesis one can conclude that:
$$
\mathbb{E}[\mathbb{I}_{B(C,\Delta_l^{\beta_2},\tilde{\mathsf{d}})^c}(Z_{n,l,l-1})]
\leq  C(n+1)\Delta_l^{\beta_2(2+\epsilon)}
$$
and hence the proof of the first statement is complete.

For the second statement, we consider only $\tilde{\mathsf{d}}(X_{n,l},X_{n,l-1})^{2+\epsilon}$ as the argument is the same for $\tilde{\mathsf{d}}(W_{n,l},W_{n,l-1})^{2+\epsilon}$. We have 
$$
\mathbb{E}[\tilde{\mathsf{d}}(X_{n,l},X_{n,l-1})^{2+\epsilon}] = \mathbb{E}[\tilde{\mathsf{d}}(X_{n,l},X_{n,l-1})^{2+\epsilon}\{
\mathbb{I}_{B(C,\Delta_l^{\beta_2},\tilde{\mathsf{d}})}(Z_{n,l,l-1}) + \mathbb{I}_{B(C,\Delta_l^{\beta_2},\tilde{\mathsf{d}})^c}(Z_{n,l,l-1})
\}].
$$
On $B(C,\Delta_l^{\beta_2},\tilde{\mathsf{d}})$, one has $\tilde{\mathsf{d}}(X_{n,l},X_{n,l-1})^{2+\epsilon}\leq C\Delta_l^{\beta_2(2+\epsilon)}$.
For the second
term on the R.H.S., $\mathsf{X}$ is compact and $\tilde{\mathsf{d}}(X_{n,l},X_{n,l-1})^{2+\epsilon}$ is bounded, 
so one can use the first part of the statement to deduce that
$$
\mathbb{E}[\tilde{\mathsf{d}}(X_{n,l},X_{n,l-1})^{2+\epsilon}] \leq C(n+1)\Delta_l^{\beta_2(2+\epsilon)}.
$$
\end{proof}

\begin{lem}\label{lem:martingale}
Assume (A\ref{ass:1}-\ref{ass:4}). Then there exists a $C\in(0,\infty)$, such that for any $(l,\varphi)\in\mathbb{N}\times\mathcal{B}_b(\mathsf{X})\cap\textrm{\emph{Lip}}_{\tilde{\mathsf{d}}}(\mathsf{X})$ we have:
$$
\mathbb{E}[\{M_{\check{\tau}_{l,l-1},l}-M_{\check{\tau}_{l,l-1},l-1}\}^2] \leq C(\|\varphi\|_{\infty}\vee \|\varphi\|_{\textrm{\emph{Lip}}})
\Delta_l^{2\big(\beta_1\wedge \beta_2\big)},
$$
where $\tilde{\mathsf{d}}$ is as (A\ref{ass:3}), $\beta_1$ is as (A\ref{ass:2}) and  $\beta_{2}$ 
is as in (A\ref{ass:4}).
\end{lem}

\begin{proof}
Set, for $n\in\mathbb{Z}^+$
$$
\check{M}_{n,l,l-1} := M_{n,l}-M_{n,l-1}
$$
then $(\check{M}_{n,l,l-1},\mathcal{F}_n)_{n\in\mathbb{Z}^+}$ is a Martingale and moreover as $\check{\tau}_{l,l-1}$ is a $\mathcal{F}_n-$stopping time,
so is 
$$
(\check{M}_{\check{\tau}_{l,l-1}\wedge n,l,l-1},\mathcal{F}_n)_{n\in\mathbb{Z}^+}.
$$ 
Then, by the Burkholder-Gundy-Davis inequality, we have
$$
\mathbb{E}\left[\check{M}_{\check{\tau}_{l,l-1}\wedge n,l,l-1}^2\right] \leq \mathbb{E}\left[\sum_{j=1}^{\check{\tau}_{l,l-1}\wedge n}\{\zeta_{l}(X_{j-1:j,l},W_{j-1:j,l}) -\zeta_{l-1}(X_{j-1:j,l-1},W_{j-1:j,l-1}) \}^2\right].
$$
By (A\ref{ass:5}), $\check{\tau}_{l,l-1}$ is almost surely finite, so by the monotone convergence theorem:
\begin{equation}\label{eq:mart_dom}
\mathbb{E}\left[\check{M}_{\check{\tau}_{l,l-1},l,l-1}^2\right] \leq \mathbb{E}\left[\sum_{j=1}^{\check{\tau}_{l,l-1}}\{\zeta_{l}(X_{j-1:j,l},W_{j-1:j,l}) -\zeta_{l-1}(X_{j-1:j,l-1},W_{j-1:j,l-1}) \}^2\right].
\end{equation}
Then we can upper-bound the R.H.S.~to yield
$$
\mathbb{E}\left[\check{M}_{\check{\tau}_{l,l-1},l,l-1}^2\right] \leq C(T_1 + T_2)
$$
where
\begin{eqnarray}
T_1 & := & \mathbb{E}\Bigg[\sum_{n=1}^{\check{\tau}_{l,l-1}}\Big\{\reallywidehat{\varphi}_l(X_{n,l})-\reallywidehat{\varphi}_l(W_{n,l})- \reallywidehat{\varphi}_{l-1}(X_{n,l-1})+ \reallywidehat{\varphi}_{l-1}(W_{n,l-1})\Big\}^2\Bigg]\label{eq:mg_prf1}\\
T_2 & := & \mathbb{E}\Bigg[\sum_{n=1}^{\check{\tau}_{l,l-1}}
\Big\{K_l(\reallywidehat{\varphi}_l)(X_{n-1,l})-K_l(\reallywidehat{\varphi}_l)(W_{n-1,l})- K_{l-1}(\reallywidehat{\varphi}_{l-1})(X_{n-1,l-1})+ \nonumber \\ & & K_{l-1}(\reallywidehat{\varphi}_{l-1})(W_{n-1,l-1})\Big\}^2
\Bigg].\label{eq:mg_prf2}
\end{eqnarray}
To conclude the proof, one must appropriately upper-bound \eqref{eq:mg_prf1} and \eqref{eq:mg_prf2}. As the forms are very similar and due to the expression for $\reallywidehat{\varphi}_s$ (see \eqref{eq:pois_eq}), $s\in\{l,l-1\}$, we will give the proof for \eqref{eq:mg_prf1} only, as the proof for \eqref{eq:mg_prf2} is almost the same.

Define, for $z\in\mathsf{X}^4$
$$
\reallywidehat{\varphi}_{l,l-1}(z) := \sum_{n\in\mathbb{Z}^+}\check{K}_{l,l-1}^n\big(\reallywidehat{\psi}_{l,l-1}\big)(z)
$$
where $\reallywidehat{\psi}_{l,l-1}(x_l,w_{l},x_{l-1}, w_{l-1})=(\reallywidehat{\varphi}_l(x_{l})-\reallywidehat{\varphi}_l(w_{l})- \reallywidehat{\varphi}_{l-1}(x_{l-1})+ \reallywidehat{\varphi}_{l-1}(w_{l-1}))^2$. We note that by (A\ref{ass:5}), one can easily verify that $\reallywidehat{\varphi}_{l,l-1}$ is a well-defined function. Now, by the first Wald
equality for Markov chains, one has
$$
\mathbb{E}[T_1] =  \mathbb{E}[\check{K}_{l,l-1}(\reallywidehat{\varphi}_{l,l-1})(Z_{0,l,l-1})].
$$
This is because 
\begin{eqnarray*}
\mathbb{E}[\check{K}_{l,l-1}(\reallywidehat{\varphi}_{l,l-1})(Z_{\check{\tau}_{l,l-1},l,l-1})] & = & 0,\\
\check{\pi}_{l,l-1}(\reallywidehat{\varphi}_{l,l-1}) & = & 0,\\
\mathbb{E}[\check{\tau}_{l,l-1}] & < & +\infty,
\end{eqnarray*}
where $\check{\pi}_{l,l-1}$ is the invariant measure of $\check{K}_{l,l-1}$ (marginally, the invariant distribution of $(X_s,W_s)$ is $\pi_s(dx_s)\delta_{x_s}(dw_s)$ for $s\in\{l,l-1\}$). Now
\begin{equation} \label{eq:mg_prf3}
T_1 =  \mathbb{E}[\reallywidehat{\psi}_{l,l-1}(Z_{0,l,l-1})] +
 \mathbb{E}\Bigg[\sum_{n\in\mathbb{N}}\mathbb{E}[\reallywidehat{\psi}_{l,l-1}(Z_{n,l,l-1})|Z_{0,l,l-1}]\Bigg]. 
\end{equation}
We have
\begin{eqnarray}
 \mathbb{E}[\reallywidehat{\psi}_{l,l-1}(Z_{0,l,l-1})] & = &   \mathbb{E}\Big[\Big(
 \reallywidehat{\varphi}_l(X_{0,l})-\reallywidehat{\varphi}_l(W_{0,l})- \reallywidehat{\varphi}_{l-1}(X_{0,l-1})+ \reallywidehat{\varphi}_{l-1}(W_{0,l-1})
 \Big)^2\Big] \nonumber\\
 & \leq & C\Big\{\mathbb{E}\Big[\Big( \reallywidehat{\varphi}_l(X_{0,l}) - \reallywidehat{\varphi}_{l-1}(X_{0,l}) + \reallywidehat{\varphi}_{l-1}(X_{0,l}) - \reallywidehat{\varphi}_{l-1}(X_{0,l-1})\Big)^2\Big] + \nonumber\\
 & & \mathbb{E}\Big[\Big( \reallywidehat{\varphi}_l(W_{0,l}) - \reallywidehat{\varphi}_{l-1}(W_{0,l}) + \reallywidehat{\varphi}_{l-1}(W_{0,l}) - \reallywidehat{\varphi}_{l-1}(W_{0,l-1})\Big)^2\Big]\Big\} \nonumber\\
  & \leq & C(\|\varphi\|_{\infty}\vee \|\varphi\|_{\textrm{Lip}})\Big(\Delta_l^{2\beta_1} + 
  \mathbb{E}[\tilde{\mathsf{d}}(X_{0,l},X_{0,l-1})^{2}] + \Delta_l^{2\beta_1} +
\mathbb{E}[\tilde{\mathsf{d}}(W_{0,l},W_{0,l-1})^{2}]\Big) \nonumber\\
  & \leq & C(\|\varphi\|_{\infty}\vee \|\varphi\|_{\textrm{Lip}})\Delta_l^{ 2( \beta_1\wedge \beta_2)} 
  \label{eq:mg_prf4}
\end{eqnarray}
where we have applied the $C_2-$inequality and  Lemmata \ref{lem:pois_lip}-\ref{lem:poisson_eq_cont} to go to the third line and Lemma \ref{lem:good_prob} to go to the last line.
Then, we also have
\begin{eqnarray*}
\mathbb{E}\Bigg[\sum_{n\in\mathbb{N}}\mathbb{E}[\reallywidehat{\psi}_{l,l-1}(Z_{n,l,l-1})|Z_{0,l,l-1}]\Bigg]  & = & 
\mathbb{E}\Bigg[
\sum_{n\in\mathbb{N}}\mathbb{E}[\mathbb{I}_{\{\check{\tau}_{l,l-1}>n\}}\reallywidehat{\psi}_{l,l-1}(Z_{n,l,l-1})|Z_{0,l,l-1}]\Bigg] \\
& = & \sum_{n\in\mathbb{N}}\mathbb{E}[\mathbb{I}_{\{\check{\tau}_{l,l-1}>n\}}\reallywidehat{\psi}_{l,l-1}(Z_{n,l,l-1})]\\
& \leq & \sum_{n\in\mathbb{N}}\mathbb{E}[\mathbb{I}_{\{\check{\tau}_{l,l-1}>n\}}]^{\epsilon/(2+\epsilon)}
\mathbb{E}[\reallywidehat{\psi}_{l,l-1}(Z_{n,l,l-1})^{\frac{2+\epsilon}{2}}]^{2/(2+\epsilon)}\\
& \leq & C\sum_{n\in\mathbb{N}}(\rho^{\frac{\epsilon}{2+\epsilon}})^n\mathbb{E}[\reallywidehat{\psi}_{l,l-1}(Z_{n,l,l-1})^{\frac{2+\epsilon}{2}}]^{2/(2+\epsilon)}
\end{eqnarray*}
where we have used H\"older's inequality to go to the third line and (A\ref{ass:5}) to go to the fourth line.
Now, by 
similar calculations that lead to \eqref{eq:mg_prf4}, we have that
$$
\mathbb{E}[\reallywidehat{\psi}_{l,l-1}(Z_{n,l,l-1})^{\frac{2+\epsilon}{2}}]^{2/(2+\epsilon)} 
\leq  C(\|\varphi\|_{\infty}\vee \|\varphi\|_{\textrm{Lip}})(n+1)\Delta_l^{2\big(\beta_1\wedge \beta_2 
\big)}
$$
and hence that
\begin{equation}  \label{eq:mg_prf5}
\mathbb{E}\Bigg[\sum_{n\in\mathbb{N}}\mathbb{E}[\reallywidehat{\psi}_{l,l-1}(Z_{n,l,l-1})|Z_{0,l,l-1}]\Bigg] \leq 
C(\|\varphi\|_{\infty}\vee \|\varphi\|_{\textrm{Lip}})\Delta_l^{2\big(\beta_1\wedge  \beta_2 
\big)} \, .
\end{equation}
Then combining \eqref{eq:mg_prf4}-\eqref{eq:mg_prf5} with \eqref{eq:mg_prf3} yields
$$
T_1 \leq C(\|\varphi\|_{\infty}\vee \|\varphi\|_{\textrm{Lip}})\Delta_l^{2\big(\beta_1\wedge  \beta_2 
\big)} \, ,
$$
and hence the proof is concluded.
\end{proof}

\begin{lem}\label{lem:sec_mom}
Assume (A\ref{ass:1}-\ref{ass:4}). Then there exists a $C\in(0,\infty)$, such that for any $(l,\varphi)\in\mathbb{N}\times\mathcal{B}_b(\mathsf{X})\cap\textrm{\emph{Lip}}_{\tilde{\mathsf{d}}}(\mathsf{X})$ we have:
$$
\mathbb{E}\Big[\reallywidehat{[\pi_l-\pi_{l-1}](\varphi)}_0^2\Big] \leq C(\|\varphi\|_{\infty}\vee \|\varphi\|_{\textrm{\emph{Lip}}})
\Delta_l^{2\big(\beta_1\wedge  \beta_2 
\big)}
$$
where $\tilde{\mathsf{d}}$ is as (A\ref{ass:3}), $\beta_1$ is as (A\ref{ass:2}) and  $\beta_{2},\epsilon$ are as in (A\ref{ass:4}).
\end{lem}

\begin{proof}
Using the decomposition in \eqref{eq:mart_decomp} along with the $C_2-$inequality, we have the upper-bound
\begin{eqnarray*}
\mathbb{E}\Big[\reallywidehat{[\pi_l-\pi_{l-1}](\varphi)}_0^2\Big] & \leq &  C\Big(
\mathbb{E}[\{\varphi(X_{0,l})-\varphi(X_{0,l-1})\}^2] +
\mathbb{E}[\{M_{\check{\tau}_{l,l-1},l}-M_{\check{\tau}_{l,l-1},l-1}\}^2] + \\& &
+ 
\mathbb{E}\Big[\Big(
 K_l(\reallywidehat{\varphi}_l)(X_{0,l})-K_l(\reallywidehat{\varphi}_l)(W_{0,l})- K_{l-1}(\reallywidehat{\varphi}_{l-1})(X_{0,l-1})+ K_{l-1}(\reallywidehat{\varphi}_{l-1})(W_{0,l-1})
 \Big)^2\Big]
\Big).
\end{eqnarray*}
The first term on the R.H.S.~can be treated by using $\varphi\in\textrm{Lip}_{\tilde{\mathsf{d}}}(\mathsf{X})$ and Lemma \ref{lem:good_prob}. For the second term on the R.H.S.~one
can use Lemma \ref{lem:martingale}. For the last term on the R.H.S.~one can use very similar calculations to those used to derive \eqref{eq:mg_prf4}. The proof is thus completed.
\end{proof}

\begin{lem}\label{lem:pi_0}
Assume (A\ref{ass:1}-\ref{ass:5}). Then there exists a $C\in(0,\infty)$, such that for any $\varphi\in\mathcal{B}_b(\mathsf{X})$ we have:
$$
\mathbb{E}[\reallywidehat{\pi_0(\varphi)}_{0}^2] \leq C\|\varphi\|_{\infty}.
$$
\end{lem}

\begin{proof}
We have
$$
\reallywidehat{\pi_0(\varphi)}_{0} = \varphi(X_{0,0}) + \check{M}_{\tau_0,0} + \check{K}_0(\reallywidehat{\varphi}_0)(X_{0,0}) - \check{K}_0(\reallywidehat{\varphi}_0)(W_{0,0}),
$$
where $\reallywidehat{\varphi}_0(x)=\sum_{n\in\mathbb{Z}^+}[K_0^n-\pi_0](\varphi)$, and for any $n\in\mathbb{N}$
$$
\check{M}_{n,0} = \sum_{j=1}^n \{\reallywidehat{\varphi}_0(X_{j,0}) - 
\check{K}_0(\reallywidehat{\varphi}_0)(X_{j-1,0}) 
-\reallywidehat{\varphi}_0(W_{j,0}) + 
\check{K}_0(\reallywidehat{\varphi}_0)(W_{j-1,0})
\}
$$
and $\check{M}_{0,0}:=0$. One has, by using the $C_2-$inequality and the fact that $\varphi\in\mathcal{B}_b(\mathsf{X})$ along with (A\ref{ass:1})
$$
\mathbb{E}[\reallywidehat{\pi_0(\varphi)}_{0}^2] \leq C\|\varphi\|_{\infty} + \mathbb{E}[\check{M}_{\tau_0,0}^2].
$$
Then, by using a similar argument to derive \eqref{eq:mart_dom}, we have
$$
\mathbb{E}[\check{M}_{\tau_0,0}^2] \leq \mathbb{E}\Bigg[\sum_{j=1}^{\tau_0} \{\reallywidehat{\varphi}_0(X_{j,0}) - 
\check{K}_0(\reallywidehat{\varphi}_0)(X_{j-1,0}) 
-\reallywidehat{\varphi}_0(W_{j,0}) + 
\check{K}_0(\reallywidehat{\varphi}_0)(W_{j-1,0})\}^2\Bigg].
$$
The proof is now completed as the summands on the R.H.S.~are upper-bounded by $C\|\varphi\|_{\infty}$ and the expectation of the stopping time is finite via (A\ref{ass:5}).
\end{proof}

%
%
%
%
%
%


\begin{thebibliography}{99}

\bibitem{sergios}
{\sc Agapiou}, S., {\sc Roberts}, G. O.~\& {\sc Vollmer}, S.~(2018). Unbiased Monte Carlo: 
Posterior estimation for intractable/infinite-dimensional models.
\emph{Bernoulli}, {\bf 24}, 1726--1786.


\bibitem{beskos2}
{\sc Beskos}, A., 
{\sc Jasra}, A., 
{\sc Law}, K. J. H., 
{\sc Marzouk}, Y., \& 
{\sc Zhou}, Y.~ (2018). 
{Multilevel sequential Monte Carlo with dimension-independent likelihood-informed proposals}.
{\em SIAM/ASA J. Uncer. Quant.},
{\bf 6}, {762--786}.




\bibitem{beskos}
{\sc Beskos}, A., 
{\sc Jasra}, A., 
{\sc Law}, K. J. H., 
{\sc Tempone}, R., \& 
{\sc Zhou}, Y.~ (2017). 
Multilevel sequential Monte Carlo samplers. {\em Stoch. Proc. Appl.}, {\bf 127}, 1417--1440.

\bibitem{bourabee_eberle_zimmer}
{\sc Bou-Rabee}, N., {\sc Eberle}, A., \& {\sc Zimmer}, R.~(2020).
Coupling and convergence for Hamiltonian Monte Carlo. \emph{Ann. Appl. Probab.},
{\bf 30}, 1209--1250.

\bibitem{brenner}
{\sc Brenner}, S. \& 
{\sc Scott}, R. (2007). 
\textit{The Mathematical Theory of Finite Element Methods}. Springer: New York.


\bibitem{ciarlet}
{\sc Ciarlet}, P. G. (2002). 
\textit{The Finite Element Method for Elliptic Problems}.  SIAM: Philadelphia.


\bibitem{pcn}
{\sc Cotter}, S.~L., {\sc Roberts}, G.~O., 
{\sc Stuart}, A.~M. \& 
{\sc White}, D.~(2013). MCMC methods for functions: modifying old algorithms to make them faster. 
\emph{Stat. Sci.}, {\bf 28}, 424--446.

\bibitem{duane}
{\sc Duane}, S., {\sc Kennedy}, A.~D., 
{\sc Pendleton}, B.~J., \& 
{\sc Roweth}, D.~(1987). Hybrid Monte Carlo. 
\emph{Phy. Lett. B.}, {\bf 28}, 216--222.

\bibitem{ub_vihola}
{\sc Franks}, J., 
{\sc Jasra}, A., 
{\sc Law}, K.~J.~H., {\sc Chada}, N. \& 
{\sc Vihola}, M.~(2018). 
Unbiased inference for discretely observed hidden Markov model diffusions. 
arXiv preprint.

\bibitem{glynn2}
{\sc Glynn}, P.~W. \& 
{\sc Rhee}, C.~H.~(2014). 
Exact estimation for Markov chain equilibrium expectations. 
\emph{J. Appl. Probab.}, {\bf 51}, 
377--389.

\bibitem{gower}
{\sc Gower}, R.~M., 
{\sc Loizou}, N., 
{\sc Qian}, X., 
{\sc Sailanbayev}, A., 
{\sc Shulgin}, E., \& 
{\sc Richtarik}, P. (2019). 
SGD: General analysis and improved rates. 
\textit{Proceedings of the 36th International Conference on Machine Learning, in PMLR} 
{\bf 97}, 5200--5209.




\bibitem{heng_jacob_2019}
{\sc Heng}, J. \& {\sc Jacob}, P.~(2019). Unbiased Hamiltonian
Monte Carlo with couplings. \emph{Biometrika}, {\bf 106}, 287--302.


\bibitem{ub_grad}
{\sc Heng}, J., {\sc Houssineau}, J. \& {\sc Jasra}, A.~(2021). On unbiased score estimation for partially observed diffusions. Work in progress.

\bibitem{jacob1}
{\sc Jacob}, P., {\sc O' Leary}, J. \& {\sc Atchad{\'e}}, Y.~(2020). Unbiased Markov chain Monte Carlo with couplings (with discussion).
\emph{J.~R.~Statist. Soc. Ser. B}, {\bf 82}, 543--600.

\bibitem{jacob2}
{\sc Jacob}, P., {\sc Lindsten}, F. \& {\sc Sch\"on}, T.~(2020). Smoothing with couplings of conditional particle filters.
\emph{J. Amer. Statist. Assoc.} {\bf 115}, 721--729.

\bibitem{rssb_disc}
{\sc Jasra}, A., {\sc Heng}, J. \& {\sc Law}, K. J. H.~(2020). Discussion of Jacob et al.
\emph{J.~R.~Statist. Soc. Ser. B}, {\bf 82}, 586--587.

\bibitem{ub_pf}
{\sc Jasra}, A., {\sc Law}, K. J. H. \& {\sc Yu}, F.~(2020). Unbiased filtering of a class of partially observed diffusions.
arXiv preprint.


\bibitem{mlpf}
{\sc Jasra}, A., {\sc Kamatani}, K., {\sc Law} K. J. H. \& {\sc Zhou}, Y.~(2017). 
Multilevel particle filters. \emph{SIAM J. Numer. Anal.}, {\bf 55}, 3068--3096.

\bibitem{mlmcmc}
{\sc Jasra}, A., {\sc Law}, K. J. H. \& {\sc Xu}, Y.~(2021).
Markov chain Simulation for Multilevel Monte Carlo. \emph{Found. Data Sci.} (to appear).

\bibitem{ub_bip}
{\sc Jasra}, A., {\sc Law}, K. J. H. \& {\sc Lu}, D.~(2021).
Unbiased estimation of the gradient of the log-likelihood in inverse problems. \emph{Stat. Comp.} (to appear).

\bibitem{johnson_1996}
{\sc Johnson}, V.~(1996). Studying convergence of Markov
chain Monte Carlo algorithms using coupled sample paths. 
\emph{J. Amer. Statist. Assoc.}, {\bf 91}, 154--166.

\bibitem{lindvall_rogers_1986}
{\sc Lindvall}, T. \& 
{\sc Rogers}, L.~(1996). Coupling of multidimensional diffusions by reflection. 
\emph{Ann. Appl. Probab.}, {\bf 14}, 860--872.

\bibitem{kushner}
{\sc Kushner}, H. \& 
{\sc Yin}, G.~G. (2003). 
\textit{Stochastic Approximation and Recursive Algorithms and Applications}. 
Springer: New York.

\bibitem{sirx}
{\sc Maier}, B.~F. \& 
{\sc Brockmann}, D.~(2020). 
Effective containment explains sub-exponential growth in recent confirmed COVID-19 cases in China. 
\emph{Science}, 
{\bf 368}, 742--746.

\bibitem{mcl}
{\sc McLeish}, D.~(2011). A general method for debiasing a Monte Carlo estimator. \emph{Monte Carlo Meth. Appl.}, {\bf 17}, 301--315.

\bibitem{wald1}
{\sc Moustakides}, G.~(1999). An extension of Wald's first lemma for Markov processes. \emph{J. Appl. Probab.}, {\bf 36}, 48--59.

\bibitem{neal}
{\sc Neal}, R.~M.~(1998). Regression and classification using Gaussian process priors.
In \emph{Bayesian statistics, 6} (Bernardo et al. eds), 475--501, Oxford: OUP.

\bibitem{rhee}
{\sc Rhee}, C. H. \& {\sc Glynn}, P.~(2015). Unbiased estimation with square root convergence for SDE models. \emph{Op. Res.},~{\bf 63}, 1026--1043. 

\bibitem{stuart}
{\sc Stuart}, A. M. (2010). 
Inverse problems: A Bayesian perspective. 
\textit{Acta Numerica}, {\bf 19}, 451--559.

\bibitem{rk}
{\sc S\"uli}, E. \& 
{\sc Mayers}, D.~F.~(2003). 
\emph{An Introduction to Numerical Analysis}. Cambridge: CUP.

\bibitem{tarantola}
{\sc Tarantola}, A.~(2005). 
\emph{Inverse problem theory and methods for model parameter estimation}. 
Society for Industrial and Applied Mathematics.

\bibitem{thor}
{\sc Thorisson}, H. (2000). \emph{Coupling, Stationarity and Regeneration}. Springer: New York.

\bibitem{vihola}
{\sc Vihola}, M.~(2018). Unbiased estimators and multilevel Monte Carlo. \emph{Op. Res.}, {\bf 66}, 448--462.

\end{thebibliography}
\end{document}